\documentclass[11pt,a4paper]{article}
\usepackage[T1]{fontenc}
\usepackage{lmodern,microtype}
\usepackage[margin=27mm]{geometry}
\usepackage{amsmath,amssymb,amsthm,mathtools,booktabs,array,longtable}
\usepackage{graphicx}
\usepackage{enumitem}
\usepackage[skip=4pt plus 1pt,indent=0pt]{parskip}
\setlist{topsep=4pt,itemsep=2pt,parsep=0pt,partopsep=0pt}
\usepackage[dvipsnames]{xcolor}
\usepackage{tikz}
\usetikzlibrary{arrows.meta}
\usepackage{tcolorbox}
\usepackage{needspace}
\usepackage[section]{placeins}
\usepackage{float}
\usepackage[colorlinks=true,linkcolor=MidnightBlue,citecolor=MidnightBlue,urlcolor=MidnightBlue]{hyperref}
\numberwithin{equation}{section}
\numberwithin{figure}{section}
\numberwithin{table}{section}
\numberwithin{footnote}{section}
\newtheorem{theorem}{Theorem}[section]
\newtheorem{lemma}[theorem]{Lemma}
\newtheorem{corollary}[theorem]{Corollary}
\AddToHook{env/theorem/before}{\Needspace{4\baselineskip}}
\AddToHook{env/lemma/before}{\Needspace{4\baselineskip}}
\AddToHook{env/corollary/before}{\Needspace{4\baselineskip}}
\AddToHook{env/example/before}{\Needspace{4\baselineskip}}
\AddToHook{env/question/before}{\Needspace{4\baselineskip}}
\newtheoremstyle{informal}
  {3pt}{3pt}{\normalfont}{}
  {\normalfont\small\scshape\color{MidnightBlue}}{}{\newline}
  {\thmname{#1}\thmnote{\enspace\textnormal{(#3)}}}
\theoremstyle{informal}\newtheorem*{informaltheorem}{Informal theorem}
\theoremstyle{definition}\newtheorem{example}{Example}[section]
\newtheorem*{question}{Question}
\theoremstyle{remark}\newtheorem{remark}[theorem]{Remark}
\newcommand{\Z}{\mathbb Z}
\newcommand{\C}{\mathbb C}
\newcommand{\E}{\mathbb E}
\newcommand{\one}{\mathbf 1}
\newcommand{\per}{\operatorname{per}}
\newcommand{\chq}{\chi_{\mathrm q}}
\newcommand{\chs}{\chi^{\star}}
\newcommand{\Om}{\Omega_n^{(\Z_q)}}
\newcommand{\mult}[2]{\binom{#1}{#2}}

\newcommand{\tr}{\operatorname{tr}}
\newcommand{\HS}{\mathrm{HS}}

\title{Optimal entanglement-assisted source coding\\
under a balanced-difference promise}
\author{Julius A. Zeiss\thanks{Email: \href{mailto:jzeiss@physik.rwth-aachen.de}{jzeiss@physik.rwth-aachen.de}}}
\date{}
\makeatletter
\renewcommand{\@maketitle}{\newpage\null\vskip0.5em
  \begin{center}{\LARGE\@title\par}\vskip0.6em
    {\large\@author\par}\vskip0.3em
    {\small Institute for Quantum Information, RWTH Aachen University, Aachen, Germany\par}
    \vskip0.7em\end{center}}
\makeatother
\hypersetup{pdftitle={Optimal entanglement-assisted source coding under a balanced-difference promise},
  pdfauthor={Julius A. Zeiss}}
\begin{document}
\maketitle
\begin{abstract}
Entanglement can reduce the communication required for coding tasks,
but establishing the minimum achievable cost is essential to understanding
its limits. We address this question in a zero-error source-coding task
where Alice receives a word and Bob knows an unordered pair of candidates
containing it. Alice does not know the pair and must enable Bob to identify
her word without error using shared entanglement and one classical message.
The candidates satisfy a \emph{balanced-difference promise}: for words in
\(\Z_q^n\) with \(n=q\ell\), each residue modulo \(q\) occurs
exactly \(\ell\) times in their coordinatewise difference. For all
integers \(q\geq2\) and \(\ell\geq1\), we prove that the task
requires exactly \(n\) messages when \((q-1)\ell\) is even and
two messages when it is odd. These minima
allow arbitrary finite-dimensional shared states independent of the
inputs and arbitrary local measurements. In even parity, this establishes
optimality of an existing entanglement-assisted protocol. In odd parity, an explicit
deterministic protocol achieves the optimum of one bit without entanglement.
Our proof combines Fourier analysis with a combinatorial counting argument
to determine the smallest eigenvalue of the associated graphs. In even
parity, this resolves the spectral assertion of Cao et al.'s Conjecture~6.3
for balanced cyclic generalized Hadamard graphs. Together with an
explicit odd-parity bipartition, this determines the quantum chromatic
number as \(n\) in even parity and \(2\) in odd parity, where the
classical chromatic number is also \(2\). All lemmas, theorems, and
corollaries are formalized and verified in Lean.
\end{abstract}

\setcounter{tocdepth}{1}
\begingroup
\makeatletter
\patchcmd{\l@section}{1.0em}{0.6em}{}{}
\makeatother
\tableofcontents
\endgroup
\setcounter{tocdepth}{1}
\clearpage

\section{Introduction and main result}\label{sec:introduction}
When can entanglement assistance reduce the communication required for a
coding task below what any classical protocol can achieve? Understanding
such advantages is a central problem in quantum information theory.
Entanglement advantages have been established for zero-error channel
coding~\cite{clmw,lmmor} and for source and source-channel
coding~\cite{bblps}, while general bounds constrain the performance
of entanglement-assisted protocols~\cite{beigi,cmrssw}.
Beyond establishing a gap, it is equally important to determine the
limits of entanglement assistance: how little communication is possible,
and which protocols attain this minimum?

We address this optimality question in a zero-error source-coding task
where the receiver already has two candidates. Alice receives a word \(x\),
while Bob receives an unordered pair \(\{u,v\}\) containing it.
Alice does not know the pair. They may share a finite-dimensional
quantum state independent of their inputs and use arbitrary local POVMs.
Alice sends Bob one noiseless classical message. A single protocol must
allow Bob to recover \(x\) perfectly: with probability one for every
allowed pair and either choice of \(x\in\{u,v\}\). The allowed pairs obey a \emph{balanced-difference promise}.
Let \(q\geq2\) and \(\ell\geq1\) be integers and set \(n=q\ell\).
All words lie in \(\Z_q^n\): a word
\(x=(x_1,\ldots,x_n)\) has \(n\) coordinate slots, each containing
a letter \(x_j\in\Z_q\). A \emph{position} \(j\in\{1,\ldots,n\}\)
specifies which slot we mean; \(x_j\) is the letter stored there.
For an allowed pair, each residue modulo \(q\) must occur exactly
\(\ell\) times in the coordinatewise difference \(v-u\).
For \(q=2\), this reduces to the
binary distributed Deutsch--Jozsa promise: the candidates disagree
in exactly half the positions~\cite{gqz}.

The optimal communication cost depends on the parity of \((q-1)\ell\).
When it is odd, we construct an explicit partition of
the words into two groups and prove that every allowed candidate
pair contains exactly one word from each group. The same partition
works for all allowed pairs, so Alice can send the group label of
her word without knowing Bob's pair. This yields an optimal
deterministic one-bit source code without entanglement.
Lemma~\ref{lem:odd-bipartition} in Section~\ref{sec:graph} gives
the partition, and
Lemma~\ref{lem:odd-source-coding} in Section~\ref{sec:source-coding}
proves the code's optimality.

We now turn to the even-parity regime. A known entanglement-assisted
protocol solves the task using \(n\) possible messages. Our main
contribution is to prove that this message count is optimal throughout
this regime.

\begin{example}[A balanced-difference pair]\label{ex:balanced-promise}
Take \(q=4\), \(n=8\), and \(\ell=2\). In the pair below, the
coordinatewise difference contains each residue twice, so the pair
satisfies the balanced-difference promise.
\begin{center}
\includegraphics[width=\linewidth]{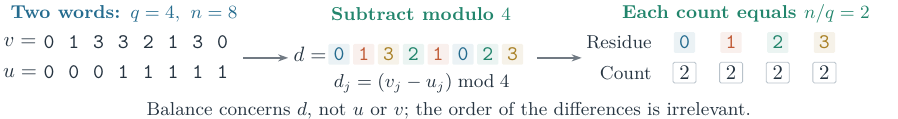}
\end{center}
At these parameters, the optimal protocol described below has eight
possible messages, encoded in three fixed-length bits.
\end{example}

For the task depicted in Figure~\ref{fig:source-coding},
combining the phase representation of Cao et al.~\cite[Lemma~4.1]{cao}
with the source-coding construction of Bri\"et et al.~\cite[Lemma~5.1]{bblps}
gives the following \(n\)-message protocol, illustrated in
Figure~\ref{fig:optimal-protocol}. Alice and Bob share a
maximally entangled state of local dimension \(n\). Given her word
\(x\), Alice performs a measurement determined by \(x\), with \(n\)
possible outcomes, and sends its outcome \(a\) to Bob. The measurement is
constructed so that, for each transmitted outcome, the balanced-difference
promise ensures that the two candidate words in Bob's pair produce
orthogonal conditional states on his subsystem. Knowing the pair and
Alice's message, Bob can therefore distinguish these states and recover
her word without error. This protocol has \(n\) possible messages,
encoded in \(\lceil\log_2 n\rceil\) fixed-length bits. We give its full
construction in Section~\ref{sec:source-coding}.

\begin{question}
When \((q-1)\ell\) is even, can a different encoding, a larger
shared entangled state, or more general measurements use fewer
than \(n\) messages?
\end{question}

In the even-parity regime, earlier results imply
the protocol's optimality for sufficiently large admissible lengths, depending
on \(q\), and in several special parameter families. To determine
whether this code is optimal at every admissible length, we turn the
coding task into a graph whose edges join words that Bob must
distinguish~\cite{wits,fb}. The graph's smallest eigenvalue gives a lower bound on the number of
messages needed for perfect recovery. Fourier analysis reduces the
calculation of this eigenvalue to a counting problem. By solving this
problem, we prove a uniform spectral theorem that resolves the spectral
assertion of Conjecture~6.3 of Cao et al.~\cite{cao} for balanced cyclic
generalized Hadamard graphs and removes the length restriction. It supplies a
matching coding converse: no protocol can use fewer than \(n\) messages,
even with an arbitrary finite-dimensional shared state and arbitrary
local measurements. This establishes the code's optimality for every \(q\geq2\) and
\(n=q\ell\) with \(\ell\geq1\) and \((q-1)\ell\) even.
Together with the deterministic odd-parity protocol, this gives the
complete coding classification.

\begin{tcolorbox}[colback=MidnightBlue!3!white,
  colframe=MidnightBlue!20!white,boxrule=.4pt,arc=1.2mm,
  left=2mm,right=2mm,top=0mm,bottom=0mm,boxsep=1mm,
  before skip=7pt,after skip=14pt]
\begin{informaltheorem}[Coding optimality in both parity cases]
Let \(q\geq2\) and \(\ell\geq1\) be integers, and set \(n=q\ell\).
For the balanced-difference source-coding task, the minimum number
of possible classical messages is:
\begin{itemize}
\item \textbf{Even \((q-1)\ell\):} exactly \(n\) messages,
attained by the entanglement-assisted protocol above. The optimal fixed-length
cost is \(\lceil\log_2 n\rceil\) classical bits.
\item \textbf{Odd \((q-1)\ell\):} exactly two messages,
attained by the deterministic group-label protocol. The optimal
cost is one classical bit, without entanglement.
\end{itemize}
These minima allow arbitrary finite-dimensional shared states
independent of the inputs and arbitrary local POVMs. A single
protocol must recover Alice's word without error for every allowed
pair and either candidate.
\end{informaltheorem}
\end{tcolorbox}

\begin{figure}[tbp]
\centering
\includegraphics[width=\linewidth]{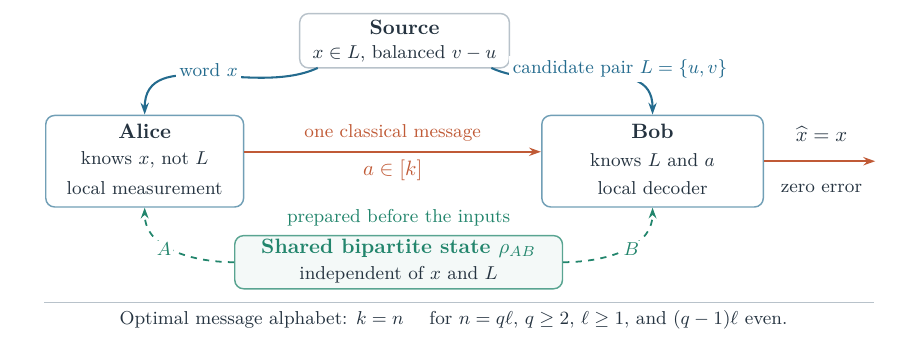}
\caption{Entanglement-assisted zero-error source coding. The source
sends a word \(x\) to Alice and a candidate pair \(L=\{u,v\}\)
containing \(x\) to Bob. The pair obeys the balanced-difference
promise: each residue modulo \(q\) occurs exactly \(\ell\) times in
the coordinatewise difference \(v-u\). Alice's local measurement depends on
\(x\), and she sends its outcome \(a\in[k]=\{1,\ldots,k\}\) as
one of \(k\) possible classical messages. For
\(n=q\ell\), the optimal message alphabet has size \(k=n\)
when \((q-1)\ell\) is even and \(k=2\) when it is odd.
The corresponding fixed-length costs are
\(\lceil\log_2 n\rceil\) bits and one bit, respectively. Bob uses \(L,a\)
and his quantum register to
recover \(x\) exactly. The shared state is independent of the source
inputs. The protocol must succeed for every pair satisfying the
promise and either choice of \(x\in\{u,v\}\).}
\label{fig:source-coding}
\end{figure}

\begin{figure}[tbp]
\centering
\includegraphics[width=\linewidth]{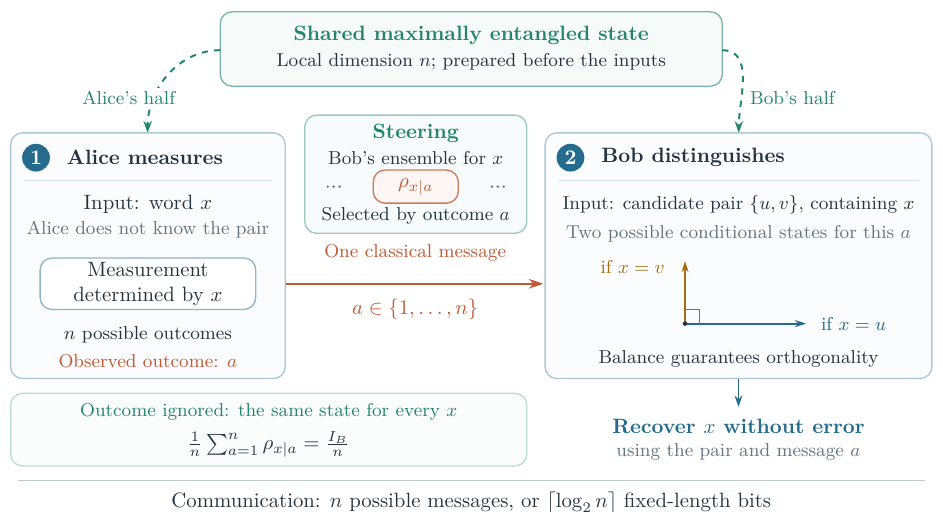}
\caption{The \(n\)-message entanglement-assisted protocol for the setup in
Figure~\ref{fig:source-coding}. Alice and Bob share a maximally
entangled state of local dimension \(n\). Alice's input-dependent
measurement prepares an ensemble of conditional states on Bob's
subsystem, illustrating the steering mechanism~\cite[Sec.~II]{cs-steering}.
Here \(\rho_{x\mid a}\) denotes Bob's normalized state conditioned on
Alice's word \(x\) and outcome \(a\). Every outcome has probability
\(1/n\), and averaging gives \(I_B/n\), independently of \(x\),
where \(I_B\) is the identity on Bob's subsystem. The highlighted
ensemble member is selected by the random outcome \(a\); Alice sends
this label as the sole classical message.
For each fixed \(a\), the balance promise guarantees orthogonality of
Bob's two candidate states; the rays schematically depict these
alternatives. Knowing the pair and \(a\), Bob recovers \(x\) without
error. The construction combines the phase representation of Cao
et al.~\cite[Lemma~4.1]{cao} with the source-coding construction of
Bri\"et et al.~\cite[Lemma~5.1]{bblps}. For \(q\geq2\),
\(n=q\ell\), \(\ell\geq1\), and even \((q-1)\ell\), our converse
proves that at least \(n\) possible messages are necessary even with arbitrary
finite-dimensional shared states and local measurements.}
\label{fig:optimal-protocol}
\end{figure}

Corollary~\ref{cor:source-coding} and
Lemma~\ref{lem:odd-source-coding} prove the even- and odd-parity
coding statements, respectively. Theorem~\ref{thm:quantum-coloring}
also determines the quantum chromatic number as \(n\) in even parity
and \(2\) in odd parity; in the latter case, the classical chromatic
number is \(2\) as well.
Corollary~\ref{cor:source-block-coding} shows that jointly encoding
\(m\) independent instances requires \(n^m\) messages in even parity
and \(2^m\) messages in odd parity; the latter is achieved by sending
one bipartition bit per instance without entanglement.
For even \((q-1)\ell\), Corollary~\ref{cor:near-balanced-coding}
also establishes exact optimality for words of length \(q\ell-1\)
when one residue occurs once fewer than the others,
using the existing reduction of Luo, Ning, and
Zhang~\cite[Theorem~6 and Corollary~7]{lnz}.

In even parity, protocols without shared entanglement require
\(\Theta(n)\) fixed-length classical bits for each fixed \(q\) and
all sufficiently large lengths \(n=q\ell\) in this regime, compared
with exactly \(\lceil\log_2 n\rceil\) bits with entanglement.
A simple classical code sends Alice's letters in \(\ell+1\) fixed
coordinates and uses \(\lceil(\ell+1)\log_2 q\rceil\) bits.
The \(\Omega(n)\) lower bound follows from the chromatic-number
bound of~\cite[Theorem~1.9]{cao} and the standard source-coding
correspondence~\mbox{\cite[Sec.~1.1]{bblps}}.
This large-length separation follows from existing results;
Appendix~\ref{app:prior-regimes} derives both bounds and the
coordinate-sending construction, and explains why a strict advantage
requires steering. For fixed \(q\), these bounds show
that the optimal classical communication cost grows linearly with
\(n\), but they do not generally determine the exact minimum
number of bits.

Entanglement need not reduce the message count at every even-parity
length: classical codes already attain the
\(n\)-message minimum in the small examples of
Appendix~\ref{app:full-graphs}. In odd parity, the optimal one-bit
classical protocol leaves no room for an entanglement advantage at
any length.

The balanced-difference promise also connects the task to classical
coding theory.
In generalized Hadamard theory, families of rows with pairwise balanced
differences form equidistant
codes with Hamming distance \(n(1-1/q)\)~\cite[Sec.~II]{bcp}.
The restriction on Bob's candidate pair substantially reduces the
communication needed for perfect recovery: if every distinct pair were
allowed, \(q^n\) possible messages would be necessary, even
with shared entanglement~\cite[Theorem~2.1]{bblps}.
Under the balanced-difference promise, we determine this minimum
exactly: \(n\) messages in even parity and two in odd parity.

\subsection{Previous work}

Zero-error source coding with receiver side information has a
classical graph-theoretic formulation. The \emph{characteristic graph}
joins two source inputs whenever they can occur with the same side
information, and a proper coloring assigns distinguishable messages
to every such pair. Witsenhausen~\cite{wits} showed that the minimum
message alphabet size equals the graph's chromatic number. Bri\"et
et al.~\cite{bblps} studied the entanglement-assisted version,
including protocols, lower bounds, and separations from classical
coding. Cubitt, Man\v{c}inska, Roberson, Severini, Stahlke, and
Winter~\cite{cmrssw} formulated entanglement-assisted source-channel
coding through graph homomorphisms and obtained bounds
using the Lov\'asz theta number and its variants.

Quantum graph coloring provides a related connection between graphs
and shared entanglement. Cameron et al.~\cite{cmnsw} studied the
quantum chromatic number, and Man\v{c}inska and Roberson~\cite{mr}
developed the broader framework of quantum graph homomorphisms.
A quantum \(k\)-coloring gives an entanglement-assisted source code
with \(k\) messages, although the two optimization problems allow different
strategies~\cite[Sec.~1.3]{bblps}. Elphick and Wocjan~\cite{ew}
established spectral lower bounds for the quantum chromatic number,
including a quantum version of Hoffman's bound.

Hadamard graphs are a natural setting for these methods.
The graph papers discussed below state spectral and coloring results,
but do not explicitly derive their consequences for the coding task
considered here. We derive these consequences in
Appendix~\ref{app:prior-regimes} using established source-coding
constructions and lower bounds~\cite{bblps}.
For the binary family, spectral bounds give quantum chromatic
number \(n\) whenever \(4\mid n\)~\cite[Sec.~3.2]{wed}.
Cao et al.~\cite{cao} studied generalized Hadamard graphs and
constructed quantum colorings from phase encodings. Their cyclic
construction~\cite[Lemma~4.1]{cao}, combined with the source-coding
construction of~\cite[Lemma~5.1]{bblps}, gives an \(n\)-message
code whenever \(q\mid n\), without a parity assumption.
Luo, Ning, and Zhang~\cite{lnz} proved the ternary spectral formula
and quantum chromatic number for every \(3\mid n\).
The balanced quaternary case \(8\mid n\) follows from lemmas of
Ning, Koolen, and Zhang~\cite{nkz} together with elementary estimates.
Further coding optimality cases follow for prime \(q\) and
\(n=q^r\), \(r\geq1\), and from generalized Hadamard matrices
of order \(n\) over \(\Z_q\), whose rows form an \(n\)-clique.
Appendix~\ref{app:prior-regimes} also gives the additional quaternary
calculation.

For general cyclic alphabets, write \(D=n!/(\ell!)^q\) for the
graph's degree. In the odd-parity case, the sign identity of
Cao et al.~\cite[Lemma~4.3]{cao}, together with regularity, implies
that the least adjacency eigenvalue is \(-D\) at every length in
this regime. For \(q=2\), the corresponding classical and quantum
chromatic numbers are already known to equal \(2\) whenever
\(n\equiv2\pmod4\)~\cite[Sec.~2]{lnz}.
Throughout the full odd-parity regime,
Lemma~\ref{lem:odd-bipartition} constructs an explicit bipartition
and proves that every edge joins different classes. Hence every
allowed candidate pair contains one word from each class. From this
partition, we derive an optimal deterministic one-bit source code in
Lemma~\ref{lem:odd-source-coding} and determine the classical
and quantum chromatic numbers in Theorem~\ref{thm:quantum-coloring}.
In the even-parity case, Cao et al.~\cite[Lemma~4.4]{cao}
proved that, for each fixed \(q\), there is a threshold \(N(q)\)
such that the least adjacency eigenvalue is \(-D/(n-1)\) for all
admissible lengths \(n=q\ell\) satisfying \(n\geq N(q)\).
The difficulty is to show that no other eigenvalue is smaller.
Their proof expresses the eigenvalues through circulant determinant
coefficients and combines eigenvalue-multiplicity bounds with
permanent coefficient estimates. These estimates yield the required
comparison only beyond a length threshold depending on \(q\).
Cao et al.~\cite[Conjecture~6.3]{cao} conjectured that the same
least-eigenvalue formula holds for every \(q\geq2\), \(\ell\geq1\),
and \(n=q\ell\) with \((q-1)\ell\) even, leaving open a theorem
valid at every admissible length.

\paragraph{Structure of the paper.}
Section~\ref{sec:graph} defines the graph and establishes its basic
properties and odd-parity bipartition. It then connects the least
eigenvalue to coding optimality, states the main spectral result,
and proves its odd-parity case.
Sections~\ref{sec:fourier}--\ref{sec:spectral} develop the counting
argument and prove the least-eigenvalue formula in the even-parity case.
Sections~\ref{sec:source-coding} and~\ref{sec:quantum} establish the
coding and coloring results. Section~\ref{sec:discussion} discusses the results
and open problems. Appendix~\ref{app:notation} collects the principal
notation. Appendices~\ref{app:normal} and~\ref{app:prior-regimes}
provide supporting proofs and consequences of earlier results.
Appendix~\ref{app:examples} collects further worked examples.

\section{From source coding to graph spectra}\label{sec:graph}

To prove the lower bound for balanced pairs, we represent the coding
problem by a graph that records the ambiguities the code must resolve.
Its vertices are Alice's possible words, and two words are adjacent
if and only if they can occur together in Bob's candidate pair.
This is the \emph{characteristic graph} of the source with side
information~\cite{wits,fb}; see also~\cite[Sec.~1.1]{bblps}.
We first define the graph and collect the properties needed for its
spectral analysis.

Let \(q\geq2\) and \(\ell\geq1\) be integers, and set \(n=q\ell\).
The cyclic group \(\Z_q=\{0,1,\ldots,q-1\}\) uses addition modulo
\(q\). For words \(u,v\in\Z_q^n\), the difference \(s=v-u\)
records the changes that take \(u\) to \(v=u+s\). All word
operations are coordinatewise modulo \(q\). The balanced-difference
promise permits exactly the shifts in
\begin{equation}\label{eq:allowed-shifts}
 S_{q,n}=\bigl\{s\in\Z_q^n:
       \#\{j:s_j=c\}=\ell\text{ for every }c\in\Z_q\bigr\}.
\end{equation}
Here \(\#\) denotes the number of elements in a set. Thus a shift is
\emph{balanced} when each of the \(q\) letters occurs exactly
\(\ell\) times.

The \emph{balanced cyclic generalized Hadamard graph}
\(\Om=\operatorname{Cay}(\Z_q^n,S_{q,n})\) has vertex set
\(\Z_q^n\), with adjacency defined by
\begin{equation}\label{eq:adjacency-definition}
 u\sim v\quad\Longleftrightarrow\quad v-u\in S_{q,n}.
\end{equation}
The notation \(u\sim v\) means that \(u\) and \(v\) are adjacent.
The Cayley graph description means that the same allowed shifts
describe the neighbors of every vertex: they are the words \(u+s\)
with \(s\in S_{q,n}\). This graph is the cyclic-group specialization
of the generalized Hadamard graph in~\cite[Sec.~2.3]{cao}. Example~\ref{ex:shift-neighborhood} in Appendix~\ref{app:examples}
works out a vertex's neighborhood.

The following lemma records the graph's size, edge structure, degree,
and symmetries.

\begin{lemma}[Basic graph properties]\label{lem:graph-properties}
Let \(q\geq2\) and \(\ell\geq1\) be integers, and let \(n=q\ell\).
The graph \(\Om\) has the following properties.
\begin{enumerate}[label=\textup{(\roman*)},leftmargin=*,itemsep=4pt]
\item\label{item:graph-size}
It is finite, with \(q^n\) vertices.

\item\label{item:graph-simple}
It is simple and undirected: there are no loops or multiple edges,
and \(u\sim v\) if and only if \(v\sim u\).

\item\label{item:graph-degree}
For every vertex \(u\), its neighborhood \(N(u)\), the set of its
neighbors, is
\begin{equation}
 N(u)=\{u+s:s\in S_{q,n}\}.
\end{equation}
The graph is regular, meaning that every vertex has the same number
of neighbors. Its common degree is
\begin{equation}\label{eq:degree}
 D_{q,n}=|S_{q,n}|=\frac{n!}{(\ell!)^q}
                  =\mult{n}{\ell,\ldots,\ell}.
\end{equation}

\item\label{item:graph-translations}
For every \(t\in\Z_q^n\), translation \(u\mapsto u+t\) is a
bijection preserving adjacency:
\begin{equation}
 u\sim v\quad\Longleftrightarrow\quad u+t\sim v+t.
\end{equation}
In particular, the graph is \emph{vertex-transitive}: for any two
vertices, there is an adjacency-preserving bijection of the vertex
set taking the first to the second.

\item\label{item:graph-coordinates}
Let \(\pi\) be a permutation of \(\{1,\ldots,n\}\), the coordinate
positions. The map \(u\mapsto\pi u\), defined by
\((\pi u)_j=u_{\pi^{-1}(j)}\), is a bijection preserving adjacency:
\begin{equation}
 u\sim v\quad\Longleftrightarrow\quad\pi u\sim\pi v.
\end{equation}
\end{enumerate}
\end{lemma}

\begin{proof}
We prove the statements in the order listed.

\emph{\ref{item:graph-size}.}
Each of the \(n\) coordinates can contain any of the \(q\) letters,
so the vertex set \(\Z_q^n\) has \(q^n\) elements.

\emph{\ref{item:graph-simple}.}
The zero word is not balanced: it contains no copy of the letter
\(1\), whereas \(\ell\geq1\). Thus \(u\not\sim u\) for every
vertex \(u\). If \(s\) is balanced, then so is \(-s\), because
the map \(c\mapsto-c\) permutes the residues modulo \(q\).
Consequently \(S_{q,n}=-S_{q,n}\), and
\begin{equation}
 v-u\in S_{q,n}\quad\Longleftrightarrow\quad u-v\in S_{q,n}.
\end{equation}
This proves that adjacency is symmetric. The adjacency rule assigns
at most one edge to each unordered pair of distinct vertices, so
there are no multiple edges.

\emph{\ref{item:graph-degree}.}
By definition, \(v\) is a neighbor of \(u\) precisely when
\(v=u+s\) for some \(s\in S_{q,n}\). For fixed \(u\), the map
\(s\mapsto u+s\) is a bijection from \(S_{q,n}\) to \(N(u)\):
its inverse is \(v\mapsto v-u\). Hence every vertex has exactly
\(|S_{q,n}|\) neighbors. To count this set, arrange \(n\)
distinguishable objects, with \(\ell\) objects assigned to each
letter, and then forget the distinctions between objects assigned
to the same letter. Each resulting word is counted \((\ell!)^q\)
times among the \(n!\) arrangements. Therefore
\(|S_{q,n}|=n!/(\ell!)^q\). This quotient counts the words of
length \(n\) in which each of the \(q\) letters occurs exactly
\(\ell\) times.

\emph{\ref{item:graph-translations}.}
Translation by \(t\) is a bijection with inverse translation by
\(-t\), and
\begin{equation}
 (v+t)-(u+t)=v-u.
\end{equation}
It therefore preserves adjacency in both directions. Given vertices
\(u,v\), translation by \(t=v-u\) maps \(u\) to \(v\), proving
vertex transitivity. Here \(t\) may be any word; it need not itself
be a balanced shift.

\emph{\ref{item:graph-coordinates}.}
The coordinate map induced by \(\pi\) is a bijection, with inverse
induced by \(\pi^{-1}\). It satisfies
\begin{equation}
 \pi v-\pi u=\pi(v-u).
\end{equation}
Permuting coordinates leaves the number of occurrences of every
letter unchanged. Thus \(v-u\) is balanced if and only if
\(\pi(v-u)\) is balanced, which proves the last assertion.
\end{proof}

The balanced shifts also determine an explicit bipartition in the
odd-parity regime.

\begin{lemma}[Bipartition in the odd-parity regime]\label{lem:odd-bipartition}
Let \(q\geq2\), \(\ell\geq1\), and \(n=q\ell\), with
\((q-1)\ell\) odd. For each vertex \(x\in\Z_q^n\), let \(\sigma(x)\) be the
representative of \(\sum_jx_j\pmod q\) in \(\{0,\ldots,q-1\}\), and
define \(b:\Z_q^n\to\{0,1\}\) by
\begin{equation}\label{eq:odd-bipartition}
 b(x)=
 \begin{cases}
  0,&0\leq\sigma(x)<q/2,\\
  1,&q/2\leq\sigma(x)<q.
 \end{cases}
\end{equation}
Then \(b(v)=1-b(u)\) for every edge \(\{u,v\}\) of \(\Om\).
Consequently, the classes \(\{x:b(x)=0\}\) and \(\{x:b(x)=1\}\)
form a bipartition of \(\Om\).
\end{lemma}
\begin{proof}
The parity hypothesis implies that \(q\) is even and \(\ell\) is odd.
Consider any edge \(\{u,v\}\). Its difference
\(s=v-u\) contains every residue exactly \(\ell\) times, so
\begin{equation}
 \sigma(v)-\sigma(u)\equiv\ell\sum_{c=0}^{q-1}c
          =\frac q2\ell(q-1)\equiv\frac q2\pmod q.
\end{equation}
Adding \(q/2\) modulo \(q\) interchanges the residue sets
\(\{0,\ldots,q/2-1\}\) and \(\{q/2,\ldots,q-1\}\).
Hence \(b(v)=1-b(u)\) for every edge, proving the asserted
bipartition.
\end{proof}

We write \(D\) for \(D_{q,n}\) when \(q\) and \(n\) are fixed.
The next subsection introduces the adjacency matrix and explains how
these graph properties support its spectral analysis.

\subsection{Adjacency matrices and spectral bounds}\label{sec:adjacency-spectra}
The allowed shifts describe which pairs of words Bob must distinguish.
We now record these pairs in the graph's adjacency matrix, introduce
its eigenvalues, and establish the spectral inequalities used below.
Section~\ref{sec:spectral-coding} then uses the degree and least
eigenvalue to bound the number of messages.

Let \(G\) be a finite simple undirected graph with nonempty vertex
set \(V\) and unordered-edge set \(E\). The adjacency matrix \(A_G\) records the edges: its entry at \(u,v\)
is one when the vertices are adjacent and zero otherwise:
\begin{equation}
 (A_G)_{u,v}=\begin{cases}
  1, & \text{if }u\sim v,\\
  0, & \text{otherwise.}
 \end{cases}
\end{equation}
For \(\Om\), write \(A=A_{\Om}\). By
Lemma~\ref{lem:graph-properties}\,\ref{item:graph-size}, this matrix has
\(N=q^n\) rows and columns, one for each possible word.
Let \(\one\) denote the vector whose entries are all one, and let
\(\mathfrak S_n\) be the group of coordinate permutations, acting as
in Lemma~\ref{lem:graph-properties}\,\ref{item:graph-coordinates}.
The following lemma records the matrix consequences of
Lemma~\ref{lem:graph-properties}.

\begin{lemma}[Adjacency structure]\label{lem:adjacency-structure}
Let \(q\geq2\), \(\ell\geq1\), and \(n=q\ell\). Then
\begin{equation}\label{eq:adjacency-structure}
 A_{uv}\in\{0,1\},\qquad A_{uu}=0,\qquad
 A^{\mathsf T}=A,\qquad A\one=D_{q,n}\one,
\end{equation}
where \(D_{q,n}=n!/(\ell!)^q\). Moreover, for all
\(u,v,t\in\Z_q^n\) and every \(\pi\in\mathfrak S_n\),
\begin{equation}\label{eq:adjacency-symmetries}
 A_{u+t,v+t}=A_{uv},\qquad A_{\pi u,\pi v}=A_{uv}.
\end{equation}
\end{lemma}
\begin{proof}
By the definition of the adjacency matrix and
Lemma~\ref{lem:graph-properties}\,\ref{item:graph-simple}, its entries
lie in \(\{0,1\}\), its diagonal is zero, and \(A^{\mathsf T}=A\).
By part~\ref{item:graph-degree}, every row contains exactly
\(D_{q,n}\) ones, so \(A\one=D_{q,n}\one\).
Finally, parts~\ref{item:graph-translations} and~\ref{item:graph-coordinates}
show that translations and coordinate permutations preserve adjacency,
which gives the two matrix-entry identities.
\end{proof}

For a general graph \(G\), a function \(f:V\to\C\) assigns a value
\(f(u)\) to each vertex \(u\). We identify this function with the vector of its values,
in the vertex order used for \(A_G\).
For two such functions, the standard inner product is
\(\langle f,g\rangle=\sum_{u\in V}\overline{f(u)}g(u)\),
and the associated norm satisfies \(\|f\|^2=\langle f,f\rangle\).
At \(u\), the product \(A_Gf\) is the sum of the values at its
neighbors: \((A_Gf)(u)=\sum_{v\sim u}f(v)\).
A nonzero vector \(f\) is an eigenvector with \emph{eigenvalue}
\(\lambda\) if this sum equals \(\lambda f(u)\) at every vertex,
with the same factor \(\lambda\) throughout. This is the equation
\(A_Gf=\lambda f\). These eigenvalues form the
graph's \emph{adjacency spectrum}. They are real because \(A_G\) is
real symmetric; for \(\Om\), this symmetry is recorded in
Lemma~\ref{lem:adjacency-structure}. We write
\(\tau=\lambda_{\min}(G)\) for the smallest. Example~\ref{ex:triangle} in Appendix~\ref{app:examples}
computes the adjacency spectrum of a triangle.

The following lemma collects the Rayleigh--Ritz characterization
and its consequences for regular graphs. Part~\ref{item:regular-spectral-bound}
decomposes each vector into a constant component and a component
orthogonal to it. It uses the known eigenvalue \(D\) on the constant
component and applies Rayleigh--Ritz to the remaining component.
We use the resulting inequality for the communication and coloring lower bounds in
Sections~\ref{sec:source-coding} and~\ref{sec:quantum}.

\begin{lemma}[Rayleigh--Ritz and spectral inequalities]
\label{lem:regular-spectral-inequality}
Let \(G\) be a finite simple undirected graph with vertex set \(V\),
unordered-edge set \(E\), and \(N=|V|\geq1\). Let \(A_G\) be its
adjacency matrix and \(\tau=\lambda_{\min}(G)\) its least eigenvalue.
\begin{enumerate}[label=\textup{(\roman*)},leftmargin=*,itemsep=4pt]
\item\label{item:rayleigh-ritz}
The Rayleigh--Ritz characterization~\cite[Sec.~2.4]{bh} gives
\begin{equation}\label{eq:rayleigh-ritz}
 \tau
 =\min_{0\ne f\in\mathbb R^V}
   \frac{f^{\mathsf T}A_Gf}{f^{\mathsf T}f}
 =\min_{0\ne f\in\mathbb R^V}
   \frac{2\sum_{\{u,v\}\in E}f(u)f(v)}{\sum_{u\in V}f(u)^2}.
\end{equation}
In particular, \(\langle f,A_Gf\rangle\geq\tau\|f\|^2\)
for every \(f\in\C^V\).

\item\label{item:regular-neighbor-average}
If \(G\) is \(D\)-regular, then \(A_G\one=D\one\), where
\(\one\) is the vector whose entries are all one. Thus \(D\)
is an eigenvalue. If \(D>0\), the operator \(A_G/D\) averages
the values at neighboring vertices:
\begin{equation}\label{eq:neighbor-average}
 \left(\frac{A_G}{D}f\right)(u)
 =\frac1D\sum_{v\sim u}f(v)
 \qquad(f\in\C^V,\ u\in V).
\end{equation}

\item\label{item:regular-spectral-bound}
If \(G\) is \(D\)-regular, let \(e=N^{-1/2}\one\) be the
normalized constant vector. Then, for every \(f\in\C^V\),
\begin{equation}\label{eq:regular-spectral-quadratic}
 \langle f,A_Gf\rangle
 \geq\tau\|f\|^2+(D-\tau)|\langle e,f\rangle|^2.
\end{equation}
Equivalently,
\begin{equation}\label{eq:regular-spectral-operator}
 A_G\succeq\tau I+(D-\tau)ee^*.
\end{equation}
Here \(I\) is the identity matrix, \(e^*\) is the conjugate transpose
of \(e\), and \(ee^*=|e\rangle\langle e|\) is the orthogonal
projection onto the constant vectors. The symbol \(\succeq\)
denotes positive semidefinite order: the difference between the
left and right sides is positive semidefinite.
\end{enumerate}
\end{lemma}

\begin{proof}
\emph{\ref{item:rayleigh-ritz}.}
Since \(A_G\) is real symmetric, it has an orthonormal basis of real
eigenvectors \(v_1,\ldots,v_N\), with corresponding eigenvalues
\(\lambda_1,\ldots,\lambda_N\). Writing a nonzero real vector as
\(f=\sum_{j=1}^N c_jv_j\), where \(c_j\in\mathbb R\), gives
\begin{equation}\label{eq:rayleigh-weighted-average}
 \frac{f^{\mathsf T}A_Gf}{f^{\mathsf T}f}
 =\frac{\sum_{j=1}^N\lambda_jc_j^2}{\sum_{j=1}^Nc_j^2}.
\end{equation}
The coefficients \(c_j^2/\sum_{i=1}^Nc_i^2\) are nonnegative
and sum to one. The quotient is therefore a weighted average of
the eigenvalues and cannot be smaller than \(\tau\). Choosing
\(f\) to be an eigenvector for \(\tau\) attains equality.
Moreover, each unordered edge is counted twice in the matrix product,
so
\begin{equation}
 f^{\mathsf T}A_Gf=2\sum_{\{u,v\}\in E}f(u)f(v).
\end{equation}
This proves \eqref{eq:rayleigh-ritz}. For a complex vector
\(f=x+iy\), with \(x,y\in\mathbb R^V\), symmetry gives
\begin{equation}
 \langle f,A_Gf\rangle
 =x^{\mathsf T}A_Gx+y^{\mathsf T}A_Gy
 \geq\tau\bigl(\|x\|^2+\|y\|^2\bigr)
 =\tau\|f\|^2.
\end{equation}
The real-vector bound also holds at zero, so this proves the stated
inequality for every complex vector.

\emph{\ref{item:regular-neighbor-average}.}
Each vertex has exactly \(D\) neighbors, so
\((A_G\one)(u)=\sum_{v\sim u}1=D\) for every \(u\).
Thus \(A_G\one=D\one\). When \(D>0\), dividing
\((A_Gf)(u)=\sum_{v\sim u}f(v)\) by \(D\) gives
\eqref{eq:neighbor-average}, the arithmetic mean of the neighboring
values.

\emph{\ref{item:regular-spectral-bound}.}
Write \(f=ce+w\), where \(c=\langle e,f\rangle\) and
\(w\perp e\). Part~\ref{item:regular-neighbor-average} gives
\(A_Ge=De\), and symmetry makes the cross terms vanish.
By part~\ref{item:rayleigh-ritz},
\(\langle w,A_Gw\rangle\geq\tau\|w\|^2\). Hence
\begin{equation}
 \begin{aligned}
  \langle f,A_Gf\rangle
  &=D|c|^2+\langle w,A_Gw\rangle\\
  &\geq D|c|^2+\tau\|w\|^2\\
  &=\tau\|f\|^2+(D-\tau)|\langle e,f\rangle|^2.
 \end{aligned}
\end{equation}
Since this holds for every \(f\), it is equivalent to
\eqref{eq:regular-spectral-operator}.
\end{proof}

\subsection{From the least eigenvalue to coding optimality}\label{sec:spectral-coding}
We now connect these spectral quantities to the coding task.
The following general bound shows how the graph's degree and least
eigenvalue constrain the number of messages required for perfect recovery.

\begin{lemma}[Spectral communication bound]\label{lem:source-hoffman}
Let \(G\) be a finite simple undirected \(D\)-regular graph with
least adjacency eigenvalue \(\tau<0\). Consider the coding setup
of Figure~\ref{fig:source-coding}, with characteristic graph \(G\):
the allowed candidate pairs are precisely the edges of \(G\).
Any protocol that recovers either
candidate with certainty for every edge, using one classical
message \(a\in[k]\), satisfies
\begin{equation}\label{eq:source-spectral-bound}
 k\geq1-\frac{D}{\tau}.
\end{equation}
The bound holds for arbitrary local POVMs and finite-dimensional
shared states independent of the inputs.
\end{lemma}

This bound follows by combining the general source-coding
bound~\cite[Theorem~2.1]{bblps} with the spectral
bound~\cite[Sec.~2.1]{bdov}.
Appendix~\ref{app:prior-regimes} gives this derivation in
Eq.~\eqref{eq:prior-spectral-coding} and explains its coding
consequences for previously established spectral formulas.
We give a direct proof in Section~\ref{sec:source-coding}.
Example~\ref{ex:triangle} in Appendix~\ref{app:examples}
illustrates this bound for the triangle.

For \(\Om\), the following spectral theorem supplies the required
least eigenvalue.

\begin{theorem}[Least adjacency eigenvalue]\label{thm:main}
Let \(q\geq2\), \(\ell\geq1\), and \(n=q\ell\).
If \((q-1)\ell\) is even, then
\begin{equation}
 \lambda_{\min}(\Om)=-\frac{D_{q,n}}{n-1}.
\end{equation}
If \((q-1)\ell\) is odd, the least adjacency eigenvalue is
\(-D_{q,n}\).
\end{theorem}

The odd-parity least-eigenvalue formula holds without a length
threshold and is already implied by the sign identity of
Cao et al.~\cite[Lemma~4.3]{cao} together with regularity.
We give a direct proof using the bipartition of Lemma~\ref{lem:odd-bipartition}.

\begin{proof}[Proof of the odd-parity case of Theorem~\ref{thm:main}]
Suppose \((q-1)\ell\) is odd, and let \(b\) be the class label
from Lemma~\ref{lem:odd-bipartition}. Set \(f(x)=(-1)^{b(x)}\),
which is nonzero and satisfies \(f(v)=-f(u)\) whenever \(u\sim v\).
By Lemma~\ref{lem:graph-properties}\,\ref{item:graph-degree}, the graph
is \(D_{q,n}\)-regular, independently of parity. Thus
\begin{equation}
 (Af)(u)=\sum_{v\sim u}f(v)=-D_{q,n}f(u).
\end{equation}
Hence \(-D_{q,n}\) is an adjacency eigenvalue. Every eigenvalue
of a \(D_{q,n}\)-regular graph is at least
\(-D_{q,n}\)~\cite[Sec.~3.1]{bh}, so it is the least eigenvalue.
\end{proof}

For \((q-1)\ell\) even, Lemma~\ref{lem:graph-properties} supplies
the graph hypotheses of Lemma~\ref{lem:source-hoffman}. Together with
the even-parity assertion of Theorem~\ref{thm:main}, this gives
\(k\geq1-D/\tau=n\), proving optimality of the existing
\(n\)-message code.

The odd-parity spectral assertion is now proved. The remaining proof of
Theorem~\ref{thm:main} concerns \((q-1)\ell\) even and requires
showing that \(-D_{q,n}/(n-1)\) is an eigenvalue and that no eigenvalue
is smaller. The main work is to rule out smaller eigenvalues.
Section~\ref{sec:fourier} uses Fourier analysis to express the
eigenvalues and reduce this task to a counting problem.
Section~\ref{sec:switching} proves the required counting bound.
In Section~\ref{sec:spectral}, we combine these results with an
explicit eigenvector to establish the formula.

\section{From the graph spectrum to a counting problem}\label{sec:fourier}
Throughout Sections~\ref{sec:fourier}--\ref{sec:spectral}, fix integers
\(q\geq2\) and \(\ell\geq1\), set \(n=q\ell\), and assume
\((q-1)\ell\) is even. In particular, \(n\geq3\).
All statements in these sections use these standing assumptions.

We now use Fourier analysis to compute the adjacency eigenvalues
and express them through polynomial coefficients. This reduces the
remaining assertion of Theorem~\ref{thm:main} to a counting problem
addressed in Section~\ref{sec:switching}.

We write \(A=A_{\Om}\) for the adjacency matrix
introduced in Section~\ref{sec:adjacency-spectra}, whose matrix
properties are recorded in Lemma~\ref{lem:adjacency-structure}.

\subsection{Fourier waves on the graph}\label{sec:fourier-waves}
To find the eigenvalues of \(\Om\), we use its description by
shifts. By
Lemma~\ref{lem:graph-properties}\,\ref{item:graph-degree}, the neighbors
of \(u\) are precisely the words \(u+s\) with \(s\in S_{q,n}\).
Thus applying the adjacency operator means summing the shifted
functions \(u\mapsto f(u+s)\). This makes the graph accessible
to the standard Fourier diagonalization of Cayley graphs over finite
abelian groups~\cite[Sec.~1.4.9]{bh}.
A Fourier wave is a nonzero function \(f:\Z_q^n\to\C\) such that
shifting its argument by any fixed \(s\in\Z_q^n\) multiplies every
value by the same phase factor \(\alpha_s\):
\(f(u+s)=\alpha_s f(u)\) for every vertex \(u\), with
\(|\alpha_s|=1\).
The following lemma turns this behavior under shifts into an
adjacency eigenvalue.

\begin{lemma}[From shifts to adjacency eigenvalues]
\label{lem:shifts-to-eigenvalues}
For every function \(f:\Z_q^n\to\C\) and every vertex \(u\),
\begin{equation}\label{eq:adjacency-shifts}
 (Af)(u)=\sum_{s\in S_{q,n}}f(u+s).
\end{equation}
If \(f\ne0\) and, for each allowed shift \(s\), there is a
complex number \(\alpha_s\) such that
\begin{equation}
 f(u+s)=\alpha_s f(u)\qquad\text{for every vertex }u,
\end{equation}
then \(f\) is an eigenvector of \(A\), with real eigenvalue
\begin{equation}\label{eq:shift-eigenvalue}
 \lambda=\sum_{s\in S_{q,n}}\alpha_s.
\end{equation}
\end{lemma}
\begin{proof}
By Lemma~\ref{lem:graph-properties}\,\ref{item:graph-degree}, the
neighbors of \(u\) are precisely the vertices \(u+s\) with
\(s\in S_{q,n}\), which gives~\eqref{eq:adjacency-shifts}.
Substituting the assumed shift relation yields
\begin{equation}
 (Af)(u)=\sum_{s\in S_{q,n}}\alpha_s f(u)=\lambda f(u).
\end{equation}
The eigenvalue is real because the adjacency matrix is real
symmetric by Lemma~\ref{lem:adjacency-structure}.
\end{proof}

Lemma~\ref{lem:shifts-to-eigenvalues} shows that every Fourier wave
is an adjacency eigenvector. We now construct an orthonormal basis
of such waves, making the standard Fourier diagonalization
explicit~\cite[Lemma~2.3]{cao}.
The vertex set \(\Z_q^n\) is a finite abelian group under
coordinatewise addition modulo \(q\), the direct product of \(n\)
copies of the cyclic group \(\Z_q\).
Put \(\zeta_q=\exp(2\pi i/q)\). For each fixed index
\(a=(a_1,\ldots,a_n)\in\Z_q^n\), define the \emph{Fourier wave}
\(\chi_a:\Z_q^n\to\C\) by
\begin{equation}
 \chi_a(x)=\zeta_q^{a\cdot x}
 =\zeta_q^{a_1x_1+\cdots+a_nx_n}
 =\exp\!\left(\frac{2\pi i}{q}\sum_{j=1}^n a_jx_j\right).
\end{equation}
Thus we raise the fixed complex number \(\zeta_q\) to an integer
power, using integer representatives of the entries of \(a\) and
\(x\). Only the exponent modulo \(q\) matters, because
\(\zeta_q^q=1\), so the value is independent of these representatives.
The index \(a\) selects the wave, while \(x\) is the vertex at
which we evaluate it. Every value \(\chi_a(x)\) has absolute value one. Example~\ref{ex:fourier-wave} in Appendix~\ref{app:examples}
evaluates one of these waves explicitly.

The frequency \(a\) specifies the phase changes: increasing
coordinate \(j\) of \(x\) by one multiplies the wave by
\(\zeta_q^{a_j}\).

These waves are \emph{characters} in the sense of representation
theory. Indeed, they satisfy
\begin{equation}
 \chi_a(x+y)=\chi_a(x)\chi_a(y),
\end{equation}
so multiplication by \(\chi_a(x)\) defines a one-dimensional unitary
representation of \(\Z_q^n\). Its character, the trace of the
representing matrix, is simply \(\chi_a(x)\).
Every irreducible complex representation of a finite abelian group
is one-dimensional, and the \(q^n\) waves above give all the
irreducible characters of \(\Z_q^n\).\footnote{Serre's textbook treats
characters, abelian groups, and cyclic groups
in~\cite[Secs.~2.1, 3.1, and~5.1]{serre}.}

For each \(s\in\Z_q^n\), the shift operator \(T_s\) acts on
functions by \((T_s f)(x)=f(x+s)\). Together, these operators
give the regular representation. Each wave spans an invariant
one-dimensional space because \(T_s\chi_a=\chi_a(s)\chi_a\).
Fourier analysis decomposes the space of functions into these
character spaces. Lemma~\ref{lem:shifts-to-eigenvalues} gives
\(A=\sum_{s\in S_{q,n}}T_s\), so the same decomposition diagonalizes
adjacency, as for any Cayley graph of
a finite abelian group~\cite[Sec.~1.4.9]{bh}.\footnote{For nonabelian
finite groups, irreducible representations can have
dimension greater than one. Fourier reduction then uses their matrix
coefficients and gives matrix blocks in general. Irreducible
characters alone span the class functions, meaning functions constant
on conjugacy classes~\cite[Secs.~2.2--2.5]{serre}.}
The following lemma makes this diagonalization explicit.

\begin{lemma}[Fourier diagonalization]\label{lem:fourier-diagonalization}
The normalized waves \(q^{-n/2}\chi_a\), indexed by
\(a\in\Z_q^n\), form an orthonormal basis of \(\C^{\Z_q^n}\).
They satisfy
\begin{equation}
 A\chi_a=\lambda(a)\chi_a,
 \qquad
 \lambda(a)=\sum_{s\in S_{q,n}}\zeta_q^{a\cdot s}\in\mathbb R.
\end{equation}
The indexed values \(\lambda(a)\) give the entire adjacency spectrum,
counted with multiplicity. Moreover, \(\lambda(\pi a)=\lambda(a)\)
for every \(\pi\in\mathfrak S_n\).
\end{lemma}
\begin{proof}
Each Fourier wave \(\chi_a\) is nonzero and satisfies
\(\chi_a(x+s)=\zeta_q^{a\cdot s}\chi_a(x)\).
Applying Lemma~\ref{lem:shifts-to-eigenvalues} with
\(\alpha_s=\zeta_q^{a\cdot s}\) gives the stated eigenvalue formula.

For \(a,b\in\Z_q^n\), character orthogonality follows from
\begin{equation}
 \sum_{x\in\Z_q^n}\overline{\chi_a(x)}\chi_b(x)
 =\prod_{j=1}^n\sum_{t\in\Z_q}\zeta_q^{(b_j-a_j)t}
 =q^n\,\boldsymbol{1}_{a=b}.
\end{equation}
Each factor is \(q\) if \(a_j=b_j\) and zero otherwise.
By Lemma~\ref{lem:graph-properties}\,\ref{item:graph-size}, the space
\(\C^{\Z_q^n}\) has dimension \(q^n\). The \(q^n\) normalized waves
therefore form an orthonormal basis, so they account for every eigenvalue
with its multiplicity. The eigenvalues are real because \(A\) is real
symmetric by Lemma~\ref{lem:adjacency-structure}.

Finally, the coordinate-permutation identity in
Lemma~\ref{lem:adjacency-structure} shows that \(S_{q,n}=N(0)\)
is preserved, since coordinate permutations fix the zero vertex.
Substituting \(s=\pi t\) in the eigenvalue sum gives
\begin{equation}
 \lambda(\pi a)
 =\sum_{t\in S_{q,n}}\zeta_q^{(\pi a)\cdot(\pi t)}
 =\sum_{t\in S_{q,n}}\zeta_q^{a\cdot t}
 =\lambda(a).\qedhere
\end{equation}
\end{proof}

The Fourier waves \(\chi_a\) depend only on the additive group
\(\Z_q^n\). They form an eigenbasis for the adjacency operator of
every Cayley graph on this group. The choice of allowed shifts
\(S_{q,n}\) determines the eigenvalue associated with each wave:
\begin{equation}
 \lambda(a)=\sum_{s\in S_{q,n}}\chi_a(s).
\end{equation}
Thus changing the allowed shifts preserves the Fourier eigenbasis
but can change its eigenvalues. In this basis, adjacency acts by
multiplying each Fourier wave by its eigenvalue, as illustrated in
Figure~\ref{fig:graph-fourier}.

\begin{figure}[htbp]
\centering
\includegraphics[width=\linewidth]{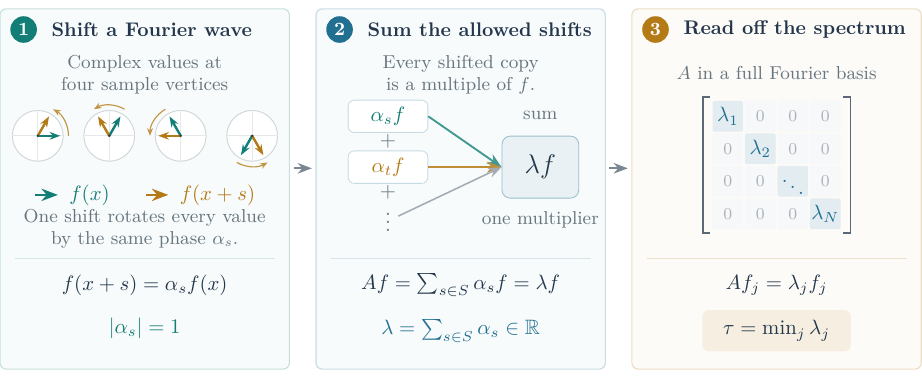}
\caption{Schematic Fourier reduction for the adjacency operator
\(A=A_{\Om}\), with \(S=S_{q,n}\).
Arrows show the complex values of a Fourier wave \(f\) at sample
vertices; a fixed shift multiplies them all by the same phase
\(\alpha_s\). Summing over allowed shifts gives \(Af=\lambda f\).
The \(N=q^n\) Fourier waves \(f_1,\ldots,f_N\) form an orthogonal
basis, so their multipliers give every eigenvalue.}
\label{fig:graph-fourier}
\end{figure}

The adjacency operator \(A\) sums a function's values over each
vertex's neighbors. Dividing by the degree \(D_{q,n}=|S_{q,n}|\)
turns this sum into an average. The corresponding eigenvalue of
\(A/D_{q,n}\) is the \emph{normalized eigenvalue}
\(\rho(a)=\lambda(a)/D_{q,n}\).
Thus \(\rho(a)\) describes how averaging over neighbors scales
the Fourier wave \(\chi_a\).

We also record the counts of the letters in the frequency index
\(a\) by its \emph{composition}:
\begin{equation}\label{eq:frequency-composition}
 \operatorname{comp}(a)=r=(r_0,\ldots,r_{q-1}),
 \qquad r_c=|\{j\in\{1,\ldots,n\}:a_j=c\}|.
\end{equation}
Here \(r_0\) counts the zero entries of \(a\), \(r_1\) counts
its entries equal to one, and so on. These counts are nonnegative
integers with \(\sum_{c=0}^{q-1}r_c=n\). Thus
\(\operatorname{comp}(a)\) is a \emph{weak composition} of \(n\)
into \(q\) parts. Throughout, we refer to weak compositions simply
as compositions. A composition is \emph{monochromatic} if
\(r_c=n\) for one letter \(c\), with all other counts zero.
The following lemma records the consequences of
Lemma~\ref{lem:fourier-diagonalization} for normalized eigenvalues
and frequency compositions.

\begin{lemma}[Normalized eigenvalues and frequency compositions]
\label{lem:normalized-compositions}
The normalized eigenvalues have the following properties.
\begin{enumerate}[label=\textup{(\roman*)},leftmargin=*,itemsep=4pt]
\item\label{item:normalized-wave}
For every frequency \(a\in\Z_q^n\),
\begin{equation}\label{eq:normalized-eigenvalue}
 \rho(a)=\frac{\lambda(a)}{D_{q,n}}
 =\frac1{|S_{q,n}|}\sum_{s\in S_{q,n}}\zeta_q^{a\cdot s}
 =\E_{X\ \mathrm{uniform\ on}\ S_{q,n}}\zeta_q^{a\cdot X}
 \in[-1,1],
\end{equation}
where \(X\) is chosen uniformly from \(S_{q,n}\), so each
allowed shift has probability \(1/D_{q,n}\). Moreover,
\begin{equation}\label{eq:neighbor-average-eigenvalue}
 \frac{A}{D_{q,n}}\chi_a=\rho(a)\chi_a.
\end{equation}
\item\label{item:composition-invariance}
For any \(a,b\in\Z_q^n\),
\begin{equation}\label{eq:composition-eigenvalue}
 \operatorname{comp}(a)=\operatorname{comp}(b)
 \quad\Longrightarrow\quad \rho(a)=\rho(b).
\end{equation}
\end{enumerate}
\end{lemma}
\begin{proof}
\emph{\ref{item:normalized-wave}.}
By Lemma~\ref{lem:graph-properties}\,\ref{item:graph-degree},
\(D_{q,n}=|S_{q,n}|>0\). Dividing the eigenvalue formula and
eigenvector relation in Lemma~\ref{lem:fourier-diagonalization}
by \(D_{q,n}\) gives the stated average and
\eqref{eq:neighbor-average-eigenvalue}. The expectation is the
same average written using the uniform distribution on \(S_{q,n}\).
Lemma~\ref{lem:fourier-diagonalization} also gives
\(\rho(a)\in\mathbb R\), and the triangle inequality yields
\begin{equation}
 |\rho(a)|\leq\frac1{|S_{q,n}|}
              \sum_{s\in S_{q,n}}|\zeta_q^{a\cdot s}|=1.
\end{equation}
\emph{\ref{item:composition-invariance}.}
If \(a\) and \(b\) have the same composition, a coordinate
permutation \(\pi\in\mathfrak S_n\) satisfies \(b=\pi a\).
Lemma~\ref{lem:fourier-diagonalization} gives
\(\lambda(b)=\lambda(a)\). Dividing by \(D_{q,n}\) proves
\eqref{eq:composition-eigenvalue}.
\end{proof}

By Lemma~\ref{lem:fourier-diagonalization}, every adjacency
eigenvalue has the form \(\lambda(a)=D_{q,n}\rho(a)\). Since
\(D_{q,n}>0\), finding the least eigenvalue amounts to finding
the smallest \(\rho(a)\) and then multiplying by the degree.
Lemma~\ref{lem:normalized-compositions}\,\ref{item:normalized-wave}
expresses \(\rho(a)\) as an average over allowed shifts. Each
shift contributes a complex phase of absolute value one, so this
average records how the phases reinforce or cancel one another.
Normalization places these real averages on the fixed scale
\([-1,1]\), independently of the degree.

Part~\ref{item:composition-invariance} of the same lemma lets us
organize the averages by composition. There are
\(\binom{n+q-1}{q-1}\) weak compositions of \(n\) into \(q\)
parts~\cite[Sec.~1.2]{stanley}, so one frequency representative
from each composition class suffices to find the least eigenvalue.
Different classes may share an eigenvalue. The dependence on letter
counts will also allow us to express these averages through
polynomial coefficients in the next subsection.

We first evaluate the frequency indices whose entries are all equal.

\begin{lemma}[Constant frequencies]\label{lem:constant-frequencies}
A frequency \(a\in\Z_q^n\) has monochromatic composition if and
only if \(a=c\one=(c,\ldots,c)\) for some \(c\in\Z_q\).
We call such an \(a\) a \emph{constant frequency}. Its wave
\(\chi_{c\one}(x)=\zeta_q^{c\sum_{j=1}^n x_j}\) is constant
on the entire vertex set if and only if \(c=0\).
Every constant frequency has normalized eigenvalue
\begin{equation}\label{eq:constant-frequency}
 \rho(c\one)=1.
\end{equation}
\end{lemma}
\begin{proof}
By the definition of \(r_c\), the equality \(r_c=n\) holds
exactly when every entry of \(a\) equals \(c\). The formula for
\(\chi_{c\one}\) follows from the definition of the Fourier wave.
For \(c=0\), it is identically one. If \(c\neq0\), its values
at the zero vertex and at \(e_1=(1,0,\ldots,0)\) are
\(1\) and \(\zeta_q^c\neq1\), respectively, so it is not
constant on the vertex set.

Every allowed shift \(s\in S_{q,n}\) contains each residue exactly
\(\ell\) times. Since \((q-1)\ell\) is even,
\begin{equation}
 \sum_{j=1}^n s_j=\frac{\ell q(q-1)}2=0\quad\text{in }\Z_q.
\end{equation}
Thus \(\zeta_q^{(c\one)\cdot s}=1\) for every allowed shift.
Averaging as in
Lemma~\ref{lem:normalized-compositions}\,\ref{item:normalized-wave}
gives \eqref{eq:constant-frequency}.
\end{proof}

Geometrically, the coordinate sum
\(\sigma(x)=\sum_jx_j\in\Z_q\) partitions the vertices into
\(q\) classes, and every edge stays within one class. The wave
\(\chi_{c\one}(x)=\zeta_q^{c\sigma(x)}\) is constant on
each class, so averaging its values over the neighbors of a vertex
leaves its value unchanged. This explains why its normalized
eigenvalue is \(1\), even when the wave is not globally constant.

These indices contribute only the eigenvalue \(D_{q,n}>0\), so
they already satisfy the required lower bound. The remaining
estimate concerns nonmonochromatic compositions, whose frequency
indices contain at least two different letters.

\begin{lemma}[Spectral bound for nonmonochromatic compositions]
\label{lem:nonmonochromatic-bound}
For every frequency \(a\in\Z_q^n\) with nonmonochromatic
composition,
\begin{equation}\label{eq:nonconstant}
 |\rho(a)|\leq\frac1{n-1}.
\end{equation}
Consequently,
\begin{equation}
 \lambda(a)=D_{q,n}\rho(a)\geq-\frac{D_{q,n}}{n-1}.
\end{equation}
\end{lemma}

The coefficient reduction in the next subsection and the switching
argument in Section~\ref{sec:switching} prove
Lemma~\ref{lem:nonmonochromatic-bound}.
Section~\ref{sec:spectral} then exhibits a frequency with normalized
eigenvalue \(-1/(n-1)\), completing the spectral proof.
Example~\ref{ex:frequency-composition} in Appendix~\ref{app:examples}
illustrates frequency compositions.

\subsection{Fourier analysis in the symmetric Schur--Weyl sector}\label{sec:fourier-counts}
Section~\ref{sec:fourier-waves} expressed the normalized eigenvalues
as phase averages depending only on frequency compositions. We now
realize these averages as Fourier matrix elements in the symmetric
Schur--Weyl sector. Polynomial coordinates on this representation
yield the circulant-coefficient formula and reduce the bound in
Lemma~\ref{lem:nonmonochromatic-bound} to a counting problem.

Let \(V=\C^q\), with standard basis \((e_c)_{c\in\Z_q}\), and
write \(|a\rangle=e_{a_1}\otimes\cdots\otimes e_{a_n}\).
The actions of \(\mathrm{GL}(V)\) by tensor powers \(g^{\otimes n}\)
and of \(\mathfrak S_n\) by tensor-factor permutations \(U_\pi\)
commute. Schur--Weyl duality decomposes their joint representation as
\begin{equation}
 V^{\otimes n}\cong
 \bigoplus_{\substack{\nu\vdash n\\\nu_{q+1}=0}}
       \mathbb S_\nu(V)\otimes[\nu],
\end{equation}
where the sum runs over partitions of \(n\) with at most \(q\)
parts, \(\mathbb S_\nu(V)\) is the corresponding irreducible
polynomial representation of \(\mathrm{GL}(V)\), and \([\nu]\)
is the corresponding irreducible representation of
\(\mathfrak S_n\)~\cite[App.~E, Proposition~E.8]{zkfb}.
For the single-row partition \(\nu=(n)\), the latter representation
is trivial and \(\mathbb S_{(n)}(V)=\operatorname{Sym}^n(V)\).
Thus the symmetric sector is
\begin{equation}
 \operatorname{Sym}^n(V)\cong(V^{\otimes n})^{\mathfrak S_n},
 \qquad
 \Pi_{\mathrm{sym}}=\frac1{n!}\sum_{\pi\in\mathfrak S_n}U_\pi.
\end{equation}
We use this identification throughout and write
\(\operatorname{Sym}^n(g)\) for the restriction of
\(g^{\otimes n}\) to this subspace. The averaging operator
\(\Pi_{\mathrm{sym}}\) is its orthogonal
projector~\cite[Sec.~1, Proposition~1]{harrow-sym}.

The coordinate-permutation orbits of the basis vectors \(|a\rangle\)
are indexed by their compositions. An orbit of composition \(r\)
has size
\begin{equation}
 M_r:=\mult n{r_0,\ldots,r_{q-1}}
     =\frac{n!}{r_0!\cdots r_{q-1}!}.
\end{equation}
Its normalized orbit sum is the occupation-number vector
\begin{equation}\label{eq:occupation-basis}
 |r\rangle_{\mathrm{sym}}
 =\frac1{\sqrt{M_r}}\sum_{\operatorname{comp}(a)=r}|a\rangle,
 \qquad
 \Pi_{\mathrm{sym}}|a\rangle
 =\frac1{\sqrt{M_r}}|r\rangle_{\mathrm{sym}}
 \quad\text{if }\operatorname{comp}(a)=r.
\end{equation}
These vectors form an orthonormal basis of
\(\operatorname{Sym}^n(V)\)~\cite[Sec.~1, Theorem~3]{harrow-sym}.\footnote{%
\(\operatorname{Sym}^n(V)\) is the \emph{Bose-symmetric subspace};
\(r_c\) is the occupation number of \(e_c\).
Bose symmetry of a density operator,
\(\Pi_{\mathrm{sym}}\sigma\Pi_{\mathrm{sym}}=\sigma\), is stronger
than permutation invariance,
\(U_\pi\sigma U_\pi^\dagger=\sigma\) for every \(\pi\);
see~\cite[Sec.~II.B]{ckmr} and~\cite[Sec.~5.2, Definitions~5.1--5.2
and Proposition~5.5]{zkfb}.
This subspace organizes the phase averages; the protocol allows
arbitrary shared states.}
They are a weight basis for its \(\mathrm{GL}(V)\) action: a
diagonal matrix \(\operatorname{diag}(t_0,\ldots,t_{q-1})\)
multiplies \(|r\rangle_{\mathrm{sym}}\) by
\(\prod_c t_c^{r_c}\). Thus compositions label weights within
the single sector \((n)\)~\cite[Sec.~6, Eq.~(6.7)]{zkfb}.

To express the phase averages in the occupation basis, let
\(\mathbf b=(\ell,\ldots,\ell)\) be the balanced composition.
The words of this composition form \(S_{q,n}\), so
\(M_{\mathbf b}=D_{q,n}\).
Define the unnormalized Fourier matrix \(\mathcal F_q\) by
\begin{equation}
 \mathcal F_q e_t=\sum_{c=0}^{q-1}\zeta_q^{ct}e_c.
\end{equation}
The matrix \(q^{-1/2}\mathcal F_q\) is unitary. We keep the
unnormalized convention so that
\(\langle a|\mathcal F_q^{\otimes n}|x\rangle=\zeta_q^{a\cdot x}\)
is exactly the phase from Section~\ref{sec:fourier-waves}.

Polynomial coordinates give a second description of this action.
Identify the symmetric algebra \(\operatorname{Sym}(V)\) with
\(\C[z_0,\ldots,z_{q-1}]\) by sending \(e_c\) to the commuting
generator \(z_c\). Its degree-\(n\) component is
\(\operatorname{Sym}^n(V)\). The quotient map sends
\(|a\rangle\) to \(\prod_j z_{a_j}=z^r\), where
\(z^r=z_0^{r_0}\cdots z_{q-1}^{r_{q-1}}\).
Under the averaging identification in~\eqref{eq:occupation-basis},
\begin{equation}\label{eq:monomial-occupation}
 z^r\longleftrightarrow M_r^{-1/2}|r\rangle_{\mathrm{sym}}.
\end{equation}
The symmetric-power action is linear substitution on these
generators: for \(g e_t=\sum_c g_{ct}e_c\),
\begin{equation}
 \operatorname{Sym}^n(g)z^r
 =\prod_{t=0}^{q-1}
       \left(\sum_{c=0}^{q-1}g_{ct}z_c\right)^{r_t}.
\end{equation}
In particular, applying the Fourier matrix to the balanced monomial
defines the generating polynomial
\begin{equation}\label{eq:phase-generating-polynomial}
 P(z):=\operatorname{Sym}^n(\mathcal F_q)z^{\mathbf b},
 \qquad z^{\mathbf b}=z_0^\ell\cdots z_{q-1}^\ell.
\end{equation}
The substitution formula gives
\begin{equation}\label{eq:balanced-phase-product}
 P(z)=\left(\prod_{t=0}^{q-1}
              \sum_{c=0}^{q-1}z_c\zeta_q^{ct}\right)^\ell.
\end{equation}
For a polynomial \(Q(z)\), write \([z^r]Q(z)\) for the
coefficient of \(z^r\). The next lemma identifies the phase averages
with symmetric-power matrix elements and with coefficients of \(P\).

\begin{lemma}[Symmetric-power representation of the normalized eigenvalues]
\label{lem:symmetric-fourier}
With the notation above, for every frequency \(a\in\Z_q^n\) of
composition \(r\),
\begin{equation}\label{eq:symmetric-fourier-eigenvalue}
 \rho(a)=
 \frac{{}_{\mathrm{sym}}\langle r|
       \operatorname{Sym}^n(\mathcal F_q)
       |\mathbf b\rangle_{\mathrm{sym}}}
      {\sqrt{M_rD_{q,n}}}
 =\frac{[z^r]P(z)}{M_r}.
\end{equation}
\end{lemma}
\begin{proof}
Expanding the normalized occupation vectors and applying
Lemma~\ref{lem:normalized-compositions} gives
\begin{equation}
 \begin{aligned}
 {}_{\mathrm{sym}}\langle r|
       \operatorname{Sym}^n(\mathcal F_q)
       |\mathbf b\rangle_{\mathrm{sym}}
 &=\frac1{\sqrt{M_rD_{q,n}}}
   \sum_{\operatorname{comp}(a')=r}\sum_{x\in S_{q,n}}
          \zeta_q^{a'\cdot x}\\
 &=\sqrt{M_rD_{q,n}}\,\rho(a).
 \end{aligned}
\end{equation}
Indeed, each inner sum is \(D_{q,n}\rho(a')\), and
\(\rho(a')=\rho(a)\) throughout the composition class.
This proves the first equality. For the second, write
\(P(z)=\sum_s p_s z^s\), where the sum runs over compositions
of \(n\) into \(q\) parts and \(p_s=[z^s]P(z)\).
By~\eqref{eq:monomial-occupation}, its occupation-vector representation
is \(\sum_s p_s M_s^{-1/2}|s\rangle_{\mathrm{sym}}\).
The balanced input monomial similarly corresponds to
\(D_{q,n}^{-1/2}|\mathbf b\rangle_{\mathrm{sym}}\), since
\(M_{\mathbf b}=D_{q,n}\). Applying the Fourier action therefore gives
\begin{equation}
 \sum_s\frac{p_s}{\sqrt{M_s}}|s\rangle_{\mathrm{sym}}
 =\frac1{\sqrt{D_{q,n}}}
       \operatorname{Sym}^n(\mathcal F_q)
       |\mathbf b\rangle_{\mathrm{sym}}.
\end{equation}
Taking the inner product with \({}_{\mathrm{sym}}\langle r|\)
and using orthonormality yields
\begin{equation}
 \frac{p_r}{\sqrt{M_r}}
 =\frac{{}_{\mathrm{sym}}\langle r|
       \operatorname{Sym}^n(\mathcal F_q)
       |\mathbf b\rangle_{\mathrm{sym}}}{\sqrt{D_{q,n}}}.
\end{equation}
Multiplying by \(\sqrt{M_r}\) and using the first equality proves
\begin{equation}\label{eq:phase-coefficient-sum}
 [z^r]P(z)
 =\sqrt{\frac{M_r}{D_{q,n}}}\,
       {}_{\mathrm{sym}}\langle r|
       \operatorname{Sym}^n(\mathcal F_q)
       |\mathbf b\rangle_{\mathrm{sym}}
 =M_r\rho(a).\qedhere
\end{equation}
\end{proof}

Thus coefficient extraction computes the desired Fourier matrix
element, and division by \(M_r\) gives the normalized eigenvalue.

The \(q\) linear factors in~\eqref{eq:balanced-phase-product} are
the Fourier eigenvalues of the \emph{circulant} matrix
\begin{equation}
 C_q(z)=(z_{j-i})_{i,j\in\Z_q}
 =\begin{pmatrix}
 z_0 & z_1 & \cdots & z_{q-1}\\
 z_{q-1} & z_0 & \cdots & z_{q-2}\\
 \vdots & \vdots & \ddots & \vdots\\
 z_1 & z_2 & \cdots & z_0
 \end{pmatrix}.
\end{equation}
All subscripts are taken modulo \(q\); each row is the preceding
row shifted one position to the right. Indeed, the vector
\((\zeta_q^{it})_{i\in\Z_q}\) is an eigenvector with eigenvalue
\(\sum_c z_c\zeta_q^{ct}\). These Fourier vectors form a basis,
so the product of these \(q\) eigenvalues is \(\det C_q(z)\).
Equation~\eqref{eq:balanced-phase-product} therefore gives
\begin{equation}\label{eq:phase-circulant-determinant}
 P(z)=\bigl(\det C_q(z)\bigr)^\ell.
\end{equation}
The determinant therefore gives an explicit expression for the
symmetric-power Fourier image \(P(z)\).

To bound the absolute values of these coefficients, we use the
\emph{permanent}. For a \(q\times q\) matrix \(M\), define
\begin{equation}
 \per M=\sum_{\sigma\in\mathfrak S_q}
              \prod_{i=1}^q M_{i,\sigma(i)},
\end{equation}
where \(\mathfrak S_q\) is the permutation group on
\(\{1,\ldots,q\}\). Each product selects one entry from every row
and column. The determinant attaches a permutation sign to each
selection; the permanent omits these signs. Set
\begin{equation}
 H_{q,\ell}(z)=\bigl(\per C_q(z)\bigr)^\ell.
\end{equation}
Both \(P(z)\) and \(H_{q,\ell}(z)\) belong to
\(\operatorname{Sym}^n(\C^q)\), since the determinant and permanent
are homogeneous of degree \(q\) and \(q\ell=n\).
The coefficient \([z^r]H_{q,\ell}(z)\) counts selections of one
entry from every row and column in each of \(\ell\) numbered
copies of \(C_q(z)\), with \(z_c\) selected \(r_c\) times in
total. The row indices of these copies give \(\ell\) blocks of
\(q\) auxiliary positions. The composition constrains the total
label counts across all blocks, with no prescribed counts within
an individual block.

The following lemma combines the determinant-coefficient identity
of Cao et al.~\cite[Eq.~(2)]{cao}, normalized by the graph degree,
with the first inequality in their permanent bound~\cite[Eq.~(4)]{cao}.
We include a self-contained derivation.

\begin{lemma}[Reduction to circulant coefficients]\label{lem:coefficient-reduction}
For every frequency \(a\in\Z_q^n\) with composition
\(r=(r_0,\ldots,r_{q-1})\), the normalized eigenvalue satisfies
\begin{equation}\label{eq:spectrum}
 \rho(a)=\frac{[z^r](\det C_q(z))^\ell}
                  {\mult n{r_0,\ldots,r_{q-1}}}.
\end{equation}
Consequently,
\begin{equation}
 \frac{|\lambda(a)|}{D_{q,n}}=|\rho(a)|
 \leq\frac{[z^r]H_{q,\ell}(z)}{\mult n{r_0,\ldots,r_{q-1}}}.
\end{equation}
\end{lemma}
\begin{proof}
Lemma~\ref{lem:symmetric-fourier} gives
\([z^r]P(z)=M_r\rho(a)\), and
\eqref{eq:phase-circulant-determinant} identifies \(P(z)\) with
\((\det C_q(z))^\ell\). Substituting the multinomial expression
for \(M_r\) proves~\eqref{eq:spectrum}. Finally, expanding each determinant as a signed sum
of permutation monomials and taking absolute values coefficient by
coefficient gives
\begin{equation}
 \bigl|[z^r](\det C_q(z))^\ell\bigr|
 \leq [z^r](\per C_q(z))^\ell.
\end{equation}
Dividing by the multinomial coefficient proves the bound.
\end{proof}

To interpret the two numerators, split the positions of each word
of composition \(r\) into \(\ell\) fixed blocks of \(q\) positions, numbered
\(0,\ldots,q-1\) within each block. Write \(w(b,i)\) for the
letter at position \(i\) in block \(b\). This letter selects
column \(i+w(b,i)\pmod q\) in row \(i\), since
\(C_q(z)_{i,i+w(b,i)}=z_{w(b,i)}\). Call an assignment
\emph{valid} if the selected columns are all distinct within every
block. Equivalently, each map \(i\mapsto i+w(b,i)\pmod q\)
must be a permutation of \(\Z_q\). This is precisely the condition
for selecting one entry from every row and column in each matrix
copy. The determinant coefficient is the signed count of valid assignments:
each contributes the product of the signs of its block permutations.
The permanent coefficient counts the same assignments with every
contribution equal to \(+1\). The common denominator counts all
words of composition \(r\). Thus the normalized eigenvalue is the
signed count divided by the class size, while the fraction of valid
assignments bounds its absolute value. The letter counts are prescribed
across all blocks; validity imposes no balance requirement within a
block. Example~\ref{ex:coefficient-counting}
in Appendix~\ref{app:example-coefficients} applies this test to all
six words of composition \((1,1,1)\), and
Figure~\ref{fig:composition-coefficients} shows the three valid
matrix selections.

Section~\ref{sec:switching} proves that, within every nonmonochromatic
composition class, at most a fraction \(1/(n-1)\) of the words are
valid. This gives the bound required in
Lemma~\ref{lem:nonmonochromatic-bound}.

\section{A switching bound for circulant coefficients}\label{sec:switching}
We now prove the main combinatorial estimate: within every
nonmonochromatic composition class, at most a fraction \(1/(n-1)\)
of the assignments are valid. Together with
Lemma~\ref{lem:coefficient-reduction}, this gives the spectral bound
of Lemma~\ref{lem:nonmonochromatic-bound}.

\begin{theorem}[Uniform coefficient bound]\label{thm:permanent}
For every nonmonochromatic composition \(r\) of \(n=q\ell\),
\begin{equation}\label{eq:permanent}
 [z^r]H_{q,\ell}(z)\leq
 \frac{1}{n-1}\mult n{r_0,\ldots,r_{q-1}}.
\end{equation}
\end{theorem}

For \(q=3\), this coefficient bound was proved by Luo, Ning, and
Zhang~\cite[Theorem~2, Eqs.~(29)--(30)]{lnz}.
We give a self-contained switching proof covering every cyclic
alphabet and all nonmonochromatic
compositions.\footnote{Exact coefficients involving only two letters follow from the
formulas of Donovan, Johnson, and
Wanless~\cite[Theorem~5.1]{djw}. When every label occurs once
(\(\ell=1\)), the coefficient counts transversals of a cyclic Latin
square, and the required bound follows from established transversal
bounds~\cite[Theorem~8 and Eq.~(8)]{mmw}.}

The key idea is that there are \emph{many ways to destroy validity,
but only a few ways to restore it}. Fix a nonmonochromatic composition,
and let \(V\) and \(I\) be its numbers of valid and invalid
assignments. Imagine placing the valid words on the left and the
invalid words on the right. Draw a connection whenever a permitted
swap of two unequal letters takes a valid word to an invalid word.
Swapping preserves the letter counts, so we stay within the same
composition class. Reversing the swap restores validity; we call
such a swap a \emph{repair}.

Now count the same connections from both sides. From the valid side,
count the swaps that destroy validity; from the invalid side, count
the swaps that restore it. Several valid words can lead to the same
invalid word. The bound on the number of repair swaps controls
exactly how much this can happen.

In the main case, choose a letter that occurs \(k\) times, where
\(2\leq k\leq n-2\), and permit only swaps of that letter
with a different letter. The proof establishes two facts:
\begin{itemize}
\item Each valid word has exactly \(k(n-k)\geq2(n-2)\)
permitted swaps, and every one destroys validity.
\item Each invalid word has at most two permitted repairs. This
is the main technical point: a swap creates a specific pattern of
repeated and missing targets, and a repair must fix that pattern.
\end{itemize}
Figure~\ref{fig:switching-counting} illustrates this comparison.

\begin{figure}[htbp]
\centering
\includegraphics[width=\linewidth]{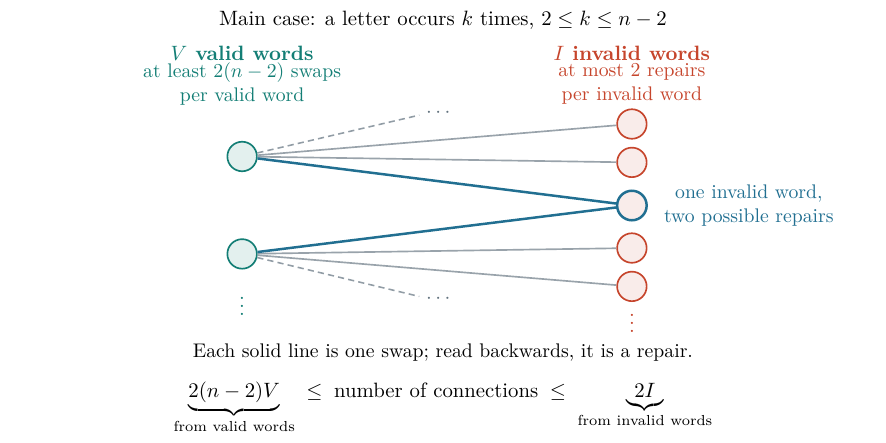}
\caption{Counting the same swap connections from both sides in the
main case. Each solid line represents a permitted swap involving the
chosen letter. The two highlighted lines share an invalid endpoint,
showing how two valid words can lead to the same invalid word.
Only selected words and swaps are shown; dashed continuations and
ellipses indicate omissions.}
\label{fig:switching-counting}
\end{figure}

Counting the connections from the two sides therefore gives
\begin{equation}
 2(n-2)V\leq k(n-k)V\leq2I,
 \qquad\text{hence}\qquad I\geq(n-2)V.
\end{equation}
There are thus at least \(n-2\) invalid words for every valid
word. The total number of words satisfies
\(V+I\geq(n-1)V\), so
\begin{equation}
 \frac{V}{V+I}\leq\frac1{n-1}.
\end{equation}
This also covers \(V=0\), since the composition class is nonempty.
For \(n>3\), the inequality forces invalid words to outnumber
valid words.

The proof treats the remaining compositions separately. When every
letter occurs once, it uses all swaps and at most three repairs to
obtain the same bound. When one letter occurs \(n-1\) times,
there are no valid words. We therefore never need to determine
\(V\) exactly: counting the ways to move between valid and invalid
words already bounds their relative numbers.

\FloatBarrier
\begin{proof}[Proof of Theorem~\ref{thm:permanent}]
Recall that \(q\geq2\), \(\ell\geq1\), \(n=q\ell\), and
\((q-1)\ell\) is even under our standing assumptions.

\emph{1. Assignments and swaps.}
Fix a nonmonochromatic composition \(r\) of \(n=q\ell\).
Let
\begin{equation}
 \begin{aligned}
 X&=\{1,\ldots,\ell\}\times\Z_q,\qquad |X|=n,\\
 W_r&=\{w:X\to\Z_q:\ |w^{-1}(\{c\})|=r_c
                    \text{ for every }c\in\Z_q\}.
 \end{aligned}
\end{equation}
For any assignment \(w:X\to\Z_q\), define its target map by
\begin{equation}
 F_w(b,i)=(b,i+w(b,i)).
\end{equation}
All additions and subtractions within a block are in \(\Z_q\).
Let \(E_r=\{w\in W_r:F_w\text{ is bijective}\}\) and
\(I_r=W_r\setminus E_r\) be the valid and invalid assignments,
respectively. Choosing a permutation \(\pi_b\) of \(\Z_q\)
in each block is equivalent to choosing an assignment with bijective
target map, via \(w(b,i)=\pi_b(i)-i\). The corresponding monomial
in the permanent expansion records the total letter counts, so
extracting \([z^r]\) restricts this bijection to \(E_r\).
Together with the multinomial count of all assignments, this gives
\begin{equation}\label{eq:count}
 |W_r|=\mult n{r_0,\ldots,r_{q-1}},\qquad
 |E_r|=[z^r]H_{q,\ell}(z).
\end{equation}
In particular, \(|W_r|>0\), and it suffices to prove
\(|E_r|/|W_r|\leq1/(n-1)\).

Write \(\binom X2\) for the set of unordered pairs of distinct
positions. For \(p=\{u,v\}\in\binom X2\), let \(\tau_p\)
exchange \(u,v\) and fix every other position, and put
\(w^p=w\circ\tau_p\). Thus \(w^p\in W_r\) and
\((w^p)^p=w\). We consider only pairs for which \(w(u)\ne w(v)\).
A \emph{repair} of \(w\in I_r\) is such a pair \(p\) with
\(w^p\in E_r\). Pairs may involve positions in different blocks.

\emph{2. The defects caused by a swap.}
For \(c\in\Z_q\), define the fixed shift
\begin{equation}
 T_c(b,i)=(b,i+c).
\end{equation}
This is a permutation of \(X\), and
\(F_w(u)=T_{w(u)}(u)\). For each fixed \(u\), the targets
\(T_c(u)\) are distinct as \(c\) varies.

Let \(w\in E_r\), \(p=\{u,v\}\), and
\(A=w(u)\ne w(v)=B\). Denote the old and new targets by
\begin{equation}
 \alpha=T_A(u),\quad \beta=T_B(v),\qquad
 \gamma=T_B(u),\quad \delta=T_A(v).
\end{equation}
Bijectivity of \(F_w\) gives \(\alpha\ne\beta\).
Moreover,
\(\{\alpha,\beta\}\cap\{\gamma,\delta\}=\varnothing\):
for example, \(\gamma\ne\alpha\) since \(A\ne B\), and
\(\gamma\ne\beta\) since \(T_B\) is injective and \(u\ne v\).
The two comparisons involving \(\delta\) follow in the same way.
Only the targets of \(u,v\) change. Hence \(\alpha,\beta\)
have no preimages under \(F_{w^p}\). Each distinct new target
\(\gamma\) or \(\delta\) already had a unique preimage under
the bijection \(F_w\). That preimage lies outside \(\{u,v\}\),
because the old targets of \(u,v\) were \(\alpha,\beta\), and
therefore remains unchanged by the swap.
If \(\gamma\ne\delta\), the swap adds one preimage to each,
so each has exactly two preimages. If \(\gamma=\delta\), the swap
adds both \(u\) and \(v\) as preimages of their common target,
which therefore has exactly three.
Every remaining target has exactly one preimage. In particular,
\(w^p\in I_r\).

It follows also that every invalid assignment admitting a repair
has exactly one of these two defect patterns: apply the preceding
argument to the valid assignment obtained by that repair and reverse
the swap. When two targets each have two preimages, any repair must
select one position from each of these two preimage sets. When one
target has three preimages, it must select two of those three
positions. Otherwise, an unchanged pair of positions would still
have the same target after the proposed repair.

\emph{3. The double-counting identity.}
Choose either all swaps of unequal letters or only swaps exchanging
a specified letter \(c\) with a different letter. For \(w\in W_r\),
let \(\mathcal P(w)\subseteq\binom X2\) be the pairs permitted
by the chosen rule. Both rules are preserved by reversing a swap:
\begin{equation}
 p\in\mathcal P(w)\quad\Longleftrightarrow\quad
 p\in\mathcal P(w^p).
\end{equation}
By Step~2, every permitted swap from \(E_r\) ends in \(I_r\).
The map \((w,p)\mapsto(w^p,p)\), whose inverse is the same map,
therefore bijects permitted swaps from valid assignments with
permitted repairs of invalid assignments. Consequently,
\begin{equation}\label{eq:switching-count}
 \sum_{w\in E_r}|\mathcal P(w)|
 =\sum_{v\in I_r}
   |\{p\in\mathcal P(v):v^p\in E_r\}|.
\end{equation}
We count swaps together with their unordered position pairs, so each
swap occurs once on each side of this identity. If every valid
assignment has \(d\) permitted swaps and every invalid assignment
has at most \(R>0\) permitted repairs, then
\begin{equation}\label{eq:switching-density}
 d|E_r|\leq R|I_r|=R(|W_r|-|E_r|),\qquad
 \frac{|E_r|}{|W_r|}\leq\frac{R}{d+R}.
\end{equation}

\emph{4. The composition cases.}
We distinguish the following three exhaustive cases.

\emph{Case 1: a letter occurs between \(2\) and \(n-2\) times.}
Suppose that \(r_c=k\) with \(2\leq k\leq n-2\).
Permit only swaps exchanging \(c\) with a different letter.
There are exactly \(d=k(n-k)\) permitted pairs for every assignment:
choose one of the \(k\) positions carrying \(c\) and one of the
\(n-k\) other positions.

Fix \(v\in I_r\). If it has no permitted repair, the required
bound is immediate. Otherwise, Step~2 applies. Each preimage set of
\(F_v\) contains at most one position carrying \(c\), because all
such positions are mapped by the injective map \(T_c\).
If two targets each have two preimages, let
\(\varepsilon,\delta\in\{0,1\}\) count the positions carrying
\(c\) in their respective preimage sets. A permitted repair must
choose one position from each set, exactly one of which carries
\(c\). The number of candidate pairs is therefore
\begin{equation}
 \varepsilon(2-\delta)+\delta(2-\varepsilon)\leq2.
\end{equation}
If one target has three preimages, a permitted repair must choose
two of them. At most one carries \(c\), leaving at most two
candidate pairs. Thus every invalid assignment has at most two
permitted repairs; no claim that every candidate pair succeeds is
needed. Taking \(R=2\) in~\eqref{eq:switching-density} gives
\begin{equation}\label{eq:two}
 \begin{aligned}
 k(n-k)|E_r|&\leq2\bigl(|W_r|-|E_r|\bigr),\\
 \frac{|E_r|}{|W_r|}&\leq\frac2{k(n-k)+2}.
 \end{aligned}
\end{equation}
Since
\begin{equation}
 k(n-k)-2(n-2)=(k-2)(n-k-2)\geq0,
\end{equation}
we obtain
\begin{equation}
 \frac{|E_r|}{|W_r|}
 \leq\frac2{2(n-2)+2}=\frac1{n-1}.
\end{equation}

\emph{Case 2: a letter occurs \(n-1\) times.}
Let this letter be \(c\), and take any \(w\in W_r\).
Subtracting \(c\) from every letter gives
\(F_{w-c}=T_{-c}\circ F_w\), so it preserves validity.
Since the composition is nonmonochromatic, \(w-c\) is zero at
all but one position \(u\), where it is nonzero. Thus
\(F_{w-c}\) fixes every position except \(u\). Its target
\(v=F_{w-c}(u)\ne u\) also satisfies \(F_{w-c}(v)=v\),
so \(F_{w-c}\) is not injective and \(w\) is invalid.
Hence \(E_r=\varnothing\), and the coefficient bound holds.

\emph{Case 3: neither of the preceding cases applies.}
Nonmonochromaticity excludes an entry \(r_c=n\), so every positive
entry of \(r\) must be one. Since there are only \(q\) letters,
this gives \(n\leq q\); combined with \(n=q\ell\geq q\), it
forces \(\ell=1\), \(n=q\), and
\(r=(1,\ldots,1)\). Since \((q-1)\ell\) is even, \(n=q\)
is odd. There is now one block, so we identify its positions with
\(\Z_n=\Z_q\).

Permit all swaps. Every assignment in \(W_r\) has distinct letters,
so it admits \(d=\binom n2\) permitted pairs. We prove that each
invalid assignment has at most three repairs. An assignment without
a repair already satisfies this bound. Otherwise, Step~2 applies.
If one target has three preimages, a repair must exchange letters
at two of them, giving at most \(\binom32=3\) choices.

For the other defect pattern, write \(U,V\in\Z_q\) for the
two targets with two preimages, whose respective preimage sets are
\(\{p_1,p_2\}\) and \(\{s_1,s_2\}\). Let \(H_1\ne H_2\)
be the two missing targets. A repair must exchange letters at some
\(p_i\) and \(s_j\), giving four candidate pairs. The two
preimage sets are disjoint because \(U\ne V\). Since the letters at these positions
are \(U-p_i\) and \(V-s_j\), their new targets after a swap are
\begin{equation}
 V+(p_i-s_j),\qquad U-(p_i-s_j).
\end{equation}
A repair must send both changed positions to the two missing targets;
in particular, a necessary condition is
\begin{equation}
 \delta_{ij}:=p_i-s_j\in
 \{H_1-V,H_2-V\}=\{h,k\},\qquad h\ne k.
\end{equation}
If all four pairs repaired the assignment, all four \(\delta_{ij}\)
would belong to this two-element set. Each row and each column of
\((\delta_{ij})\) has distinct entries, because the positions within
either pair are distinct. Up to exchanging \(h,k\), the matrix
would consequently be
\begin{equation}
 \begin{pmatrix}h&k\\k&h\end{pmatrix}.
\end{equation}
But
\(\delta_{11}+\delta_{22}=\delta_{12}+\delta_{21}\)
would then give \(2h=2k\) in \(\Z_q\).
Since \(q\) is odd, multiplication by two is injective, contradicting
\(h\ne k\). At most three of the four pairs can therefore repair the
assignment.

Applying~\eqref{eq:switching-density} with
\(d=\binom n2=n(n-1)/2\) and \(R=3\) yields
\begin{equation}\label{eq:three}
 \begin{aligned}
 d|E_r|&\leq3\bigl(|W_r|-|E_r|\bigr),\\
 \frac{|E_r|}{|W_r|}
 &\leq\frac3{d+3}
 =\frac3{n(n-1)/2+3}\leq\frac1{n-1}.
 \end{aligned}
\end{equation}
The last inequality is equivalent to
\((n-3)(n-4)\geq0\), which holds for every odd integer \(n\geq3\).
All compositions have now been covered, proving
\eqref{eq:permanent} by \eqref{eq:count}.
\end{proof}

The coefficient bound now completes the proof of the spectral
estimate stated in Section~\ref{sec:fourier-waves}.

\begin{proof}[Proof of Lemma~\ref{lem:nonmonochromatic-bound}]
Let \(a\in\Z_q^n\) have nonmonochromatic composition
\(r=(r_0,\ldots,r_{q-1})\). Lemma~\ref{lem:coefficient-reduction}
and Theorem~\ref{thm:permanent} give
\begin{equation}
 |\rho(a)|\leq
 \frac{[z^r]H_{q,\ell}(z)}{\mult n{r_0,\ldots,r_{q-1}}}
 \leq\frac1{n-1},
\end{equation}
which proves \eqref{eq:nonconstant}.
By Lemma~\ref{lem:normalized-compositions}\,\ref{item:normalized-wave},
\(\rho(a)\) is real, so \(\rho(a)\geq-1/(n-1)\).
Multiplying by the positive degree \(D_{q,n}\) gives the stated
lower bound for \(\lambda(a)=D_{q,n}\rho(a)\).
\end{proof}

\section{The least adjacency eigenvalue}\label{sec:spectral}
We now show that the spectral lower bound is attained, completing
the proof of Theorem~\ref{thm:main}.

\begin{proof}[Completion of the proof of Theorem~\ref{thm:main}]
We prove the remaining even-parity assertion. By the Fourier eigenbasis
of Lemma~\ref{lem:fourier-diagonalization}, every adjacency eigenvalue
has the form \(\lambda(a)=D_{q,n}\rho(a)\).
Constant frequencies have \(\rho(a)=1\) by
Lemma~\ref{lem:constant-frequencies}, while all other frequencies
satisfy \(\rho(a)\geq-1/(n-1)\) by
Lemma~\ref{lem:nonmonochromatic-bound}. These two cases cover the
entire spectrum, so
\begin{equation}
 \lambda_{\min}(\Om)\geq-\frac{D_{q,n}}{n-1}.
\end{equation}

It remains to exhibit a frequency attaining this bound. Take
\(a=(1,-1,0,\ldots,0)\), choose a balanced shift
\(X=(X_1,\ldots,X_n)\) uniformly from \(S_{q,n}\), and set
\(W_j=\zeta_q^{X_j}\), where \(\zeta_q=e^{2\pi i/q}\).
By Lemma~\ref{lem:normalized-compositions}\,\ref{item:normalized-wave},
\begin{equation}
 \rho(a)=\E\zeta_q^{a\cdot X}
        =\E\zeta_q^{X_1-X_2}
        =\E W_1\overline W_2.
\end{equation}
Thus we need to evaluate the average product of the phases at two
distinct coordinates.

Balancedness means that each \(q\)-th root of unity occurs exactly
\(\ell\) times among \(W_1,\ldots,W_n\). Since these roots sum to
zero, every allowed shift satisfies
\begin{equation}
 \sum_{j=1}^nW_j=\ell\sum_{t=0}^{q-1}\zeta_q^t=0.
\end{equation}
Write \(c=\E W_1\overline W_2\). Coordinate permutations preserve
the set of balanced shifts
(Lemma~\ref{lem:graph-properties}\,\ref{item:graph-coordinates})
and hence the uniform distribution of \(X\). Any ordered pair of
distinct coordinates can be sent to any other by such a permutation,
so \(\E W_i\overline W_j=c\) for every \(i\ne j\).
Expanding the squared modulus of the zero sum gives
\begin{equation}
 \begin{aligned}
 0&=\E\left|\sum_{j=1}^nW_j\right|^2\\
  &=\sum_{j=1}^n\E|W_j|^2
    +\sum_{i\ne j}\E W_i\overline W_j\\
  &=n+n(n-1)c.
 \end{aligned}
\end{equation}
Here each of the \(n\) diagonal terms equals one because
\(|W_j|=1\), and the off-diagonal sum contains \(n(n-1)\)
ordered pairs, each with expectation \(c\).
Solving for \(c\) therefore gives \(\rho(a)=c=-1/(n-1)\).
The Fourier wave \(\chi_a(x)=\zeta_q^{x_1-x_2}\) consequently
has eigenvalue
\begin{equation}
 \lambda(a)=D_{q,n}\rho(a)=-\frac{D_{q,n}}{n-1}.
\end{equation}
This attains the lower bound and proves the claimed equality.
When \(q=2\), the first two frequency entries both equal one
because \(-1=1\) in \(\Z_2\), and the same calculation applies.
\end{proof}

\section{Zero-error source coding}\label{sec:source-coding}
We now use the least eigenvalue to determine the minimum
communication for balanced pairs.
For a finite graph \(G\), consider the task in
Figure~\ref{fig:source-coding}, with an arbitrary edge as Bob's
candidate pair. Let \(\chs(G)\) denote the smallest classical
message alphabet permitting a perfect protocol with a shared
finite-dimensional state. This is the entanglement-assisted chromatic
number from zero-error source coding~\cite{bblps}: the graph of pairs
that Bob's side information can leave ambiguous is exactly \(G\).
We will combine Lemma~\ref{lem:source-hoffman} with
Theorem~\ref{thm:main} and matching protocols to
determine this minimum in both parity cases.

A quantum state is represented by a density matrix, a positive
semidefinite matrix with trace one. A measurement is described by
a POVM, a family of positive semidefinite operators summing to the
identity. For each input and message, we use Bob's unnormalized
conditional state: his post-measurement state weighted by the
probability of that message.

For completeness, we give a direct spectral proof of this known
communication bound, using the conditional-state formulation of
Bri\"et et al.~\cite[Sec.~1.2]{bblps}.
Their proof of Theorem~2.1 gives a stronger general bound through
semidefinite programming~\cite[Sec.~4]{bblps}.
Here we apply an adjacency eigenvalue inequality to vectors formed
from the conditional states and conclude with Cauchy--Schwarz.

\begin{proof}[Proof of Lemma~\ref{lem:source-hoffman}]
Let \(\sigma\) be a protocol's shared density matrix and
\(\{A_{x,a}\}_{a\in[k]}\) Alice's POVM for input \(x\).
The conditional states form a \emph{quantum state assemblage}
\(\{R_{x,a}\}_{x,a}\)~\cite[Sec.~II]{cs-steering}, indexed by Alice's
input \(x\) and outcome \(a\), which she sends as her message:
\begin{equation}\label{eq:source-conditional-states}
 R_{x,a}=\operatorname{Tr}_A\bigl[
 (A_{x,a}^{1/2}\otimes I)\sigma(A_{x,a}^{1/2}\otimes I)\bigr]
 =\operatorname{Tr}_A\bigl[(A_{x,a}\otimes I)\sigma\bigr].
\end{equation}
The first expression proves positivity.
The trace \(p(a\mid x)=\tr R_{x,a}\) is the probability of
message \(a\); whenever it is positive, Bob's normalized conditional
state is \(R_{x,a}/p(a\mid x)\).
The second expression follows by
cycling operators on the traced-out register and gives
\(\sum_aR_{x,a}=\operatorname{Tr}_A\sigma=\rho\), with
\(\tr\rho=1\). This common marginal, independent of \(x\), is
the assemblage's no-signalling condition.
For an edge \(\{u,v\}\) and message \(a\), let \(F\) be Bob's
effect for declaring \(u\); his other effect is \(I-F\).
Zero error gives
\(\tr[(I-F)R_{u,a}]=\tr(FR_{v,a})=0\).
For positive matrices \(R,E\), the identity
\(\tr(ER)=\|E^{1/2}R^{1/2}\|_{\HS}^2\) shows that a zero trace
pairing implies \(ER=0\). Thus
\((I-F)R_{u,a}=0\) and \(FR_{v,a}=0\). Taking the adjoint of
the first identity gives \(R_{u,a}F=R_{u,a}\), whence
\(R_{u,a}R_{v,a}=R_{u,a}FR_{v,a}=0\).
The conditional states therefore satisfy
\begin{equation}\label{eq:source-ensembles}
 \sum_{a=1}^k R_{x,a}=\rho,\qquad
 \tr\rho=1,\qquad R_{u,a}R_{v,a}=0\quad(u\sim v).
\end{equation}

Thus, for each fixed message, conditional states at adjacent vertices
have orthogonal supports, allowing Bob to distinguish them without
error. The adjacency matrix records exactly which pairs must obey these
constraints.

Let \(N=|V(G)|\), let \(A_G\) be the adjacency matrix, and let
\(e_x\) denote the standard coordinate basis of \(\C^{V(G)}\).
Put \(e=N^{-1/2}\one\).
Lemma~\ref{lem:regular-spectral-inequality}\,\ref{item:regular-spectral-bound}
gives
\begin{equation}
 A_G\succeq\tau I+(D-\tau)|e\rangle\langle e|.
\end{equation}
Using the Hilbert--Schmidt inner product on the matrix factor, put
\(c=\|\rho\|_{\HS}^2>0\),
\(z_a=\sum_x e_x\otimes R_{x,a}\), and \(T_a=\sum_xR_{x,a}\).
Expanding the quadratic form gives
\begin{equation}\label{eq:source-edge-quadratic}
 \langle z_a,(A_G\otimes I)z_a\rangle
 =\sum_{u,v}(A_G)_{uv}\tr(R_{u,a}R_{v,a})=0.
\end{equation}
Indeed, the adjacency entry vanishes unless \(u\sim v\), in which
case \(R_{u,a}R_{v,a}=0\) by~\eqref{eq:source-ensembles}.
The component along the normalized constant vector is
\begin{equation}\label{eq:source-constant-component}
 (\langle e|\otimes I)z_a
 =\sum_x\langle e,e_x\rangle R_{x,a}
 =\frac1{\sqrt N}\sum_xR_{x,a}
 =\frac{T_a}{\sqrt N}.
\end{equation}
Thus the quadratic form of
\(|e\rangle\langle e|\otimes I\) at \(z_a\) is
\(\|T_a\|_{\HS}^2/N\).
Tensoring the spectral inequality with the identity on the matrix
space, applying it to each \(z_a\), and summing over messages gives
\begin{equation}
 0\geq\tau\sum_a\|z_a\|^2
       +\frac{D-\tau}{N}\sum_a\|T_a\|_{\HS}^2.
\end{equation}
For positive matrices \(R,S\),
\(\tr(RS)=\|S^{1/2}R^{1/2}\|_{\HS}^2\geq0\).
Expanding \(\tr[(\sum_aR_{x,a})^2]=c\) therefore shows
\(\sum_a\|R_{x,a}\|_{\HS}^2\leq c\), and hence
\(\sum_a\|z_a\|^2\leq Nc\).
Also \(\sum_aT_a=N\rho\). Cauchy--Schwarz in the
Hilbert--Schmidt space gives
\begin{equation}\label{eq:source-message-cauchy-schwarz}
 N^2c=\left\|\sum_{a=1}^kT_a\right\|_{\HS}^2
 \leq k\sum_{a=1}^k\|T_a\|_{\HS}^2,
 \qquad
 \sum_a\|T_a\|_{\HS}^2\geq\frac{N^2c}{k}.
\end{equation}
Here \(k\geq1\), since the common sum has trace one.
Since \(\tau<0\), the bound \(\sum_a\|z_a\|^2\leq Nc\)
implies \(\tau\sum_a\|z_a\|^2\geq\tau Nc\).
Together with \(D-\tau>0\) and the preceding lower bound, this gives
\begin{equation}
 0\geq \tau Nc+(D-\tau)\frac{Nc}{k}.
\end{equation}
Cancelling \(Nc>0\) and multiplying by \(k>0\) gives
\(\tau k+D-\tau\leq0\), or
\((-\tau)k\geq D-\tau\). Division by \(-\tau>0\) yields
\(k\geq1-D/\tau\), as claimed.
\end{proof}

The odd-parity protocol uses the graph bipartition of
Lemma~\ref{lem:odd-bipartition}.

\begin{lemma}[Odd-parity coding]\label{lem:odd-source-coding}
Let \(q\geq2\), \(\ell\geq1\), and \(n=q\ell\), with
\((q-1)\ell\) odd. Then
\begin{equation}
 \chs(\Om)=2.
\end{equation}
The minimum is attained by a protocol using one classical bit and
no shared entanglement.
\end{lemma}
\begin{proof}
Let \(b(x)\in\{0,1\}\) be the class label from
Lemma~\ref{lem:odd-bipartition}, which satisfies \(b(v)=1-b(u)\)
on every edge.
Since edges are precisely the allowed candidate pairs, every such
pair contains exactly one word with each label.
Alice sends the single classical bit \(b(x)\). Bob computes
\(b(u)\) and \(b(v)\) and outputs the unique candidate whose bit
equals Alice's message. The protocol has zero error for either
candidate in every allowed pair. Alice's rule depends only on her
word, not on Bob's pair, and requires no shared entanglement.
Conversely, Theorem~\ref{thm:main} gives
\(\tau=-D_{q,n}\), so Lemma~\ref{lem:source-hoffman} yields
\(k\geq1-D_{q,n}/\tau=2\), even with arbitrary
finite-dimensional entanglement assistance.
\end{proof}

We now turn to the even-parity case and the Fourier protocol
introduced in Section~\ref{sec:introduction}.

\begin{corollary}[Optimal coding alphabet]\label{cor:source-coding}
If \(q\geq2\), \(\ell\geq1\), \(n=q\ell\), and
\((q-1)\ell\) is even, then
\begin{equation}
 \chs(\Om)=n.
\end{equation}
The minimum is over arbitrary finite-dimensional shared density
matrices, Alice's POVMs, and Bob's message-dependent binary decoders.
\end{corollary}
\begin{proof}
By Lemma~\ref{lem:graph-properties}, the graph is finite, simple,
and undirected. Lemma~\ref{lem:adjacency-structure} gives
\(A\one=D\one\), so it is \(D\)-regular.
By Lemma~\ref{lem:source-hoffman} and Theorem~\ref{thm:main}, every
perfect protocol therefore requires \(k\geq1-D/\tau=n\) messages.

For achievability, we apply the source-coding construction
of~\cite[Lemma~5.1]{bblps} to the phase representation
of~\cite[Lemma~4.1]{cao}. Explicitly, we use the flat-Fourier projectors
of~\cite[Proposition~7]{cmnsw}.
Let \(F_n\) be the normalized Fourier matrix, and for each vertex
\(x\in\Z_q^n\) define
\begin{equation}
 U_x=\operatorname{diag}(\zeta_q^{x_1},\ldots,\zeta_q^{x_n})F_n.
\end{equation}
Let \(P_{x,a}\) be the rank-one projection onto column \(a\) of
\(U_x\); these \(n\) projections sum to the identity.
At adjacent vertices \(u,v\), the inner product of columns with the
same index is
\begin{equation}
 \frac1n\sum_{j=1}^n\zeta_q^{v_j-u_j}
 =\frac{\ell}{n}\sum_{t=0}^{q-1}\zeta_q^t=0.
\end{equation}
Thus \(P_{u,a}P_{v,a}=0\) for every edge and message.
Share \(|\Phi_n\rangle=n^{-1/2}\sum_i|i\rangle|i\rangle\).
Alice measures \(\{P_{x,a}\}_a\) and sends \(a\).
To compute Bob's conditional state, expand the maximally entangled
state in this basis. For any \(n\times n\) matrix \(M\),
\begin{equation}\label{eq:source-maximally-entangled-transpose}
 \begin{aligned}
 \operatorname{Tr}_A\bigl[(M\otimes I)
             |\Phi_n\rangle\langle\Phi_n|\bigr]
 &=\frac1n\sum_{i,j}\langle j|M|i\rangle\,|i\rangle\langle j|\\
 &=\frac{M^{\mathsf T}}n.
 \end{aligned}
\end{equation}
Applying~\eqref{eq:source-conditional-states} with \(M=P_{x,a}\)
therefore gives the unnormalized state
\(R_{x,a}=P_{x,a}^{\mathsf T}/n\).
Since \(P_{x,a}\) has rank one, outcome \(a\) has probability
\(\tr R_{x,a}=1/n\), and Bob's normalized conditional state is
\(P_{x,a}^{\mathsf T}\).
Given \(\{u,v\}\), he uses the binary measurement
\(\{P_{u,a}^{\mathsf T},I-P_{u,a}^{\mathsf T}\}\), declaring
\(u\) or \(v\), respectively. Conditioned on message \(a\),
the probability of declaring \(u\) is
\begin{equation}
 \tr\bigl(P_{u,a}^{\mathsf T}P_{x,a}^{\mathsf T}\bigr)
 =\begin{cases}
  1,&x=u,\\
  0,&x=v.
 \end{cases}
\end{equation}
The first case follows from idempotence and unit trace; the second
from \(P_{v,a}P_{u,a}=0\), the adjoint of the edge relation.
The complementary outcome therefore declares \(v\) with certainty
when \(x=v\).
An arbitrary fixed ordering of the pair specifies which candidate
is tested first. This is a perfect protocol with \(n\) messages.
\end{proof}

\begin{remark}[Changing parity by preprocessing]\label{rem:parity-preprocessing}
To transfer the odd-case protocol, Alice would need a fixed map
\(F\) taking every allowed pair \(\{u,v\}\) to an allowed
odd-case pair \(\{F(u),F(v)\}\). Composing \(F\) with the
odd-case partition label would then give a two-message protocol
for the original task. Corollary~\ref{cor:source-coding} excludes
this in the even-parity regime, where \(n\geq3\) messages are
necessary. Choosing the lift separately for each candidate pair
does not directly define a valid protocol, since Alice does not
know that pair.
\end{remark}

Optimality also persists when several source instances are encoded
jointly. For \(m\geq1\), the \emph{strong power}
\(G^{\boxtimes m}\) has vertex set \(V(G)^m\); two distinct tuples
are adjacent exactly when their entries in every position are either
equal or adjacent in \(G\). This is the characteristic graph of
\(m\) independent source instances~\cite[Secs.~1.1--1.2]{bblps}.

\begin{corollary}[Block optimality]\label{cor:source-block-coding}
Let \(q\geq2\), \(\ell\geq1\), and \(n=q\ell\), and put
\(G=\Om\). For every integer \(m\geq1\),
\begin{equation}\label{eq:source-block-optimality}
 \chs\!\left(G^{\boxtimes m}\right)=
 \begin{cases}
 n^m,&(q-1)\ell\text{ even},\\
 2^m,&(q-1)\ell\text{ odd}.
 \end{cases}
\end{equation}
Equivalently, suppose Alice receives \(m\) words and Bob receives
one allowed candidate pair for each word. Recovering all \(m\) words
with zero error for every tuple of pairs and every choice of candidates
requires exactly the number of messages in
Eq.~\eqref{eq:source-block-optimality}. This minimum allows arbitrary
finite-dimensional shared states independent of the inputs and joint
local measurements across all \(m\) instances.
In the even-parity case, the optimal fixed-length binary cost is
\(\lceil m\log_2 n\rceil\), with asymptotic cost
\(\log_2 n\) bits per instance. In the odd-parity case, the optimal
cost is exactly \(m\) bits, or one bit per instance, and is attained
by a deterministic protocol without shared entanglement.
\end{corollary}

The upper bounds come from using an optimal single-instance code in
each position. The issue is whether a collective protocol can use fewer
messages. In even parity, the spectral lower bound survives under
strong products. In odd parity, fixing the same candidate pair in
every position produces \(2^m\) mutually confusable tuples; the clique
bound below shows why shared entanglement cannot compress their message
alphabet.

\begin{proof}
\emph{Even parity.}
Let \(\vartheta\) denote the Lov\'asz theta number and
\(\overline G\) the complement of \(G\). The spectral theta
bound~\cite[Sec.~2.1]{bdov} and the entanglement-assisted
source-coding bound~\cite[Theorem~2.1]{bblps}, also recorded in
Eq.~\eqref{eq:prior-spectral-coding}, give
\begin{equation}\label{eq:source-theta-tight}
 n=1-\frac D\tau
 \leq\vartheta(\overline G)\leq\chs(G)=n,
\end{equation}
where \(D=D_{q,n}\), \(\tau=-D/(n-1)\) by
Theorem~\ref{thm:main}, and the last equality is
Corollary~\ref{cor:source-coding}.
Thus \(\vartheta(\overline G)=n\). The multiplicativity identity
in~\cite[Eq.~(10)]{bblps} and its Theorem~2.1 imply
\begin{equation}
 n^m=\vartheta(\overline G)^m
 =\vartheta\!\left(\overline{G^{\boxtimes m}}\right)
 \leq\chs\!\left(G^{\boxtimes m}\right).
\end{equation}
For the reverse inequality, run the protocol of
Corollary~\ref{cor:source-coding} on each instance using a tensor
product of its shared states. Alice sends the tuple of its \(m\)
outcomes, which has \(n^m\) possible values, and Bob decodes each
word. This is the standard strong-product construction
following~\cite[Definition~1.5]{bblps}.

\emph{Odd parity.}
Let \(b:V(G)\to\{0,1\}\) be the bipartition label from
Lemma~\ref{lem:odd-bipartition}. Alice sends
\((b(x_1),\ldots,b(x_m))\). Each allowed pair has one word of
each label, so Bob identifies the correct word in every position.
This gives \(2^m\) messages using no entanglement. The label vector
is also a proper coloring of \(G^{\boxtimes m}\), since two adjacent
distinct tuples differ along an edge in at least one position.

For the converse, a clique of size \(t\geq1\) requires at least \(t\)
messages even with entanglement. To see this directly, let
\(R_{x,a}\) be Bob's conditional states in a perfect \(k\)-message
protocol, with \(\sum_a R_{x,a}=\rho\) and
\(\operatorname{tr}\rho=1\), as in
Lemma~\ref{lem:source-hoffman}. Restrict \(x\) to the clique and put
\(T_a=\sum_x R_{x,a}\) and \(c=\operatorname{tr}(\rho^2)>0\).
For fixed \(a\), the states indexed by distinct clique vertices have
orthogonal supports. Positivity and completeness give
\(\sum_{a,x}\|R_{x,a}\|_{\mathrm{HS}}^2\leq tc\).
Consequently, Cauchy--Schwarz and orthogonality imply
\begin{equation}\label{eq:block-clique-bound}
 t^2c=\left\|\sum_a T_a\right\|_{\mathrm{HS}}^2
 \leq k\sum_a\|T_a\|_{\mathrm{HS}}^2
 =k\sum_{a,x}\|R_{x,a}\|_{\mathrm{HS}}^2
 \leq ktc.
\end{equation}
Cancelling \(tc>0\) yields \(k\geq t\).
Since \(D_{q,n}>0\), choose an edge \(\{u,v\}\) of \(G\).
The \(2^m\) tuples in \(\{u,v\}^m\) form a clique of
\(G^{\boxtimes m}\): any two distinct tuples have equal or adjacent
entries in every position. The clique bound gives the required
\(k\geq2^m\), including for arbitrary collective protocols.

\emph{Operational interpretation and bits.}
The characteristic-graph formulation applies also to Bob's
\(2^m\)-element candidate sets: for a fixed message, pairwise
orthogonal conditional-state supports permit simultaneous perfect
discrimination~\cite[Lemma~1.2]{bblps}.
Encoding an alphabet of size \(r^m\) needs
\(\lceil\log_2(r^m)\rceil=\lceil m\log_2 r\rceil\) bits.
For \(r=n\), dividing by \(m\) and taking the limit gives
\(\log_2 n\); for \(r=2\), the cost is exactly \(m\) bits.
\end{proof}

The equality in Eq.~\eqref{eq:source-block-optimality} rules out a
reduction in the message alphabet through collective encoding in
either parity regime. In even parity, jointly packing the labels into
binary strings can still remove the rounding overhead from encoding
each label separately. In odd parity, the optimal label vector already
uses exactly one bit per instance.

Luo, Ning, and Zhang relate balanced graphs to graphs obtained by
deleting one coordinate~\cite[Theorem~6 and Corollary~7]{lnz}.
Combining their existing reduction with
Corollary~\ref{cor:source-coding} gives the following operational
corollary throughout our admissible parameter range.

\begin{corollary}[One missing occurrence]\label{cor:near-balanced-coding}
Let \(q\geq2\), \(\ell\geq1\), and \(n=q\ell-1\), with
\((q-1)\ell\) even. Alice receives a word \(x\in\Z_q^n\), and
Bob receives an unordered pair \(\{u,v\}\) of distinct words
containing \(x\). Suppose that, for some \(r\in\Z_q\) and
every \(a\in\Z_q\),
\begin{equation}
 \#\{j:v_j-u_j=a\}=\begin{cases}
   \ell-1,&a=r,\\
   \ell,&a\ne r.
 \end{cases}
\end{equation}
The deficient residue \(r\) may vary with the pair, which Alice
does not know. The minimum entanglement-assisted message alphabet
has size \(n+1=q\ell\), and the optimal fixed-length cost is
\(\lceil\log_2(n+1)\rceil\) classical bits. A single protocol must
succeed without error for every promised pair and either candidate;
arbitrary finite-dimensional shared states and local POVMs are allowed.
\end{corollary}
\begin{proof}
Let \(H\) be the graph of the promised pairs and let
\(G=\Omega_{n+1}^{(\Z_q)}\) be the balanced graph. An
edge-preserving vertex map \(f:H\to G\) transfers any protocol
for \(G\) to \(H\): Alice encodes \(f(x)\), and Bob uses the
pair \(\{f(u),f(v)\}\). The two image words are distinct, so recovering
the image identifies the original candidate. Hence
\(\chs(H)\leq\chs(G)\). We recall the maps in both directions
from~\cite[Theorem~6]{lnz}.

For the map from \(H\) to \(G\), append a checksum coordinate:
\begin{equation}
 E(x)=\left(x_1,\ldots,x_n,-\sum_{j=1}^n x_j\right)\in\Z_q^{n+1}.
\end{equation}
If \(d=v-u\) has deficient residue \(r\), then admissibility gives
\begin{equation}
 \sum_{j=1}^n d_j
 =\ell\sum_{a=0}^{q-1}a-r
 =-r\pmod q.
\end{equation}
Thus the appended difference is \(r\), restoring the missing
occurrence and making \(E(v)-E(u)\) balanced.

Conversely, deleting the last coordinate sends every balanced pair
in \(G\) to a promised pair in \(H\): the removed residue occurs
\(\ell-1\) times, and all others occur \(\ell\) times. The two
shortened words remain distinct, since a balanced difference has
\((q-1)\ell\geq2\) nonzero coordinates. This gives
\(\chs(G)\leq\chs(H)\). Corollary~\ref{cor:source-coding},
applied at length \(n+1=q\ell\), now yields
\begin{equation}
 \chs(H)=\chs(G)=n+1.
\end{equation}
\end{proof}

\section{Quantum coloring}\label{sec:quantum}
The spectral theorem and the explicit odd-parity bipartition
together determine the quantum chromatic number throughout the
full parameter range.
In the graph-coloring game, two separated players receive vertices
of a graph and return colors. They must return the same color on
equal vertices and different colors on adjacent vertices. A perfect
classical strategy is a proper coloring; shared entanglement can
reduce the number of colors needed~\cite{cmnsw}.
Figure~\ref{fig:coloring-game} summarizes the game and its quantum resources.

\begin{figure}[htbp]
\centering
\includegraphics[width=\linewidth]{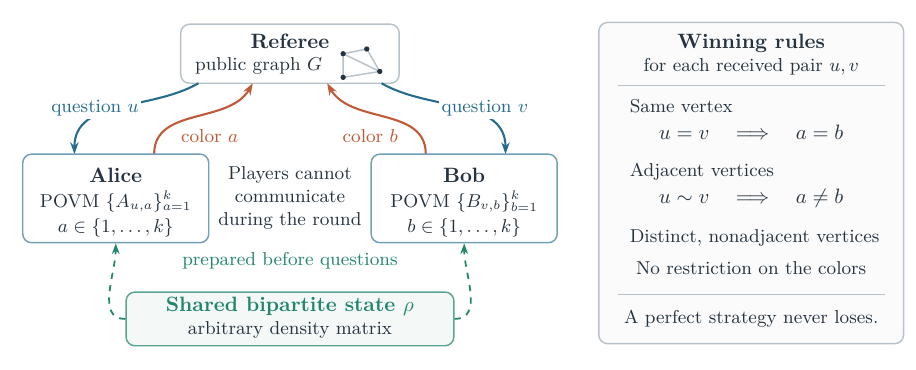}
\caption{The graph-coloring game for a public graph \(G\).
The referee sends vertices \(u,v\) to Alice and Bob, who reply with
colors \(a,b\in\{1,\ldots,k\}\). The players share a density matrix
\(\rho\) prepared before the round and use question-dependent local
POVMs, without communication during the round. Equal vertices require
equal colors; adjacent vertices require different colors; distinct
nonadjacent vertices impose no condition. A perfect strategy satisfies
these rules with probability one for every question pair.}
\label{fig:coloring-game}
\end{figure}

We use the usual finite-dimensional quantum chromatic number
\(\chq(G)\). Explicitly, a quantum \(k\)-coloring consists of a
bipartite density matrix \(\rho\) and local POVMs
\(\{A_{u,a}\}_{a=1}^k\), \(\{B_{v,b}\}_{b=1}^k\), in arbitrary
nonzero finite dimensions (cf.~\cite{zkfb,zksp}), for which the Born probabilities
\begin{equation}
 p(a,b\mid u,v)=\tr\bigl(\rho(A_{u,a}\otimes B_{v,b})\bigr)
\end{equation}
vanish when \(u=v\) and \(a\ne b\), or when \(u\sim v\) and
\(a=b\). The minimum possible \(k\) is \(\chq(G)\).
Appendix~\ref{app:normal} proves the equivalent projector formulation
used below, without purity or projectivity assumptions on the
original strategy.

\begin{theorem}[Quantum chromatic number]\label{thm:quantum-coloring}
Let \(q\geq2\), \(\ell\geq1\), and \(n=q\ell\).
Then
\begin{equation}
 \chq(\Om)=
 \begin{cases}
  n,&(q-1)\ell\text{ even},\\
  2,&(q-1)\ell\text{ odd}.
 \end{cases}
\end{equation}
In the odd-parity case, \(\chi(\Om)=2\) as well, and the
two-color strategy requires no entanglement.
\end{theorem}

\begin{remark}[Odd-parity coloring in the literature]\label{rem:odd-coloring-literature}
For \(q=2\) and \(n\equiv2\pmod4\), Godsil and
Newman~\cite[Sec.~2]{gn} describe the bipartition by Hamming-weight
parity, and Luo, Ning, and Zhang~\cite[Sec.~2]{lnz} explicitly
record \(\chq=\chi=2\). For general graphs, Cameron
et al.~\cite[Proposition~3]{cmnsw} prove
\(\chq(G)=2\) if and only if \(\chi(G)=2\).
For \(q=4\), the odd-parity spectral sign symmetry follows from
Ning, Koolen, and Zhang~\cite[Theorem~4.3 and Eq.~(45)]{nkz}.
\mbox{Cao~et~al.~\cite[Lemma~4.3]{cao}} give the corresponding symmetry
for the full cyclic family, without a length threshold.
These symmetries preserve multiplicities and imply
bipartiteness~\cite{bh}, hence the two-color value; the latter two
papers do not explicitly state these odd-parity coloring conclusions.
Lemma~\ref{lem:odd-bipartition} gives a concrete partition throughout
the odd-parity regime and proves that every edge crosses it. This
separates every allowed candidate pair, yielding the optimal one-bit
protocol of Lemma~\ref{lem:odd-source-coding}.
\end{remark}

For the even-parity lower bound, we use the standard projector normal
form~\cite[Proposition~1 and Eq.~(4)]{cmnsw}, proved in
Appendix~\ref{app:normal}. A quantum \(k\)-coloring therefore gives
Hermitian projections \(P_{u,a}\in M_h(\C)\), with \(h\geq1\), satisfying
\begin{equation}
 \sum_{a=1}^kP_{u,a}=I_h,\qquad
 P_{u,a}P_{v,a}=0\quad(u\sim v).
\end{equation}

The following lemma is the regular-graph specialization of the
quantum Hoffman bound~\cite[Secs.~5--6]{ew}. We include a direct proof.

\begin{lemma}[Regular-graph quantum Hoffman bound]\label{lem:hoffman}
If a finite \(D\)-regular graph has least adjacency eigenvalue
\(\tau<0\), then \(\chq(G)\geq1-D/\tau\).
\end{lemma}
\begin{proof}
Let \(N=|V(G)|\), let \(A_G\) be its adjacency matrix, and put
\(e=N^{-1/2}\one\). Let \(e_u\) denote the standard coordinate basis
of \(\C^{V(G)}\).
In \(\C^{V(G)}\otimes M_h(\C)\), using the Hilbert--Schmidt inner
product on the matrix factor, put
\begin{equation}
 z_a=\sum_u e_u\otimes P_{u,a},\qquad T_a=\sum_uP_{u,a}.
\end{equation}
As in~\eqref{eq:source-edge-quadratic}, the coloring relations give
\(\langle z_a,(A_G\otimes I)z_a\rangle=0\).
Lemma~\ref{lem:regular-spectral-inequality}\,\ref{item:regular-spectral-bound}
gives
\(A_G\succeq\tau I+(D-\tau)|e\rangle\langle e|\).
The constant-component calculation in
\eqref{eq:source-constant-component}, with \(R_{u,a}\) replaced by
\(P_{u,a}\), gives \((\langle e|\otimes I)z_a=T_a/\sqrt N\).
Tensoring the spectral inequality with the identity on the matrix
space and applying it to \(z_a\) therefore yields
\begin{equation}
 0\geq\tau\|z_a\|^2+\frac{D-\tau}{N}\|T_a\|_{\rm HS}^2.
\end{equation}
Since \(\|P_{u,a}\|_{\HS}^2=\tr P_{u,a}\), completeness gives
\(\sum_a\|z_a\|^2=Nh\). Also \(\sum_aT_a=NI_h\), and
Cauchy--Schwarz gives \(\sum_a\|T_a\|_{\rm HS}^2\geq N^2h/k\).
Summing the preceding inequalities yields
\begin{equation}
 0\geq\tau Nh+(D-\tau)\frac{Nh}{k},
\end{equation}
which rearranges to \(k\geq1-D/\tau\).
\end{proof}

\begin{proof}[Proof of Theorem~\ref{thm:quantum-coloring}]
Suppose first that \((q-1)\ell\) is odd. By
Lemma~\ref{lem:odd-bipartition}, the class label \(b\) satisfies
\(b(v)=1-b(u)\) on every edge.
It therefore defines a proper two-coloring. Both players can return
\(b(x)+1\) on receiving vertex \(x\): equal vertices give equal
colors, and adjacent vertices give different colors. This is a
perfect strategy without entanglement. Since \(D_{q,n}>0\), the
graph has an edge, so one color is impossible even for a quantum
strategy. Thus \(\chq(\Om)=\chi(\Om)=2\).

Suppose now that \((q-1)\ell\) is even.
Lemma~\ref{lem:adjacency-structure} gives \(A\one=D\one\), so
\(\Om\) is \(D\)-regular. Theorem~\ref{thm:main} and
Lemma~\ref{lem:hoffman} therefore give
\(\chq(\Om)\geq1-D/\tau=n\).
For the matching upper bound, use the Fourier projections
\(P_{x,a}\) constructed in the proof of
Corollary~\ref{cor:source-coding}. They sum to the identity at each
vertex and satisfy \(P_{u,a}P_{v,a}=0\) on every edge. The projector
normal form therefore gives a quantum \(n\)-coloring.
\end{proof}

\section{Discussion and open problems}\label{sec:discussion}
In the even-parity regime, the uniform coefficient bound determines
the least adjacency eigenvalue at every admissible length. It proves
optimality of the existing entanglement-assisted \(n\)-message source
code and establishes that the quantum chromatic number is \(n\).
The converse therefore places a sharp limit on the benefit of
entanglement: even under the balanced-difference promise, perfect
recovery requires at least \(n\) possible classical messages, or
\(\lceil\log_2 n\rceil\) fixed-length bits. This minimum holds for
arbitrary finite-dimensional shared states and local POVMs, so
neither a larger entangled resource nor more general measurements
can reduce it.

In the odd-parity regime, our explicit bipartition yields an optimal
deterministic one-bit source code and classical and quantum chromatic
numbers equal to \(2\). Thus entanglement provides no advantage in
either task at any odd-parity length: the optimal message and color
counts are already attained without it.

The following questions concern the even-parity regime, focusing on
the quantum resources needed to attain optimality and on protocols
that allow errors.

\paragraph{Entanglement and local dimension.}
The Fourier constructions in Sections~\ref{sec:source-coding}
and~\ref{sec:quantum} use a maximally entangled state with local
dimension \(n\). What is the smallest \(d\) for which an optimal
strategy exists with both local dimensions at most \(d\)? This
question should be considered separately for source codes with
\(n\) messages and quantum colorings with \(n\) colors, since the
two tasks allow different strategies. A related question is how
much entanglement is needed. For pure shared states, one precise
formulation asks for the infimum of the entanglement entropy,
the entropy of either reduced state, among strategies attaining
the corresponding optimum.

\paragraph{Allowing a small error.}
Our results require perfect success on every allowed input.
For \(0<\varepsilon<1/2\), one can instead require failure
probability at most \(\varepsilon\) for every candidate pair and
either choice of Alice's word. The analogous coloring problem
allows failure probability at most \(\varepsilon\) on every equal
or adjacent question pair. What are the optimal message and color
counts as functions of \(q,n,\varepsilon\)? In particular, how do
they behave when \(\varepsilon\) decreases with \(n\)?
The zero-error converse does not by itself determine these
tradeoffs; quantitative bounds would clarify how the exact
optimality results change when errors are permitted.

\paragraph{Formal verification.}
Every lemma, theorem, and corollary has been verified in Lean~4.19.0
with Mathlib. The Lean sources are available in~\cite{zeiss-lean}.

\section*{Author contributions and use of AI}
OpenAI Codex assisted with literature searches, mathematical exploration
and proof refinement, Lean implementation and debugging, and manuscript
preparation. The formalized results were checked using Lean and Mathlib.
The author is responsible for the content of this article.

\section*{Acknowledgments}
The author thanks Tobias Rippchen, Steven Kim, Cormac Stopes, and
Robert Salzmann for discussions on quantum graph coloring games,
and the Lean and Mathlib communities for the proof
assistant and mathematical library used in this work.
JZ acknowledges support from the European Research Council
(ERC Grant Agreement No.~948139) and the Excellence Cluster --
Matter and Light for Quantum Computing (ML4Q-2).

\appendix
\clearpage
\phantomsection
\addcontentsline{toc}{section}{Appendix}
\addtocontents{toc}{\protect\setcounter{tocdepth}{-1}}
\section{Notation}\label{app:notation}
Table~\ref{tab:notation} groups the principal notation by context,
following the structure of the paper. Symbols reused in different contexts are
listed separately; temporary variables are defined locally.
Unless stated otherwise, \(n=q\ell\). In
Corollary~\ref{cor:near-balanced-coding}, the shortened word length
is instead \(n=q\ell-1\).

\begingroup
\footnotesize
\renewcommand{\arraystretch}{1.00}
\setlength{\tabcolsep}{5pt}
\setlength{\LTpre}{6pt}
\begin{longtable}{@{}>{\raggedright\arraybackslash}p{.255\linewidth}>{\raggedright\arraybackslash}p{.70\linewidth}@{}}
\caption{Principal notation.}\label{tab:notation}\\
\toprule
Symbol & Meaning \\
\midrule
\endfirsthead
\multicolumn{2}{l}{\small Table~\thetable\ (continued)}\\
\toprule Symbol & Meaning \\
\midrule
\endhead
\midrule
\multicolumn{2}{r@{}}{\small Continued on the next page}\\
\endfoot
\bottomrule
\endlastfoot
\multicolumn{2}{@{}l}{\textbf{Graph and parameters (Section~\ref{sec:graph})}}\\*
\(q,\ell,n=q\ell\) & Alphabet size, residue multiplicity, and balanced word length; \(q\geq2,\ell\geq1\).\\
\(\Z_q\), \(\Om\) & Cyclic group \(\Z/q\Z\); graph with vertex set \(\Z_q^n\).\\
\(S_{q,n}\), \(u\sim v\) & Differences in which every residue occurs \(\ell\) times; adjacency means \(v-u\in S_{q,n}\).\\
\(D=D_{q,n}\), \(N\) & Degree and vertex count of a regular graph; for \(\Om\), \(D=n!/(\ell!)^q\) and \(N=q^n\).\\
\(\chi(G)\) & Classical chromatic number.\\
\(K_{m_1,\ldots,m_t}\) & Complete multipartite graph with independent parts of sizes \(m_1,\ldots,m_t\).\\
\(\sigma(x)\), \(b(x)\) & Coordinate sum modulo \(q\), represented in \(\{0,\ldots,q-1\}\); odd-parity bipartition label from~\eqref{eq:odd-bipartition}.\\[3pt]
\multicolumn{2}{@{}l}{\textbf{Adjacency and Fourier analysis (Sections~\ref{sec:graph}, \ref{sec:fourier}, and~\ref{sec:spectral})}}\\*
\(\mathfrak S_n\), \(\pi u\) & Symmetric group on \(\{1,\ldots,n\}\); coordinate action \((\pi u)_j=u_{\pi^{-1}(j)}\).\\
\(V,E,f\) & Vertex set, unordered-edge set, and a vector of values assigned to vertices.\\
\(\zeta_q\), \(\chi_a(x)\) & Root \(e^{2\pi i/q}\); character \(\zeta_q^{a\cdot x}\), with frequency \(a\in\Z_q^n\).\\
\(A,A_G\), \(\lambda(a)\), \(\tau\) & Adjacency operator; eigenvalue at frequency \(a\); least adjacency eigenvalue.\\
\(\succeq\) & Positive semidefinite matrix order.\\
\(\rho(a)\) & Normalized adjacency eigenvalue \(\lambda(a)/D_{q,n}\); the density matrices \(\rho\) are listed below.\\
\(r=\operatorname{comp}(a)\) & Composition \((r_0,\ldots,r_{q-1})\); \(r_c\) counts occurrences of label \(c\), so \(r_c\geq0\) and \(\sum_c r_c=n\).\\
\(X\), \(W_j\), \(\E\) & Uniform random balanced word; phase \(\zeta_q^{X_j}\); expectation. The switching position set \(X\) is listed below.\\
\(c=\E W_1\overline W_2\) & Common expectation \(\E W_i\overline W_j\) for \(i\ne j\) in Section~\ref{sec:spectral}; it equals \(-1/(n-1)\). Distinct from the residue label and coding purity denoted by \(c\) below.\\[3pt]
\multicolumn{2}{@{}l}{\textbf{Symmetric Schur--Weyl sector (Section~\ref{sec:fourier-counts})}}\\*
\(V=\C^q\), \(\Pi_{\mathrm{sym}}\) & Single-particle space; orthogonal projector onto \((V^{\otimes n})^{\mathfrak S_n}\).\\
\(U_\pi\), \(\nu\vdash n\) & Tensor-factor permutation; partition of \(n\), with at most \(q\) parts in the Schur--Weyl decomposition.\\
\(\mathbb S_\nu(V)\), \([\nu]\) & Corresponding irreducible representations of \(\mathrm{GL}(V)\) and \(\mathfrak S_n\).\\
\(\operatorname{Sym}^n(V)\) & Symmetric sector \((V^{\otimes n})^{\mathfrak S_n}\), identified with homogeneous degree-\(n\) polynomials.\\
\(\operatorname{Sym}^n(g)\) & Restriction of \(g^{\otimes n}\) to the symmetric sector, also realized by substitution on degree-\(n\) polynomials.\\
\(M_r\), \(|r\rangle_{\mathrm{sym}}\) & Size \(\binom n{r_0,\ldots,r_{q-1}}\) of a composition orbit; its normalized orbit sum.\\
\(\mathbf b\), \(\mathcal F_q\) & Balanced composition \((\ell,\ldots,\ell)\), with \(M_{\mathbf b}=D_{q,n}\); unnormalized Fourier matrix \((\zeta_q^{ct})_{c,t\in\Z_q}\).\\
\(z_c\), \([z^r]\) & Commuting indeterminates; coefficient of \(z^r=\prod_c z_c^{r_c}\).\\
\(C_q(z)\), \(H_{q,\ell}(z)\) & Circulant \((z_{j-i})_{i,j\in\Z_q}\); permanent power \((\per C_q(z))^\ell\).\\
\(P(z)\) & Fourier image \(\operatorname{Sym}^n(\mathcal F_q)z^{\mathbf b}=(\det C_q(z))^\ell\).\\
\(s\), \(p_s\) & Composition index and coefficient \(p_s=[z^s]P(z)\) in \(P(z)=\sum_s p_s z^s\); in particular, \(p_r=[z^r]P(z)\). Here \(M_s\) and \(|s\rangle_{\mathrm{sym}}\) use the same definitions as \(M_r\) and \(|r\rangle_{\mathrm{sym}}\).\\
\(\per\), \(\binom n{r_0,\ldots,r_{q-1}}\) & Permanent; multinomial coefficient \(n!/\prod_c r_c!\).\\[3pt]
\multicolumn{2}{@{}l}{\textbf{Assignments and switching (Section~\ref{sec:switching})}}\\*
\(X=\{1,\ldots,\ell\}\times\Z_q\) & Set of \(n\) positions; \((b,i)\) is site \(i\) in block \(b\).\\
\(w:X\to\Z_q\), \(c\) & Assignment of displacement labels; a specified residue label.\\
\(W_r,E_r,I_r\) & All assignments of composition \(r\), the valid assignments, and the invalid assignments \(I_r=W_r\setminus E_r\).\\
\(F_w(b,i)\), \(T_c(b,i)\) & Target map \((b,i+w(b,i))\); fixed shift \((b,i+c)\). Validity means \(F_w\) is a permutation.\\
\(p\), \(\tau_p\), \(w^p\) & Unordered pair of distinct positions, its transposition, and the swapped assignment \(w\circ\tau_p\).\\
\(\mathcal P(w)\), \(d\), \(R\) & Permitted swap pairs, their number per valid assignment, and an upper bound on permitted repairs per invalid assignment.\\
\(r_c(n-r_c)\) & Number of swaps exchanging label \(c\) with a different label; when all letters occur once, allowing every swap gives \(d=\binom n2\).\\
\(V\), \(I\) & Numbers \(|E_r|\) and \(|I_r|\) of valid and invalid assignments in the counting explanations.\\[3pt]
\multicolumn{2}{@{}l}{\textbf{Source coding (Section~\ref{sec:source-coding}, Appendix~\ref{app:prior-regimes})}}\\*
\(\chs(G)\), \(L=\{u,v\}\) & Minimum entanglement-assisted message-alphabet size; Bob's adjacent candidate pair, containing Alice's input \(x\).\\
\(k\), \(a\in[k]\) & Number of possible messages; message sent by Alice, with \([k]=\{1,\ldots,k\}\).\\
\(m\), \(G^{\boxtimes m}\) & Number of jointly encoded source instances; the \(m\)-fold strong power of \(G\).\\
\(\overline G\), \(\vartheta(G)\) & Complement graph; Lov\'asz theta number.\\
\(\sigma\), \(\rho\), \(R_{x,a}\) & Shared density matrix, Bob's reduced density matrix, and his unnormalized assemblage element for input \(x\) and message \(a\).\\
\(\operatorname{Tr}_A\), \(F\), \(F_{L,a}\) & Partial trace over Alice's register; Bob's effect for declaring the first candidate, with complement \(I-F\), written \(F_{L,a}\) when the pair and message are explicit.\\
\(c=\|\rho\|_{\HS}^2\) & Purity of Bob's reduced state, used in the coding converse.\\
\(z_a\), \(T_a\) & Vector \(\sum_x e_x\otimes R_{x,a}\) and sum \(\sum_xR_{x,a}\) used in the coding converse; their projector counterparts are listed below.\\
\(|\Phi_n\rangle\), \(P_{x,a}\) & Shared vector \(n^{-1/2}\sum_i|i\rangle\otimes|i\rangle\) in the Fourier protocol; rank-one projection onto column \(a\) of \(U_x\).\\
\(p(a\mid x)\), \(P_{x,a}^{\mathsf T}\) & Message probability \(\tr R_{x,a}\); normalized conditional state in the Fourier protocol, where \(p(a\mid x)=1/n\) and \(R_{x,a}=P_{x,a}^{\mathsf T}/n\).\\
\(M\) & Arbitrary \(n\times n\) matrix in~\eqref{eq:source-maximally-entangled-transpose}; its transpose is taken in the basis defining \(|\Phi_n\rangle\).\\
\(p_\lambda\), \(\sigma_\lambda\), \(p(a\mid x,\lambda)\) & Hidden-variable distribution, input-independent hidden state, and message probabilities in the local-hidden-state model~\eqref{eq:source-lhs}.\\[3pt]
\multicolumn{2}{@{}l}{\textbf{Quantum strategies (Section~\ref{sec:quantum}, Appendix~\ref{app:normal})}}\\*
\(\chq(G)\) & Finite-dimensional quantum chromatic number.\\
\(k\), \(a,b\) & Number of output colors; Alice's and Bob's outputs in \(\{1,\ldots,k\}\).\\
\(\rho\), \(A_{u,a},B_{v,b}\) & Bipartite density matrix; local POVM effects for questions \(u,v\) and outputs \(a,b\).\\
\(p(a,b\mid u,v)\) & Born probability \(\tr[\rho(A_{u,a}\otimes B_{v,b})]\).\\
\(P_{u,a}\), \(h\), \(I_h\) & Coloring projection; positive local dimension; identity on \(\C^h\).\\
\(F_n\), \(U_v\) & Normalized Fourier matrix of order \(n\); \(\operatorname{diag}(\zeta_q^{v_1},\ldots,\zeta_q^{v_n})F_n\).\\
\(e_u\), \(e\), \(\one\) & Vertex basis vector; unit constant vector \(N^{-1/2}\one\); all-ones vector.\\
\(z_a\), \(T_a\) & Hoffman vector \(\sum_u e_u\otimes P_{u,a}\); projector sum \(\sum_uP_{u,a}\).\\
\(\psi\), \(\Psi\), \(H_0\) & Common nonzero column of \(\rho\), its coefficient matrix, and the invariant space \(\operatorname{ran}\Psi\).\\
\(|\Phi_h\rangle\) & Normalized maximally entangled vector \(h^{-1/2}\sum_i|i\rangle\otimes|i\rangle\).\\
\((\cdot)^*,(\cdot)^{\mathsf T}\), \(\|\cdot\|_{\HS}\) & Adjoint, transpose; \(\|M\|_{\HS}^2=\tr(M^*M)\), with the ordinary trace.\\
\end{longtable}
\endgroup
\clearpage

\section{From arbitrary quantum strategies to projectors}\label{app:normal}
We prove the standard finite-dimensional projector normal
form~\cite[Proposition~1 and Eq.~(4)]{cmnsw} used in
Section~\ref{sec:quantum}. The equivalence applies to the arbitrary
mixed states and local POVMs allowed in
Theorem~\ref{thm:quantum-coloring}.

\begin{lemma}[Projector normal form]\label{lem:normal}
A finite graph has a perfect quantum \(k\)-coloring with an arbitrary
finite-dimensional mixed state and local POVMs if and only if there
are Hermitian projections \(P_{u,a}\in M_h(\C)\), with \(h\geq1\),
such that
\begin{equation}\label{eq:projectors}
 \sum_{a=1}^kP_{u,a}=I_h,
 \qquad P_{u,a}P_{v,a}=0\quad(u\sim v).
\end{equation}
\end{lemma}
\begin{proof}
Write \(\rho\) for the shared density matrix and
\(A_{u,a},B_{v,b}\) for the local POVM effects.

\emph{1. A consequence of positivity.}
For positive
semidefinite matrices \(R,E\) of the same size,
\begin{equation}\label{eq:trace-zero}
 \tr(RE)=\|E^{1/2}R^{1/2}\|_{\HS}^2\geq0,
 \qquad \tr(RE)=0\ \Longrightarrow\ ER=0.
\end{equation}
Indeed, cyclicity of the trace gives the equality. If the norm
vanishes, \(E^{1/2}R^{1/2}=0\), so multiplication on the left by
\(E^{1/2}\) and on the right by \(R^{1/2}\) gives \(ER=0\).
In particular the Born probabilities of the mixed strategy are real
and nonnegative. Their sum is one, since both POVMs sum to the
identity and \(\tr\rho=1\).

\emph{2. From a perfect strategy to projections.}
Suppose a perfect mixed-state strategy with local POVMs is given.
A question-answer tuple \((u,v,a,b)\) is losing when
\(u=v\) and \(a\ne b\), or when \(u\sim v\) and \(a=b\).
For each such tuple, the tensor effect
\(E=A_{u,a}\otimes B_{v,b}\) is positive semidefinite
(for instance, it is the square of
\(A_{u,a}^{1/2}\otimes B_{v,b}^{1/2}\)).
The zero Born probability and~\eqref{eq:trace-zero} imply
\(E\rho=0\). Choose a nonzero column \(\psi\) of \(\rho\), which
exists because \(\tr\rho=1\). The same column is annihilated by
every losing tensor effect. In fixed orthonormal bases of the two
local spaces, write
\begin{equation}
 \psi=\sum_{i,j}\Psi_{ij}|i\rangle\otimes|j\rangle,
 \qquad \Psi\ne0.
\end{equation}
The matrix of coefficients of \((A\otimes B)\psi\) is
\(A\Psi B^{\mathsf T}\). Therefore
\begin{equation}\label{eq:zero-amplitudes}
 A_{u,a}\Psi B_{v,b}^{\mathsf T}=0
 \quad\text{whenever }(u,v,a,b)\text{ is losing}.
\end{equation}

Synchrony and completeness now give, for every \(u,a\),
\begin{equation}\label{eq:intertwining}
 A_{u,a}\Psi
 =A_{u,a}\Psi B_{u,a}^{\mathsf T}
 =\Psi B_{u,a}^{\mathsf T}.
\end{equation}
For the first equality, insert \(\sum_bB_{u,b}^{\mathsf T}=I\)
and discard the terms with \(b\ne a\) using
\eqref{eq:zero-amplitudes}; the second follows by inserting
\(\sum_bA_{u,b}=I\) on the left. It follows that
\begin{equation}
 A_{u,a}^2\Psi
 =A_{u,a}\Psi B_{u,a}^{\mathsf T}=A_{u,a}\Psi,
 \qquad
 A_{u,a}A_{v,a}\Psi
 =A_{u,a}\Psi B_{v,a}^{\mathsf T}=0\quad(u\sim v).
\end{equation}
Let \(H_0=\operatorname{ran}\Psi\), a nonzero subspace of Alice's
space, and put \(h=\dim H_0\).
Equation~\eqref{eq:intertwining} shows that every \(A_{u,a}\)
preserves \(H_0\). Its restriction to \(H_0\) is Hermitian, since
\(A_{u,a}\) is Hermitian, and is idempotent by the first equality
above. The restricted operators sum to \(I_{H_0}\) and satisfy the
edge relations by the second equality. In an orthonormal basis
of \(H_0\) they are the required projections.

\emph{3. From projections to a perfect strategy.}
Conversely, suppose~\eqref{eq:projectors} holds. Projections at a
fixed vertex are pairwise orthogonal: for any \(u,a\),
\begin{equation}
 0=P_{u,a}(I-P_{u,a})P_{u,a}
   =\sum_{b\ne a}P_{u,a}P_{u,b}P_{u,a}
   =\sum_{b\ne a}(P_{u,b}P_{u,a})^*(P_{u,b}P_{u,a}).
\end{equation}
Taking the trace gives
\(\sum_{b\ne a}\|P_{u,b}P_{u,a}\|_{\HS}^2=0\).
Hence \(P_{u,b}P_{u,a}=0\) for every \(b\ne a\).
Use the normalized maximally entangled state
\begin{equation}
 |\Phi_h\rangle=\frac1{\sqrt h}\sum_{i=1}^h|i\rangle\otimes|i\rangle,
 \qquad \rho=|\Phi_h\rangle\langle\Phi_h|,
\end{equation}
with Alice's projective measurement \(P_{u,a}\) and Bob's
\(P_{v,b}^{\mathsf T}\). These are POVMs because transposition
preserves positivity and their sums are identities. The coefficient
matrix of the vector after applying the two outcome projectors is
\(h^{-1/2}P_{u,a}P_{v,b}\), which vanishes for both kinds of losing
tuples. For projections the Born probability is exactly the squared
norm of this vector, so the strategy is perfect.
\end{proof}

\section{Coding consequences of earlier results}\label{app:prior-regimes}
We derive the coding consequences of existing graph and source-coding
results for the source considered here. The cited results supply the
ingredients; the deductions below do not assert that those papers
explicitly formulated this particular coding application.

\paragraph{Identifying the coding costs.}
Let \(q\geq2\), \(\ell\geq1\), \(n=q\ell\), and \(G=\Om\).
Every edge is a possible candidate pair, and either
candidate can be Alice's input, so the characteristic graph is exactly
\(G\). A classical encoder must assign different messages to adjacent
vertices; conversely, sending a vertex's color resolves every pair.
Thus the minimum classical message alphabet size is \(\chi(G)\), as in
the general framework of~\cite[Sec.~1.1]{bblps}. Its
entanglement-assisted counterpart is the minimum alphabet size
\(\chs(G)\), defined in
Section~\ref{sec:source-coding}.
Shared or private randomness does not reduce this classical minimum.
Indeed, sample all random tapes independently of the inputs. Since
there are finitely many allowed inputs and the protocol succeeds
with probability one on each, there is a choice of tapes for which
it succeeds on all of them. Fixing those tapes gives a deterministic
code with the same message alphabet.

\paragraph{Steering is necessary for a strict coding advantage.}
Consider the conditional-state assemblage \(\{R_{x,a}\}_{x,a}\)
of a perfect \(k\)-message protocol, as defined in
Section~\ref{sec:source-coding}.
A local-hidden-state model expresses it as
\begin{equation}\label{eq:source-lhs}
 R_{x,a}=\sum_\lambda p_\lambda\,p(a\mid x,\lambda)\,
 \sigma_\lambda,
\end{equation}
where \(p_\lambda\) is a probability distribution, each
\(\sigma_\lambda\) is a density matrix independent of \(x\),
and \(p(a\mid x,\lambda)\) is a conditional probability
distribution. Since the input and outcome sets are finite, the
hidden-variable set can also be taken finite~\cite[Secs.~II and III.A]{cs-steering}.
Such a model gives a classical simulation with the same message
alphabet. Alice and Bob share the random variable \(\lambda\);
on input \(x\), Alice samples \(a\) from \(p(a\mid x,\lambda)\)
and sends it. For candidate pair \(L\) and message \(a\), let
\(F_{L,a}\) be the original decoder's effect for declaring the
first candidate in a fixed ordering of \(L\). Bob declares that
candidate with probability \(\tr(F_{L,a}\sigma_\lambda)\), and
the other candidate otherwise. By~\eqref{eq:source-lhs}, averaging
over \(\lambda\) reproduces all of the original protocol's
message and decoding probabilities. The simulation is therefore
perfect and, by fixing its random tapes as above, yields a
deterministic classical code with the same \(k\) messages.
Consequently \(k\geq\chi(G)\). A strict advantage over the
classical minimum thus requires an assemblage with no
local-hidden-state model, which is steering in the technical
sense~\cite[Sec.~II.A]{cs-steering}. This is a necessary condition;
steering alone does not guarantee a coding advantage.

\paragraph{An elementary classical code.}
Fix a set \(J\subseteq[n]\) of \(\ell+1\) coordinates, known to
both parties, and let Alice send \(x|_J\).
The balanced-difference promise says that exactly \(\ell\)
coordinates of \(v-u\) equal zero, so the two candidates agree
in exactly \(\ell\) positions. They therefore cannot agree on
all of \(J\). Bob recovers \(x\) by comparing the received tuple
with \(u|_J\) and \(v|_J\), exactly one of which matches.
There are \(q^{\ell+1}\) possible tuples. Encoding the entire
tuple as a single label gives
\begin{equation}
 \chi(G)\leq q^{\ell+1},\qquad
 \lceil\log_2\chi(G)\rceil
 \leq\lceil(\ell+1)\log_2 q\rceil.
\end{equation}
This elementary upper bound holds in both parity regimes;
it need not attain the exact classical optimum.

\paragraph{Combining the upper and lower bounds.}
Let \(\xi'(G)\) be the smallest dimension of an orthogonal
representation whose coordinates all have modulus one. The phase
representation of~\cite[Lemma~4.1]{cao} and the source-coding
construction of~\cite[Lemma~5.1]{bblps} give, respectively,
\begin{equation}
 \xi'(G)\leq n,\qquad \chs(G)\leq\xi'(G),
 \qquad\text{hence}\qquad \chs(G)\leq n.
\end{equation}
This holds whenever \(q\geq2\) and \(n>0\) is divisible by \(q\),
without a parity assumption. Section~\ref{sec:source-coding} gives
the resulting protocol explicitly.
For fixed \(q\), \cite[Theorem~1.9]{cao} gives constants
\(N_q,\varepsilon_q>0\) such that
\(\chi(G)\geq(1+\varepsilon_q)^n\) whenever \(n\geq N_q\),
\(q\mid n\), and \((q-1)\ell\) is even. Under these conditions,
the optimal fixed-length bit costs satisfy
\begin{equation}
 \underbrace{\lceil\log_2\chs(G)\rceil}_{\text{with entanglement}}
 \leq\lceil\log_2 n\rceil,
 \qquad
 \underbrace{\lceil\log_2\chi(G)\rceil}_{\text{classical}}
 \geq n\log_2(1+\varepsilon_q).
\end{equation}
Combining the classical lower bound with the coordinate-sending code
and \(\ell=n/q\) gives
\begin{equation}
 n\log_2(1+\varepsilon_q)
 \leq\lceil\log_2\chi(G)\rceil
 \leq\left\lceil\left(\frac nq+1\right)\log_2 q\right\rceil.
\end{equation}
Hence, along the even-parity lengths, the optimal classical cost is
\(\Theta(n)\) for fixed \(q\),
whereas the known entanglement-assisted code uses
\(\lceil\log_2 n\rceil\) bits. The \(O(\log n)\)-versus-\(\Theta(n)\)
separation already follows from existing results: even sending
the entire word gives the sufficient classical upper bound
\(\lceil n\log_2 q\rceil\). No use of our uniform spectral theorem
is needed for that separation. Corollary~\ref{cor:source-coding}
additionally proves that the entanglement-assisted bit cost is
exactly \(\lceil\log_2 n\rceil\) at every even-parity length.

Here \(n\) is the word length within one source instance;
\(m\) in Corollary~\ref{cor:source-block-coding} is the number of
source instances encoded jointly. The classical lower bound
and the resulting separation require sufficiently large \(n\)
depending on \(q\); they do not determine the exact classical cost
at every parameter pair or imply a strict advantage at every finite
length. In particular, the classical codes in
Appendix~\ref{app:full-graphs} attain \(n\) messages at
\((q,n)=(2,4)\) and \((3,3)\).

\paragraph{From a spectral formula to a coding converse.}
Let \(\vartheta\) denote the Lov\'asz theta number, and let
\(\overline G\) be the complement of \(G\). For a \(D\)-regular
graph with least eigenvalue \(\tau<0\), the spectral theta bound
in~\cite[Sec.~2.1]{bdov} and the source-coding bound
of~\cite[Theorem~2.1]{bblps} give
\begin{equation}\label{eq:prior-spectral-coding}
 1-\frac D\tau\leq\vartheta(\overline G)\leq\chs(G).
\end{equation}
The proof of Lemma~\ref{lem:source-hoffman} in
Section~\ref{sec:source-coding} gives a direct derivation of this
communication lower bound. For \(G=\Om\), the identity
\(A\one=D\one\) in Lemma~\ref{lem:adjacency-structure} gives regularity.
Thus, whenever an earlier result establishes \(\tau=-D/(n-1)\),
Eq.~\eqref{eq:prior-spectral-coding} and the upper bound
\(\chs(G)\leq n\) give
\begin{equation}
 n=1-\frac D\tau\leq\chs(G)\leq n.
\end{equation}
This is the extra implication needed to obtain coding optimality
from a published spectral result. The equality \(\chq(G)=n\)
alone would not provide the lower bound, since
\(\chs(G)\leq\chq(G)\)~\cite[Sec.~2.1]{bblps}.

\paragraph{Earlier spectral inputs.}
The required least-eigenvalue formula holds for \(q=2\),
\(4\mid n\), as recorded in~\cite[Sec.~3.2]{wed}, and for
\(q=3\), \(3\mid n\), by~\cite[Theorem~2]{lnz}.
Cao et al. prove it for every fixed \(q\) at sufficiently large
even-parity lengths~\cite[Lemma~4.4]{cao}. Their finite-field
formula~\cite[Lemma~4.6]{cao} also applies to the cyclic alphabet
when \(q\) is prime and \(n=q^r\), \(r\geq1\), within the
even-parity range. For nonprime prime powers, the additive group of
the finite field is not the cyclic group, so that identification
does not apply. Substituting each of these spectral inputs
into~\eqref{eq:prior-spectral-coding} proves \(\chs(G)=n\)
in the corresponding regime.

\paragraph{From a generalized Hadamard matrix to a coding converse.}
An order-\(n\) generalized Hadamard matrix over \(\Z_q\) has
balanced differences between every two distinct rows, so its rows
form an \(n\)-clique in \(G\). The clique bound for theta and
the same source-coding bound~\cite{bblps} then give
\begin{equation}
 n\leq\vartheta(\overline G)\leq\chs(G)\leq n.
\end{equation}
This supplies another sufficient condition for coding optimality,
without requiring a spectral computation. For example, the
order-ten Butson matrix over fifth roots of unity recorded
in~\cite[Remark~5.11]{efo} corresponds to a generalized Hadamard
matrix over \(\Z_5\) by~\cite[Lemma~2.3]{efo}. The rows of the
latter matrix therefore certify
the coding optimum at \((q,n)=(5,10)\).

\paragraph{The additional deduction for \(q=4\).}
Let \(n=4\ell\), with \(\ell\geq1\) even, and put \(D=n!/(\ell!)^4\).
We derive the spectral input from~\cite[Lemmas~4.5--4.7]{nkz}
and a binomial estimate. For a frequency with composition
\(t=(t_0,t_1,t_2,t_3)\), write \(\lambda(t)\) for its
adjacency eigenvalue and set \(k=t_0+t_2\).
For \(0<k<n\), \cite[Lemma~4.7]{nkz} gives
\(|\lambda(t)|\leq D/(n-1)\) when \(n\geq10\), hence for
every \(n\geq16\) divisible by eight. Its proof also gives a
coefficient estimate valid at \(n=8\): the eigenvalue is zero
for odd \(k\), and for even \(k\),
\begin{equation}
 \frac{|\lambda(t)|}{D}
 \leq\frac{\binom{n/2}{k/2}}{\binom nk}.
\end{equation}
For \(n=8\), the interior even values \(k=2,4,6\) give
\(1/7,3/35,1/7\), respectively, all at most \(1/(n-1)\).

For \(k=0,n\), all frequency coordinates have the same parity.
Subtracting that common parity places them in \(\{0,2\}\).
This constant shift preserves the eigenvalue because \(\ell\) is even.
If \(a\) coordinates equal two, the corresponding phase is a
product of \(a\) signs sampled uniformly without replacement from
\(n/2\) plus signs and \(n/2\) minus signs. Equivalently, the exact formula
of~\cite[Lemma~4.6]{nkz} becomes
\begin{equation}
 \frac{\lambda(t)}D
 =\frac{[z^a](1-z^2)^{n/2}}{\binom na}
 =\begin{cases}
    0,&a\text{ odd},\\
    (-1)^m R_m,&a=2m,
   \end{cases}
 \qquad
 R_m=\frac{\binom{n/2}{m}}{\binom n{2m}}.
\end{equation}
The identities
\begin{equation}
 R_m=R_{n/2-m},\qquad
 \frac{R_{m+1}}{R_m}=\frac{2m+1}{n-2m-1}
\end{equation}
show that \(R_m\leq R_1=1/(n-1)\) for
\(1\leq m\leq n/2-1\). The values \(a=0\) and \(a=n\) both
give \(\lambda(t)/D=1\), since \(n/2\) is even.
Thus every eigenvalue is at least \(-D/(n-1)\).
The frequency type \((0,1,n-2,1)\) attains this value
by~\cite[Lemma~4.5]{nkz}. Consequently
\begin{equation}
 \lambda_{\min}(G)=-\frac D{n-1},\qquad \chs(G)=n
 \quad\text{for }q=4,\ 8\mid n,
\end{equation}
where the coding equality follows
from~\eqref{eq:prior-spectral-coding} and the upper bound.
This derivation uses the cited lemmas and the additional estimate
above, independently of our uniform spectral theorem.

These regimes overlap. Theorem~\ref{thm:main} supplies the spectral
input, and hence the coding converse, for every finite even-parity
parameter pair \((q,n)\), without a length threshold or a
matrix-existence assumption.

\section{Additional worked examples}\label{app:examples}
These examples illustrate graph neighborhoods, adjacency spectra,
Fourier waves, the reduction from frequency compositions to
polynomial coefficients, and the switching bound. We conclude with
two small graphs and their optimal classical codes.

\Needspace{7\baselineskip}
\subsection{Graph neighborhoods}
\label{app:example-neighborhoods}

\begin{example}[Building a vertex's neighborhood]
\label{ex:shift-neighborhood}
Take \(q=n=2\), so \(\ell=1\). The allowed shifts are
\(S_{2,2}=\{(0,1),(1,0)\}\). For \(u=(0,1)\), addition modulo
two gives
\begin{equation}
 (0,1)+(0,1)=(0,0),\qquad
 (0,1)+(1,0)=(1,1).
\end{equation}
Thus the two neighbors of \(u\) are \((0,0)\) and \((1,1)\).
The word \((1,0)\) is not a neighbor, because its difference from
\(u\) is \((1,1)\), which is not balanced.
\end{example}

\Needspace{7\baselineskip}
\subsection{Adjacency spectra and Fourier waves}
\label{app:example-spectra}

\begin{example}[The triangle]\label{ex:triangle}
The triangle \(K_3\) has adjacency matrix and spectrum
\begin{equation}
 \vcenter{\hbox{%
 \begin{tikzpicture}[every node/.style={circle,draw=MidnightBlue,
   fill=white,minimum size=5mm,inner sep=1pt,font=\small}]
  \node (v1) at (0,1.35) {1};
  \node (v2) at (-.78,0) {2};
  \node (v3) at (.78,0) {3};
  \draw[MidnightBlue,line width=.7pt] (v1)--(v2)--(v3)--(v1);
 \end{tikzpicture}}}
 \qquad
 A_{K_3}=\begin{pmatrix}
  0&1&1\\
  1&0&1\\
  1&1&0
 \end{pmatrix},
 \qquad
 \operatorname{spec}(A_{K_3})=(2,-1,-1).
\end{equation}
Rows and columns follow the vertex labels \(1,2,3\). The diagonal
entries are zero because there are no loops; every off-diagonal
entry is one because the corresponding vertices are joined. For
values \(f=(a,b,c)^{\mathsf T}\) assigned to the three vertices,
neighbor summation gives
\begin{equation}
 A_{K_3}(a,b,c)^{\mathsf T}
 =(b+c,a+c,a+b)^{\mathsf T}.
\end{equation}
For the constant vector \(\one=(1,1,1)^{\mathsf T}\), which assigns
the value \(1\) to every vertex, the result is \((2,2,2)^{\mathsf T}\),
so its eigenvalue is \(2\).
For \((1,-1,0)^{\mathsf T}\), the result is
\((-1,1,0)^{\mathsf T}\): the vector changes sign, giving
eigenvalue \(-1\).

More generally, if \(a+b+c=0\), each neighbor sum is the negative
of the value at that vertex, so \(A_{K_3}f=-f\). These zero-sum
vectors form a two-dimensional space with basis
\((1,-1,0)^{\mathsf T}\) and \((1,0,-1)^{\mathsf T}\), which
explains why \(-1\) occurs twice in the spectrum. Together with
the constant vector, they span \(\mathbb R^3\), so the spectrum is
complete. Thus \(D=2\) and \(\tau=-1\).
Lemma~\ref{lem:source-hoffman} therefore requires at least three
messages for the corresponding source-coding task.
\end{example}

\begin{example}[Evaluating a Fourier wave]\label{ex:fourier-wave}
For \(q=3\), \(n=3\), and \(a=(1,0,0)\), we have
\begin{equation}
 \chi_a(x)=\zeta_3^{x_1}
 =\begin{cases}
  1,&x_1=0,\\
  e^{2\pi i/3},&x_1=1,\\
  e^{4\pi i/3},&x_1=2.
 \end{cases}
\end{equation}
This wave assigns a value to each vertex using only its first
coordinate. Other choices of \(a\) can combine several coordinates.
\end{example}

\begin{example}[Frequency vectors and their compositions]\label{ex:frequency-composition}
For \(q=3\) and \(n=6\), consider
\begin{equation}
 a=(0,1,0,2,1,0),\qquad a'=(0,0,0,1,1,2).
\end{equation}
Both contain three zeros, two ones, and one two. Their compositions
therefore agree, and
Lemma~\ref{lem:normalized-compositions}\,\ref{item:composition-invariance} gives
\begin{equation}
 \operatorname{comp}(a)=\operatorname{comp}(a')=(3,2,1),
 \qquad \rho(a)=\rho(a').
\end{equation}
These frequency vectors define different Fourier waves with the
same eigenvalue. Their common composition is not monochromatic.
The monomial associated with both vectors in
Section~\ref{sec:fourier-counts} is \(z_0^3z_1^2z_2\).
For \(b=(2,2,2,2,2,2)\), however,
\(\operatorname{comp}(b)=(0,0,6)\) is monochromatic: only the
label \(2\) occurs, and it occurs in all six positions.
\end{example}

\Needspace{7\baselineskip}
\subsection{From compositions to coefficients}
\label{app:example-coefficients}

\begin{example}[From words to coefficient counts]\label{ex:coefficient-counting}
Take \(q=3\), \(\ell=1\), and \(r=(1,1,1)\). There are
\(3!=6\) words in this composition class, obtained by arranging
\(0,1,2\). Number the positions in this single block by
\(i=0,1,2\). For a word \(w=(w_0,w_1,w_2)\), the letter
\(w_i\) selects column \(i+w_i\pmod3\) in row \(i\) of
\(C_3(z)\). This selects one entry from every row automatically.
The word is valid exactly when the selected columns are all different,
so that every column is also selected once.

For instance, \(w=(0,1,2)\) selects columns \((0,2,1)\).
The selected-column permutation exchanges \(1,2\), giving
permutation sign \(-1\). The word \((0,2,1)\), however,
selects column \(0\) in all three rows and is invalid.
The following table applies this test to all six words. The sign
is that of the selected-column permutation and is defined only
for valid words.
\begin{center}
\renewcommand{\arraystretch}{1.1}
\begin{tabular}{@{}cccc@{}}
\toprule
Word \(w\) & Selected columns & Valid? & Permutation sign \\
\midrule
\((0,1,2)\) & \((0,2,1)\) & Yes & \(-1\) \\
\((0,2,1)\) & \((0,0,0)\) & No  & \(\text{--}\) \\
\((1,0,2)\) & \((1,1,1)\) & No  & \(\text{--}\) \\
\((1,2,0)\) & \((1,0,2)\) & Yes & \(-1\) \\
\((2,0,1)\) & \((2,1,0)\) & Yes & \(-1\) \\
\((2,1,0)\) & \((2,2,2)\) & No  & \(\text{--}\) \\
\bottomrule
\end{tabular}
\end{center}
The three valid words give the matrix selections in
Figure~\ref{fig:composition-coefficients}. Each selects the entries
\(z_0,z_1,z_2\) once, producing \(z_0z_1z_2\), and each
selected-column permutation exchanges two positions. Hence
\begin{equation}
 [z_0z_1z_2]\det C_3(z)=-3,\qquad
 [z_0z_1z_2]\per C_3(z)=3.
\end{equation}
Dividing by the six words in the class,
Lemma~\ref{lem:coefficient-reduction} gives
\begin{equation}
 \rho(a)=\frac{-3}{6}=-\frac12,\qquad
 |\rho(a)|\leq\frac36=\frac12
 \quad\text{when }\operatorname{comp}(a)=(1,1,1).
\end{equation}
All six words are balanced, but only three are valid. Balance fixes
the letter counts; validity requires one selected entry per row and
column in the coefficient expansion. These are different conditions.
Example~\ref{ex:switching} continues with these six words to explain
the switching argument of Section~\ref{sec:switching}.
\end{example}

\begin{figure}[htbp]
\centering
\includegraphics[width=\linewidth]{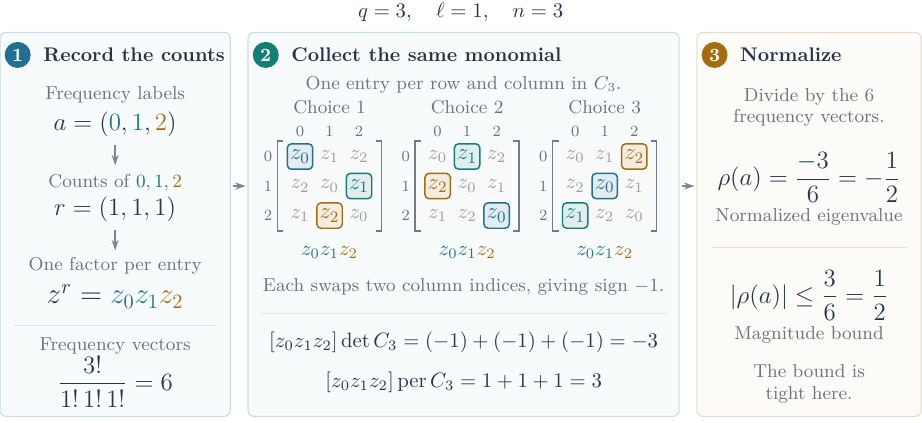}
\caption{From a composition to an eigenvalue, with \(q=3\) and
\(\ell=1\). The middle panels show all three selections of one entry
per row and column of \(C_3(z)\) that produce \(z_0z_1z_2\).
Rows and columns are indexed by \(0,1,2\); colors track the residue
labels. Each selection has negative permutation sign, giving
determinant coefficient \(-3\) and permanent coefficient \(3\).
Dividing by the six frequency vectors of composition \((1,1,1)\)
gives the normalized eigenvalue and its bound in
Lemma~\ref{lem:coefficient-reduction}.}
\label{fig:composition-coefficients}
\end{figure}

\Needspace{7\baselineskip}
\subsection{Switching and valid assignments}\label{app:example-switching}

\begin{example}[From swaps to a counting bound]\label{ex:switching}
Continue Example~\ref{ex:coefficient-counting} with \(q=n=3\)
and composition \(r=(1,1,1)\). We now use swaps to compare the
numbers of valid and invalid words in this same class. A \emph{swap}
exchanges the letters at two positions. The positions stay fixed,
and the letter counts are preserved, so the new word remains in the
same composition class. After swapping, we recompute the selected
columns \(i+w_i\pmod3\), which we call the \emph{targets} of
the positions.

\emph{Three positions.}
Start with the valid word \((0,1,2)\), whose targets are
\((0,2,1)\). Exchanging the letters at positions \(0,1\) gives
\begin{equation}
 \underbrace{(0,1,2)}_{\text{valid}}
 \xrightarrow{\text{swap letters at }0,1}
 \underbrace{(1,0,2)}_{\text{invalid}}.
\end{equation}
The letter at position \(2\) has not moved and still selects
target \(1\). Both changed positions now also select target \(1\),
so the new targets are \((1,1,1)\): target \(1\) is reached
three times, while targets \(0,2\) are missing. The upper panel of
Figure~\ref{fig:switching} shows these arrows before and after the
swap.

\Needspace{11\baselineskip}
A \emph{repair} is any swap that restores validity, possibly at a
different valid word. For the invalid word \((1,0,2)\),
all three choices of two positions give a repair:
\begin{center}
\renewcommand{\arraystretch}{1.1}
\begin{tabular}{@{}ccc@{}}
\toprule
Positions swapped & Repaired word & Recomputed targets \\
\midrule
\(0,1\) & \((0,1,2)\) & \((0,2,1)\) \\
\(0,2\) & \((2,0,1)\) & \((2,1,0)\) \\
\(1,2\) & \((1,2,0)\) & \((1,0,2)\) \\
\bottomrule
\end{tabular}
\end{center}
Each target list contains \(0,1,2\) exactly once. The three
repairs therefore lead to the three valid words in
Example~\ref{ex:coefficient-counting}.

The lower panel of Figure~\ref{fig:switching} records all swaps
between valid and invalid words in this class. Each valid word has
three swaps, all leading to invalid words; each invalid word has
three repairing swaps. Write \(V=|E_r|\) and \(I=|I_r|\) for
the numbers of valid and invalid words. Each line represents one reversible
swap, with one valid endpoint and one invalid endpoint. Counting
the same lines from the two sides gives
\begin{equation}
 \underbrace{3V}_{\text{from the valid words}}
 =\underbrace{3I}_{\text{from the invalid words}},
 \qquad \frac{V}{V+I}=\frac12=\frac1{n-1}.
\end{equation}
Here \(V=I=3\), so both counts give nine lines. This is the
switching method: compare how many swaps leave a valid word with
how many can repair an invalid word.

Every letter occurs once in this example, so it illustrates the
three-repair case in Section~\ref{sec:switching}. When a letter
occurs between \(2\) and \(n-2\) times, the proof instead
restricts swaps to those involving that letter and establishes
at most two repairs per invalid word.

\begin{figure}[htbp]
\centering
\includegraphics[width=\linewidth]{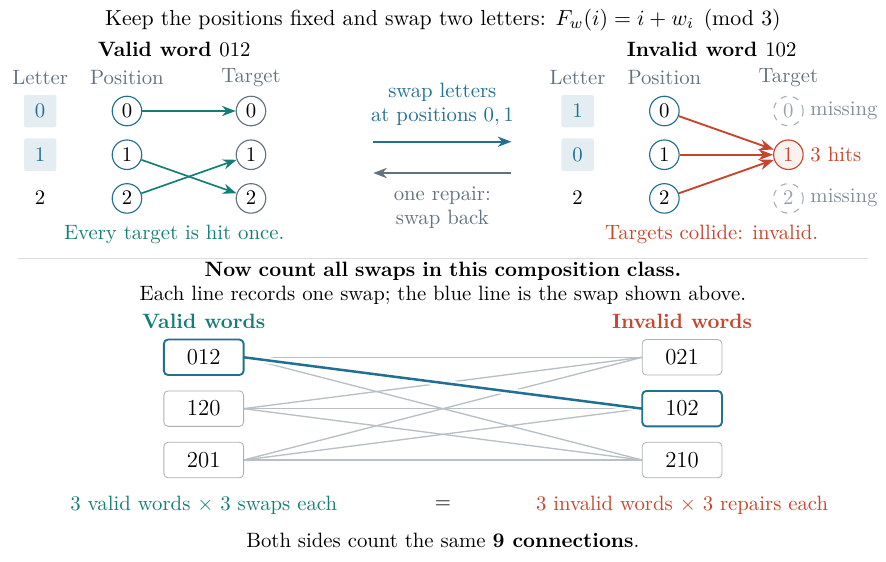}
\caption{Switching within the composition class \((1,1,1)\) from
Example~\ref{ex:coefficient-counting}. Top: swapping two letters of
\((0,1,2)\) creates the invalid word \((1,0,2)\), whose three
targets coincide. Bottom: all nine swaps between the three valid
and three invalid words. The highlighted line is the swap shown
above. Every line can be read in either direction: as a swap
destroying validity or as a repair.}
\label{fig:switching}
\end{figure}

\Needspace{12\baselineskip}
\emph{An equality family.}
To see how invalid assignments become more numerous as \(n\) grows,
keep \(q=3\) and take \(n=3\ell\) with composition
\(r=(n-2,1,1)\). There are \(n(n-1)\) assignments: choose
the position carrying \(1\), then a different position carrying
\(2\); all remaining letters are zero. Every zero position maps
to itself, so validity requires the two nonzero positions to exchange
targets. If \(1\) is at position \(i\) in a block, then \(2\)
must be at position \(i+1\pmod3\) in the same block. Thus each
of the \(n\) choices for the position of \(1\) determines exactly
one valid assignment, giving
\begin{equation}
 V=n,\qquad I=n(n-2),\qquad
 \frac{I}{V}=n-2,\qquad \frac{V}{V+I}=\frac1{n-1}.
\end{equation}
The table gives exact counts for this family; its first row is the
six-word example above.
\begin{center}
\renewcommand{\arraystretch}{1.2}
\begin{tabular}{@{}rrrrcc@{}}
\toprule
\(n\) & Total \(V+I\) & Valid \(V\) & Invalid \(I\)
      & \(I/V\) & \(V/(V+I)\) \\
\midrule
 3 &          6 &  3 &          3 &  1 & \(1/2\)  \\
 6 &         30 &  6 &         24 &  4 & \(1/5\)  \\
 9 &         72 &  9 &         63 &  7 & \(1/8\)  \\
12 &        132 & 12 &        120 & 10 & \(1/11\) \\
30 &        870 & 30 &        840 & 28 & \(1/29\) \\
60 & \(3\,540\) & 60 & \(3\,480\) & 58 & \(1/59\) \\
\bottomrule
\end{tabular}
\end{center}
For example, at \(n=30\) there are \(28\) invalid assignments
for every valid one. This family attains the bound of
Theorem~\ref{thm:permanent}; for any nonmonochromatic composition,
the theorem guarantees \(I\geq(n-2)V\), although equality need
not hold.

\emph{Six positions.}
For \(n=6\), the following table lists all \(30\) words of
composition \((4,1,1)\), using the same validity test as
Example~\ref{ex:coefficient-counting}. A vertical bar separates the
two fixed blocks of three positions. Within each block, the positions
are numbered \(0,1,2\), and the targets are \(i+w_i\pmod3\).
A word is valid exactly when \emph{each block} of its target list
contains \(0,1,2\) once; targets may repeat across different blocks.
For valid words, the sign is the product of the two block-permutation
signs.
\begingroup
\small
\renewcommand{\arraystretch}{1.0}
\begin{longtable}{@{}cccc@{}}
\toprule
Word \(w\) & Selected columns in each block & Valid? & Sign \\
\midrule
\endfirsthead
\multicolumn{4}{c}{Words of composition \((4,1,1)\) (continued)} \\
\toprule
Word \(w\) & Selected columns in each block & Valid? & Sign \\
\midrule
\endhead
\midrule
\multicolumn{4}{r@{}}{\small Continued on the next page} \\
\endfoot
\bottomrule
\endlastfoot
\((0,0,0\mid 0,1,2)\) & \((0,1,2\mid 0,2,1)\) & Yes & \(-1\) \\
\((0,0,0\mid 0,2,1)\) & \((0,1,2\mid 0,0,0)\) & No & \(\text{--}\) \\
\((0,0,0\mid 1,0,2)\) & \((0,1,2\mid 1,1,1)\) & No & \(\text{--}\) \\
\((0,0,0\mid 1,2,0)\) & \((0,1,2\mid 1,0,2)\) & Yes & \(-1\) \\
\((0,0,0\mid 2,0,1)\) & \((0,1,2\mid 2,1,0)\) & Yes & \(-1\) \\
\((0,0,0\mid 2,1,0)\) & \((0,1,2\mid 2,2,2)\) & No & \(\text{--}\) \\
\((0,0,1\mid 0,0,2)\) & \((0,1,0\mid 0,1,1)\) & No & \(\text{--}\) \\
\((0,0,1\mid 0,2,0)\) & \((0,1,0\mid 0,0,2)\) & No & \(\text{--}\) \\
\((0,0,1\mid 2,0,0)\) & \((0,1,0\mid 2,1,2)\) & No & \(\text{--}\) \\
\((0,0,2\mid 0,0,1)\) & \((0,1,1\mid 0,1,0)\) & No & \(\text{--}\) \\
\((0,0,2\mid 0,1,0)\) & \((0,1,1\mid 0,2,2)\) & No & \(\text{--}\) \\
\((0,0,2\mid 1,0,0)\) & \((0,1,1\mid 1,1,2)\) & No & \(\text{--}\) \\
\((0,1,0\mid 0,0,2)\) & \((0,2,2\mid 0,1,1)\) & No & \(\text{--}\) \\
\((0,1,0\mid 0,2,0)\) & \((0,2,2\mid 0,0,2)\) & No & \(\text{--}\) \\
\((0,1,0\mid 2,0,0)\) & \((0,2,2\mid 2,1,2)\) & No & \(\text{--}\) \\
\((0,1,2\mid 0,0,0)\) & \((0,2,1\mid 0,1,2)\) & Yes & \(-1\) \\
\((0,2,0\mid 0,0,1)\) & \((0,0,2\mid 0,1,0)\) & No & \(\text{--}\) \\
\((0,2,0\mid 0,1,0)\) & \((0,0,2\mid 0,2,2)\) & No & \(\text{--}\) \\
\((0,2,0\mid 1,0,0)\) & \((0,0,2\mid 1,1,2)\) & No & \(\text{--}\) \\
\((0,2,1\mid 0,0,0)\) & \((0,0,0\mid 0,1,2)\) & No & \(\text{--}\) \\
\((1,0,0\mid 0,0,2)\) & \((1,1,2\mid 0,1,1)\) & No & \(\text{--}\) \\
\((1,0,0\mid 0,2,0)\) & \((1,1,2\mid 0,0,2)\) & No & \(\text{--}\) \\
\((1,0,0\mid 2,0,0)\) & \((1,1,2\mid 2,1,2)\) & No & \(\text{--}\) \\
\((1,0,2\mid 0,0,0)\) & \((1,1,1\mid 0,1,2)\) & No & \(\text{--}\) \\
\((1,2,0\mid 0,0,0)\) & \((1,0,2\mid 0,1,2)\) & Yes & \(-1\) \\
\((2,0,0\mid 0,0,1)\) & \((2,1,2\mid 0,1,0)\) & No & \(\text{--}\) \\
\((2,0,0\mid 0,1,0)\) & \((2,1,2\mid 0,2,2)\) & No & \(\text{--}\) \\
\((2,0,0\mid 1,0,0)\) & \((2,1,2\mid 1,1,2)\) & No & \(\text{--}\) \\
\((2,0,1\mid 0,0,0)\) & \((2,1,0\mid 0,1,2)\) & Yes & \(-1\) \\
\((2,1,0\mid 0,0,0)\) & \((2,2,2\mid 0,1,2)\) & No & \(\text{--}\) \\
\end{longtable}
\endgroup
The six valid words have one block with a transposition and one
with the identity, so all six signs are \(-1\). There are
\(24\) invalid words. Restricting swaps to a zero and a nonzero
letter gives eight swaps from each valid word and exactly two
repairs of each invalid word. Thus the switching count is
\begin{equation}
 8V=2I=48,\qquad I=4V,\qquad
 \frac{V}{V+I}=\frac6{30}=\frac15=\frac1{n-1}.
\end{equation}
\end{example}

\Needspace{7\baselineskip}
\subsection{Small graphs through their shifts}\label{app:full-graphs}
Figure~\ref{fig:full-small-graphs} constructs two graphs from their
shifts, as in Section~\ref{sec:graph}. Words are written without commas.
Each listed \(S_{q,n}\) contains six shifts, giving every word \(u\)
the six neighbors \(u+s\)
(Lemma~\ref{lem:graph-properties}\,\ref{item:graph-degree}).
Bold arrows apply one fixed shift at every vertex. Both graphs are
undirected by part~\ref{item:graph-simple} of the same lemma.

For \((q,n)=(2,4)\), the allowed shifts have exactly two ones.
The highlighted \(0011\) flips the last two bits, pairing words as
in \(0000\leftrightarrow0011\). Every allowed shift preserves parity.
Distinct words in a parity class differ in two positions, unless they
are complementary and differ in four. Thus only complementary pairs
are nonadjacent; these pairs form the independent parts of the two
\(K_{2,2,2,2}\) components.

\begin{figure}[htbp]
\centering
\includegraphics[width=.90\linewidth]{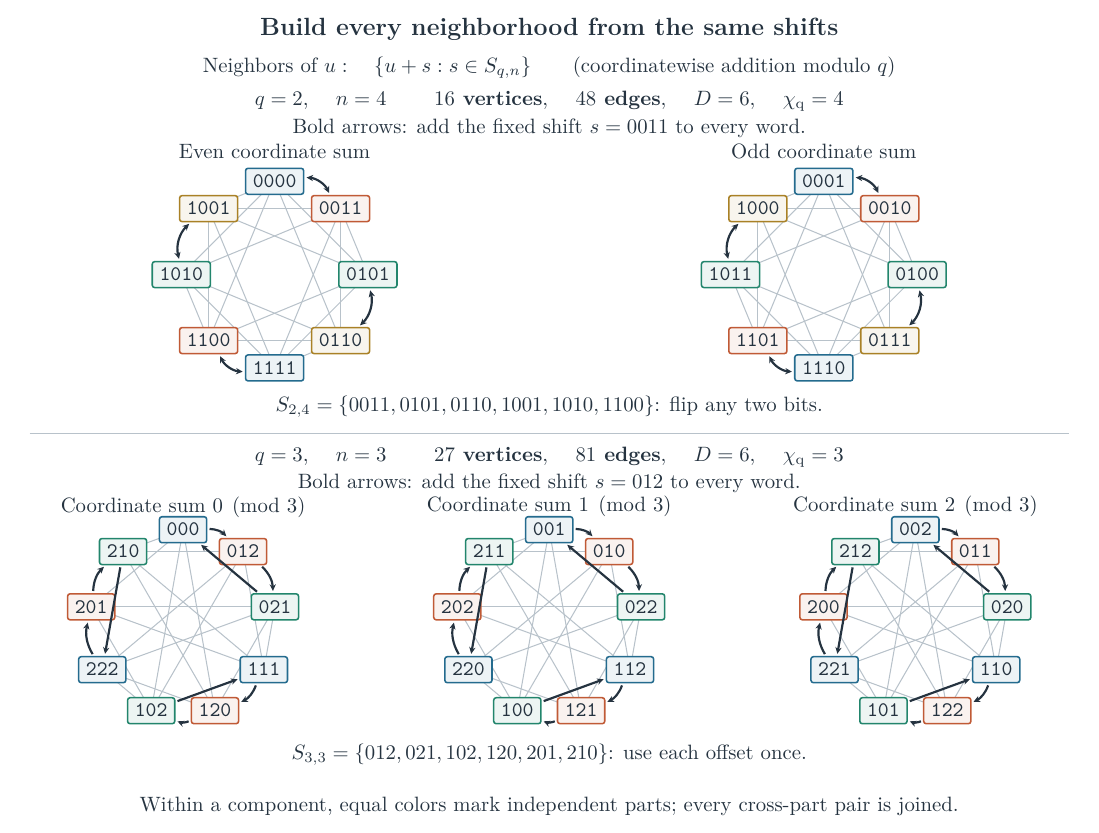}
\caption{Allowed shifts generate all edges. Bold arrows add \(0011\)
modulo two (top) and \(012\) modulo three (bottom); other edges are
gray. Colors mark independent parts. Both undirected graphs
(Lemma~\ref{lem:graph-properties}\,\ref{item:graph-simple}) satisfy
\(\chi=\chq=n\).}
\label{fig:full-small-graphs}
\end{figure}

For \((q,n)=(3,3)\), the shifts are the six permutations of \(012\).
Iterating the highlighted \(012\) gives cycles such as
\(000\to012\to021\to000\); its inverse \(021\) reverses them.
All shifts have coordinate sum zero modulo three, so edges remain
within coordinate-sum classes. A sum-zero difference is constant or
a permutation of \(012\); only the latter gives an edge.
Constant differences therefore define the independent triples in the
three \(K_{3,3,3}\) components.

For a Fourier wave \(f\), neighbor summation adds the six shifted
functions \(u\mapsto f(u+s)\), one for each allowed shift \(s\).
Each is a phase multiple of \(f\)
(Section~\ref{sec:fourier-waves}).

In each graph, one vertex from each part of a component gives an
\(n\)-clique. Reusing the same \(n\) colors for the parts in every
component gives an \(n\)-coloring. Alice sends her word's color to
distinguish Bob's candidates, so classical codes attain the
\(n\)-message optimum in both examples.

\begin{thebibliography}{10}
\raggedright
\small
\setlength{\itemsep}{2pt}
\setlength{\parsep}{0pt}
\setlength{\parskip}{0pt}
\bibitem{clmw}
T.~S.~Cubitt, D.~Leung, W.~Matthews, and A.~Winter,
\emph{Improving zero-error classical communication with entanglement},
Physical Review Letters \textbf{104} (2010), 230503.
\href{https://doi.org/10.1103/PhysRevLett.104.230503}{doi:10.1103/PhysRevLett.104.230503}.
\bibitem{lmmor}
D.~Leung, L.~Man\v{c}inska, W.~Matthews, M.~Ozols, and A.~Roy,
\emph{Entanglement can increase asymptotic rates of zero-error classical
communication over classical channels},
Communications in Mathematical Physics \textbf{311} (2012), 97--111.
\href{https://doi.org/10.1007/s00220-012-1451-x}{doi:10.1007/s00220-012-1451-x}.
\bibitem{bblps}
J.~Bri\"et, H.~Buhrman, M.~Laurent, T.~Piovesan, and G.~Scarpa,
\emph{Entanglement-Assisted Zero-Error Source-Channel Coding},
IEEE Transactions on Information Theory \textbf{61}(2) (2015), 1124--1138.
\href{https://doi.org/10.1109/TIT.2014.2385080}{doi:10.1109/TIT.2014.2385080}.
\bibitem{beigi}
S.~Beigi,
\emph{Entanglement-assisted zero-error capacity is upper bounded by the
Lov\'asz theta function},
Physical Review A \textbf{82} (2010), 010303(R).
\href{https://doi.org/10.1103/PhysRevA.82.010303}{doi:10.1103/PhysRevA.82.010303}.
\bibitem{cao}
X.~Cao, K.~Feng, H.~Huang, Y.~Yang, and Z.~Zhang,
\emph{On the quantum chromatic number of Hamming and generalized
Hadamard graphs}, arXiv:2510.14209v2, 11 March 2026.
\href{https://arxiv.org/abs/2510.14209v2}{arXiv:2510.14209v2}.
\bibitem{lnz}
T.~Luo, Y.~Ning, and X.~Zhang,
\emph{Quantum Chromatic Number of Subgraphs of Orthogonality Graphs
and the Distance-2 Hamming Graph},
Electronic Journal of Combinatorics \textbf{33}(3) (2026), P3.45.
\href{https://doi.org/10.37236/14936}{doi:10.37236/14936}.
\bibitem{nkz}
Y.~Ning, J.~H.~Koolen, and X.~Zhang,
\emph{On the Smallest Eigenvalues and Quantum Chromatic Numbers of
Hamming Graphs and Generalizations}, arXiv:2605.28402v1, 27 May 2026.
\href{https://arxiv.org/abs/2605.28402v1}{arXiv:2605.28402v1}.
\bibitem{cmnsw}
P.~J.~Cameron, A.~Montanaro, M.~W.~Newman, S.~Severini, and A.~Winter,
\emph{On the quantum chromatic number of a graph},
Electronic Journal of Combinatorics \textbf{14} (2007), R81.
\href{https://doi.org/10.37236/999}{doi:10.37236/999}.
\bibitem{ew}
C.~Elphick and P.~Wocjan,
\emph{Spectral lower bounds for the quantum chromatic number of a graph},
Journal of Combinatorial Theory, Series A \textbf{168} (2019), 338--347.
\href{https://doi.org/10.1016/j.jcta.2019.06.008}{doi:10.1016/j.jcta.2019.06.008}.
\bibitem{wed}
P.~Wocjan, C.~Elphick, and P.~Darbari,
\emph{Spectral Lower Bounds for the Quantum Chromatic Number of a
Graph---Part II},
Electronic Journal of Combinatorics \textbf{27}(4) (2020), P4.47.
\href{https://doi.org/10.37236/9295}{doi:10.37236/9295}.
\bibitem{efo}
R.~Egan, D.~Flannery, and P.~\'O~Cath\'ain,
\emph{Classifying Cocyclic Butson Hadamard Matrices},
in C.~J.~Colbourn (ed.), \emph{Algebraic Design Theory and Hadamard
Matrices}, Springer Proceedings in Mathematics \& Statistics
\textbf{133}, Springer, Cham (2015), 93--106.
\href{https://doi.org/10.1007/978-3-319-17729-8_8}{doi:10.1007/978-3-319-17729-8\_8}.
\bibitem{djw}
D.~Donovan, K.~Johnson, and I.~M.~Wanless,
\emph{Permanents and Determinants of Latin Squares},
Journal of Combinatorial Designs \textbf{24}(3) (2016), 132--148.
\href{https://doi.org/10.1002/jcd.21418}{doi:10.1002/jcd.21418}.
\bibitem{mmw}
B.~D.~McKay, J.~C.~McLeod, and I.~M.~Wanless,
\emph{The number of transversals in a Latin square},
Designs, Codes and Cryptography \textbf{40} (2006), 269--284.
\href{https://doi.org/10.1007/s10623-006-0012-8}{doi:10.1007/s10623-006-0012-8}.
\bibitem{gqz}
J.~Gruska, D.~Qiu, and S.~Zheng,
\emph{Generalizations of the distributed Deutsch--Jozsa promise problem},
Mathematical Structures in Computer Science \textbf{27}(3) (2017), 311--331.
\href{https://doi.org/10.1017/S0960129515000158}{doi:10.1017/S0960129515000158}.
\bibitem{bcp}
C.~Bracken, Y.~M.~Chee, and P.~Purkayastha,
\emph{Optimal family of \(q\)-ary codes obtained from a substructure
of generalised Hadamard matrices},
Proceedings of the IEEE International Symposium on Information Theory
(ISIT), Boston, MA (2012), 116--119.
\href{https://ppurka.github.io/files/BCP-hadamard-ISIT2012.pdf}{Author's manuscript}.
\bibitem{bdov}
C.~Bachoc, E.~DeCorte, F.~M.~de Oliveira Filho, and F.~Vallentin,
\emph{Spectral bounds for the independence ratio and the chromatic
number of an operator},
Israel Journal of Mathematics \textbf{202}(1) (2014), 227--254.
\href{https://doi.org/10.1007/s11856-014-1070-7}{doi:10.1007/s11856-014-1070-7}.
\bibitem{wits}
H.~S.~Witsenhausen,
\emph{The zero-error side information problem and chromatic numbers},
IEEE Transactions on Information Theory \textbf{22}(5) (1976), 592--593.
\href{https://doi.org/10.1109/TIT.1976.1055607}{doi:10.1109/TIT.1976.1055607}.
\bibitem{cmrssw}
T.~Cubitt, L.~Man\v{c}inska, D.~Roberson, S.~Severini, D.~Stahlke,
and A.~Winter,
\emph{Bounds on Entanglement Assisted Source-Channel Coding via the
Lov\'asz \(\vartheta\) Number and Its Variants},
IEEE Transactions on Information Theory \textbf{60}(11) (2014), 7330--7344.
\href{https://doi.org/10.1109/TIT.2014.2349502}{doi:10.1109/TIT.2014.2349502}.
\bibitem{mr}
L.~Man\v{c}inska and D.~E.~Roberson,
\emph{Quantum homomorphisms},
Journal of Combinatorial Theory, Series B \textbf{118} (2016), 228--267.
\href{https://doi.org/10.1016/j.jctb.2015.12.009}{doi:10.1016/j.jctb.2015.12.009}.
\bibitem{zkfb}
J.~A.~Zeiss, G.~Ko\ss mann, O.~Fawzi, and M.~Berta,
\emph{Approximating fixed size quantum correlations in polynomial time},
arXiv:2507.12302v2, 5 August 2026.
\href{https://arxiv.org/abs/2507.12302v2}{arXiv:2507.12302v2}.
\bibitem{harrow-sym}
A.~W.~Harrow,
\emph{The Church of the Symmetric Subspace},
arXiv:1308.6595v1, 29 August 2013.
\href{https://arxiv.org/abs/1308.6595v1}{arXiv:1308.6595v1}.
\bibitem{ckmr}
M.~Christandl, R.~K\"onig, G.~Mitchison, and R.~Renner,
\emph{One-and-a-half quantum de Finetti theorems},
Communications in Mathematical Physics \textbf{273}(2) (2007), 473--498.
\href{https://doi.org/10.1007/s00220-007-0189-3}{doi:10.1007/s00220-007-0189-3}.
\href{https://arxiv.org/abs/quant-ph/0602130v4}{arXiv:quant-ph/0602130v4}.
\bibitem{zksp}
J.~A.~Zeiss, G.~Ko\ss mann, R.~Schwonnek, and M.~Pl\'avala,
\emph{Finite de Finetti for convex bodies and Polynomial Optimization},
arXiv:2601.15184v1, 21 January 2026.
\href{https://arxiv.org/abs/2601.15184v1}{arXiv:2601.15184v1}.
\bibitem{bh}
A.~E.~Brouwer and W.~H.~Haemers,
\emph{Spectra of Graphs},
Universitext, Springer, New York, 2012.
\href{https://doi.org/10.1007/978-1-4614-1939-6}{doi:10.1007/978-1-4614-1939-6}.
\bibitem{serre}
J.-P.~Serre,
\emph{Linear Representations of Finite Groups},
translated by L.~L.~Scott,
Graduate Texts in Mathematics \textbf{42}, Springer, New York, 1977.
\href{https://doi.org/10.1007/978-1-4684-9458-7}{doi:10.1007/978-1-4684-9458-7}.
\bibitem{cs-steering}
D.~Cavalcanti and P.~Skrzypczyk,
\emph{Quantum steering: a review with focus on semidefinite programming},
Reports on Progress in Physics \textbf{80}(2) (2017), 024001.
\href{https://doi.org/10.1088/1361-6633/80/2/024001}{doi:10.1088/1361-6633/80/2/024001}.
\bibitem{fb}
M.~J.~Ferguson and D.~W.~Bailey,
\emph{Zero-error coding for correlated sources},
unpublished manuscript, 1975.
\bibitem{stanley}
R.~P.~Stanley,
\emph{Enumerative Combinatorics}, Volume~1, second edition,
Cambridge Studies in Advanced Mathematics \textbf{49},
Cambridge University Press, Cambridge, 2012.
\href{https://doi.org/10.1017/CBO9781139058520}{doi:10.1017/CBO9781139058520}.
\bibitem{gn}
C.~D.~Godsil and M.~W.~Newman,
\emph{Colouring an Orthogonality Graph},
arXiv:math/0509151v1, 7 September 2005.
\href{https://arxiv.org/abs/math/0509151v1}{arXiv:math/0509151v1}.
\bibitem{zeiss-lean}
J.~A.~Zeiss,
\emph{Lean verification for optimal entanglement-assisted source coding
under a balanced-difference promise},
Lean source code (2026).
\href{https://github.com/JuliusAZeiss/Lean-Verification-for-Optimal-entanglement-assisted-source-coding-under-a-balanced-di-erence-promise}{GitHub repository}.
\end{thebibliography}
\end{document}